\documentclass[10pt, conference]{IEEEtran}
\makeatletter
\def\ps@headings{%
\def\@oddhead{\mbox{}\scriptsize\rightmark \hfil \thepage}%
\def\@evenhead{\scriptsize\thepage \hfil \leftmark\mbox{}}%
\def\@oddfoot{}%
\def\@evenfoot{}}
\makeatother
\usepackage[long]{optional}

\usepackage{cite}
   \usepackage[dvips]{graphicx}

\usepackage[cmex10]{amsmath}
\usepackage{algorithmicx}
\usepackage[ruled]{algorithm}
\usepackage{array}
\usepackage[tight,footnotesize]{subfigure}

\usepackage{graphicx, amssymb,subfigure,float}
\usepackage{times}
\usepackage[section]{placeins}
\usepackage{comment}

\newtheorem{theorem}{Theorem}
\newtheorem{lemma}{Lemma}

\def\squareforqed{\hbox{\rlap{$\sqcap$}$\sqcup$}}
\def\qed{\ifmmode\squareforqed\else{\unskip\nobreak\hfil
\penalty50\hskip1em\null\nobreak\hfil\squareforqed
\parfillskip=0pt\finalhyphendemerits=0\endgraf}\fi}

\begin{document}

\title{Virtual Machine Trading in a Federation of Clouds: Individual Profit and Social Welfare Maximization\vspace{-4mm}}

\author{\IEEEauthorblockN{
Hongxing Li\IEEEauthorrefmark{1},
Chuan Wu\IEEEauthorrefmark{1}, Zongpeng Li\IEEEauthorrefmark{2} and
Francis C.M.
Lau\IEEEauthorrefmark{1}}\\ \vspace{-3.5mm}
\IEEEauthorblockA{\IEEEauthorrefmark{1}Department of Computer
Science, The University of Hong
Kong, Hong Kong, 
Email: \{hxli, cwu, fcmlau\}@cs.hku.hk}
\IEEEauthorblockA{\IEEEauthorrefmark{2}Department of
Computer Science, University of Calgary, Canada, 
Email: zongpeng@ucalgary.ca}\vspace{-10mm}}

\maketitle

\begin{abstract}
By sharing resources among different cloud providers, the paradigm of federated clouds exploits temporal availability of resources and geographical diversity of operational costs for efficient job service. While interoperability issues across different cloud platforms in a cloud federation have been extensively studied, fundamental questions on cloud economics remain: When and how should a cloud trade resources ({\em e.g.}, virtual machines) with others, such that its net profit is maximized over the long run, while a close-to-optimal social welfare in the entire federation can also be guaranteed? To answer this question, a number of important, inter-related decisions, including job scheduling, server provisioning and resource pricing, should be dynamically and jointly made, while the long-term profit optimality is pursued. In this work, we design efficient algorithms for inter-cloud virtual machine (VM) trading and scheduling in a cloud federation. For VM transactions among clouds, we design a double-auction based mechanism that is strategyproof, individual rational, ex-post budget balanced, and efficient to execute over time. Closely combined with the auction mechanism is a dynamic VM trading and scheduling algorithm, which carefully decides the true valuations of VMs in the auction, optimally schedules stochastic job arrivals with different SLAs onto the VMs, and judiciously turns on and off servers based on the current electricity prices. Through rigorous analysis, we show that each individual cloud, by carrying out the dynamic algorithm in the online double auction, can achieve a time-averaged profit arbitrarily close to the offline optimum. Asymptotic optimality in social welfare is also achieved under homogeneous cloud settings. We carry out simulations to verify the effectiveness of our algorithms, and examine the achievable social welfare under heterogeneous cloud settings, as driven by the real-world Google cluster usage traces.

\end{abstract}

\section{Introduction}

The emerging federated cloud paradigm advocates sharing of disparate cloud services (in separate data centers) from different cloud providers, and interconnects them based on common standards and policies to provide a universal environment for cloud computing. Such a cloud federation exploits temporal and spatial availability of resources ({\em e.g.}, virtual machines) and diversity of operational costs ({\em e.g.}, electricity prices): when a cloud experiences a burst of incoming jobs, it may resort to VMs from other clouds with idle resources; when the electricity price for running servers and VMs is high at one cloud data center, the cloud can schedule jobs onto other cloud data centers with lower electricity charge at the moment. In this way, the aggregate job processing capacity of the cloud federation can be potentially higher than the aggregation of capacities of separate clouds operating alone, and the overall profit can be larger.


To implement the federated cloud paradigm, significant interest has arisen on developing interfaces and standards to enable cloud interoperability and job portability across different cloud platforms (\cite{IBM09}\cite{gcc09}). However, fundamental problems on cloud economics remain to be investigated. A cloud in the real world is selfish, and aims to maximize its own profit, {\em i.e.}, its income from handling jobs and leasing VMs to other clouds subtracting its operational costs and expenses in VM rental from other clouds. Only if its profit can be maximized and in any case not lower than when operating alone, can a cloud be incentivized to join a federation. This calls for an efficient mechanism to carry out resource trading and scheduling among federated clouds, to achieve profit maximization for individual clouds, as well as to perform well in social welfare. A number of inter-related, practical decisions are involved: (1) {\em VM pricing}: what mechanism should be advocated for VM sale and purchase among the clouds, and at what prices? (2) {\em Job scheduling}: with time-varying job arrivals at each cloud, targeting different resources and SLA requirements, should a cloud serve the jobs right away or later, to exploit time-varying electricity prices? And should a cloud serve a job using its own resources or others' resources? (3) {\em Server provisioning}: is it more beneficial for a cloud to keep many of its servers running to serve jobs of its own and from others, or to turn some of them down to save electricity? These decisions should be efficiently and optimally made in an online fashion, while guaranteeing long-term optimality of individual cloud's profits, as well as the social welfare.

In this paper, we 
 design efficient algorithms for inter-cloud resource trading and scheduling, in a federation consisting of disparate cloud data centers. A double-auction based mechanism is proposed for the sell and purchase of available VMs across cloud boundaries over time. The auction is strategy-proof, individual rational, ex-post budget balanced, and computationally efficient (polynomial time complexity). Closely combined with the auction mechanism is 
 an efficient, dynamic VM trading and scheduling algorithm, which carefully decides the true valuations of VMs to participate in the auction, optimally schedules randomly-arriving jobs with different resource requirements ({\em e.g.}, number of VMs) and SLAs (\emph{e.g.}, maximum job scheduling delay) onto different data centers, and judiciously turns on and off servers in the clouds based on the current electricity prices. The dynamic algorithm serves as an efficient strategy for each cloud to employ in the online double auction, and is proven to maximize individual profit for each cloud, over the long run of the system.
  The contributions of this work are summarized below.

{\em First}, among the first in the literature, we address selfishness of individual clouds in a cloud federation, and design efficient mechanisms to maximize the net profit of each cloud. This profit is not only guaranteed to be larger than that when the cloud operates alone, but also maximized over the long run, in the presence of time-varying job arrivals and electricity prices at the cloud.

{\em Second}, we novelly combine a truthful double auction mechanism with stochastic Lyapunov optimization techniques, and design an online VM trading and scheduling algorithm, for a cloud to optimally price the VMs and to judiciously schedule the VM and server usages. Each cloud values different VMs based on the back pressure in job queue scheduling, and bids them in the auction for effective VM acquisition.



{\em Third}, we demonstrate that by applying the dynamic algorithm in the online double auction, each cloud can achieve a time-averaged profit arbitrarily close to its offline optimum (obtained if the cloud knows complete information on incoming jobs and electricity prices in the entire time span).  We also prove that the social welfare, \emph{i.e.}, the time-averaged overall profit in the federation, can be asymptotically maximized when the number of clouds grows, under homogenous cloud settings. Trace-driven simulations examine the achievable social welfare with our dynamic algorithm under heterogenous settings. 

In the rest of the paper, we discuss related literature in Sec.~\ref{sec:relatedwork}, present the system model in Sec.~\ref{sec:problemmodel}, and introduce the detailed resource trading and scheduling mechanisms in Sec.~\ref{sec:algorithm}. A double auction mechanism is proposed in Sec.~\ref{sec:auction_mechanism}, and a benchmark social-welfare maximization algorithm is discussed in Sec.~\ref{sec:welfare}. Theoretical analysis and simulation studies are presented in Sec.~\ref{sec:analysis} and Sec.~\ref{sec:simulation}, respectively. Sec.~\ref{sec:conclusion} concludes the paper.

\section{Related work}\label{sec:relatedwork}

\subsection{Optimal Scheduling in Cloud Systems}

Most existing literature (\hspace{-0.1mm}\cite{rao-infocom10,ren-icdcs12,noms10,yao-infocom12} and references therein) on resource scheduling in cloud systems focus on a single cloud that operates alone. A common theme is to minimize the operational costs (mainly consisting of electricity bills) in one or multiple data centers of the cloud, while providing certain performance guarantee of job scheduling, \emph{e.g.}, in terms of average job completion times \cite{rao-infocom10,ren-icdcs12,noms10,yao-infocom12}.

Urgaonkar {\em et al.}~\cite{noms10} propose an algorithm with joint job admission control, routing and resource allocation for power consumption reduction in a virtualized data center. Rao {\em et al.}~\cite{rao-infocom10} advocate minimization of electricity expenses 
 by exploiting the temporal and spatial diversities of electricity prices. Yao {\em et al.}~\cite{yao-infocom12} minimize the power cost with a two-time scale algorithm for delay tolerant workloads. Ren {\em et al.}~\cite{ren-icdcs12} also aim to minimize the energy cost while addressing the fairness in resource allocation. All the above works provide {\em average} delay guarantees for job services.


Different from these studies on a stand-alone cloud with centralized control, this work investigates profit maximization for individual selfish clouds in a federation, where each participant makes its own decisions. Besides, bounded scheduling delay for each job is guaranteed even in worst cases, contrasting the existing solutions that ensure average delays.

\subsection{Resource Trading Mechanisms}

A rich body of literature is devoted to resource trading in grid computing \cite{Buyya-hpc00} and wireless spectrum leasing \cite{zhou-infocom09}\cite{xu-infocom10}. Various mechanisms have been studied, \emph{e.g.}, bargaining \cite{Buyya-hpc00}, fixed or dynamic pricing based on a contract or the supply-demand ratio \cite{amazon}, and auctions \cite{zhou-infocom09}\cite{xu-infocom10}.

A bargaining mechanism \cite{Buyya-hpc00} typically has an unacceptable complexity by negotiating between each pair of traders. Fixed pricing, \emph{e.g.}, Amazon EC2 on-demand instances, has been shown to be inefficient in social welfare maximization in cases of system dynamics \cite{Marian-CCGrid10}. Dynamic pricing, such as Amazon EC2 spot instances, could be inefficient too, where the participants can quote the resources untruthfully \cite{Marian-CCGrid12}.

Auction stands out as a promising mechanism, on which there have been abundant solutions (\hspace{-0.1mm}\cite{zhou-infocom09,xu-infocom10} and references therein) with truthful design and polynomial complexity. Although some recent works \cite{Marian-CCGrid10,Marian-CCGrid12,gomes-exchange12} aim to design an auction mechanism with individual rationality (non-negative profit gain) for trading in federated clouds, they do not explicitly address individual profit maximization over the long run, nor other desirable properties such as truthfulness, ex-post budget balance, and social welfare maximization. Moreover, little literature on auctions provides methods to quantitatively calculate the true valuations in each bid, which are simply assumed as known. Our design addresses these issues.

%
%


\section{System Model and Auction Framework}\label{sec:problemmodel}

\subsection{Federation of Clouds}

We consider a federation of $F$ clouds, each located at a different geometric location and operates autonomously to gain profit by serving its customers' job requests, managing server provisioning and trading resources with other clouds.

\vspace{1mm}
\noindent \textbf{Service demands}: Each individual cloud $i\in [1, F]$ has a front-end proxy server, which accepts job requests from its customers. There are $S$ types of jobs serviced at each cloud, each specified by a three-tuple $<m_s, g_s, d_s>$. Here, $m_s\in [1,M]$ specifies the type of the required VM instances, where $M$ is the maximum number of VM types, and each type corresponds to a different set of configurations of CPU, storage and memory; $g_s$ is 
 the number of type-$m_s$ VMs that the job needs simultaneously (See Amazon EC2 API \cite{amazon}); and $d_s$ stands for the SLA (Service Level Agreement) of job type $s\in [1,S]$, evaluated by the maximal response delay for scheduling a job, \emph{i.e.}, the time-span from when the job arrives to when it starts to run on scheduled VMs. In a cloud in practice, it is common to buy servers of the same configuration and provision the same type of VMs on one machine \cite{linode}. Therefore, we suppose each cloud $i$ has $N_{i}^m$ homogenous servers to provision VMs of type $m\in [1, M]$, each of which can provide a maximum of $C_{i}^m$ VMs of this type; the total number of servers in cloud $i$ is $\sum_{m=1}^M N_{i}^m$.

The system runs in a time-slotted fashion. At the beginning of each time slot $t$, $r_i^s(t)\in [0, R_i^s]$ jobs arrive at cloud $i$, for each job type $s$. $R_i^s$ is an upper-bound on the number of type-$s$ jobs submitted to cloud $i$ in a time slot. The arrival of jobs is an ergodic process at each cloud. We suppose the arrival rate is given, and how a customer decides which cloud to use is orthogonal to this study. Let $p_i^{s}(t)\in [0, p_i^{s(max)}]$ be the given service charge to the customer by cloud $i$, for accepting a job of type $s$ in time slot $t$, which remains fixed within a time slot, but may vary across time slots.  Here, $p_i^{s(max)}$ is the max possible price for $p_i^s(t)$. Such a general charging model subsumes pricing schemes in practice: {\em e.g.}, time-independent $p_i^s(t)$ corresponds to the \emph{on-demand} VM charging scheme, while time-varying $p_i^s(t)$ can represent the \emph{spot instance} prices based on the current demand {\em vs}. supply \cite{amazon}.

\vspace{1mm}
\noindent \textbf{Job scheduling}: Each incoming job to cloud $i$ enters a FIFO queue of its type --- a cloud $i$ maintains a queue to buffer unscheduled jobs of each type $s$, with $Q_i^s(t)$ as its length in $t$. When the required VMs of a job are allocated, the job departs from its queue and starts to run on the VMs. A cloud may schedule its jobs on either its own VMs or VMs leased from other clouds, for the best economic benefits. Let $\mu_{ij}^{s}(t)$ be the number of type-$s$ jobs of cloud $i$ that are scheduled for processing in cloud $j$ at the beginning of time slot $t$. 

When a job's demanded maximum response time (the SLA) cannot be met, in cases of system overload, it is dropped. A penalty is enforced in this case, to compensate for the customer's loss. Let

\vspace{-5mm}{\small
\begin{align}
D_i^{s}(t)\in [0, D_i^{s(max)}]\label{eqn:drop-cons}
\end{align}}\vspace{-6mm}

\noindent be the number of type-$s$ jobs dropped by cloud $i$ in $t$, where $D_i^{s(max)}$ is the maximum value of $D_i^s(t)$. Let $\xi_i^s$ be the penalty to drop one such job, which is at least the maximum price charged to customers when accepting the jobs, {\em i.e.}, $\xi_i^s\geq p_i^{s(max)}$.

Hence, the number of unscheduled jobs buffered at each cloud $i\in [1, F]$ can be updated with the following queueing law:

\vspace{-6mm}{\small\begin{align}
Q_i^{s}(t+1) = & \max\{Q_i^{s}(t)-\sum_{j=1}^F \mu_{ij}^{s}(t)- D_i^{s}(t), 0\}\nonumber \\
                &+ r_i^{s}(t),~~~~\forall s\in [1, S].\label{eqn:queue1}
\end{align}}\vspace{-6mm}

\noindent Job scheduling should satisfy the following SLA constraint:

\vspace{-4mm}{\small \begin{align}
&\text{Each type-$s$ job in cloud $i$ is either scheduled or dropped (subject}\nonumber\\ &\text{to a penalty) before its maximum response delay $d_s$,}\ \forall s\in [1, S].\label{eqn:SLA}
\end{align}}\vspace{-5mm}


We apply the $\epsilon-$persistence queue technique \cite{neely-infocom11}, to create a virtual queue $Z_i^{s}$ associated with each job queue $Q_i^s$ ($\forall i\in [1, F]$):  

\vspace{-5mm}{\small\begin{align}
Z_i^{s}(t+1) =&  \max\{Z_i^{s}(t) + \mathbf{1}_{\{Q_i^{s}(t)> 0\}}\cdot [\epsilon_{s} - \sum_{j=1}^F\mu_{ij}^{s}(t)]- D_i^{s}(t)\nonumber\\
              &
- \mathbf{1}_{\{Q_i^{s}(t)=0\}}\cdot \sum_{j=1}^F \frac{C_{j}^{m_s}\cdot N_{j}^{m_s}}{g_s}, 0\},
\forall s\in [1, S].\label{eqn:queue2}
\end{align}}\vspace{-4mm}

\noindent Here, $\epsilon_s>0$ is a constant. $\mathbf{1}_{\{Q_i^{s}(t)> 0\}}$ and $\mathbf{1}_{\{Q_i^{s}(t)=0\}}$ are indicator functions such that

\vspace{-4mm}{\small
\begin{align*}
\mathbf{1}_{\{Q_i^{s}(t)> 0\}}=\begin{cases}1 & \text{if }Q_i^{s}(t)> 0\\ 0 & \text{Otherwise}\end{cases};\ \mathbf{1}_{\{Q_i^{s}(t)=0\}}= \begin{cases}1 & \text{if }Q_i^{s}(t)= 0\\ 0 & \text{Otherwise}\end{cases}.
\end{align*}}\vspace{-4mm}


\noindent 
Length of this virtual queue reflects the cumulated response delay of jobs from the respective job queue. Our algorithm seeks to bound the lengths of job queues and virtual queues, with properly set $\epsilon_s$, and hence the maximum response delay of jobs can be bounded, {\em i.e.}, constraint (\ref{eqn:SLA}) is satisfied.



\vspace{1mm}
\noindent \textbf{Server provisioning}: We consider electricity cost, for running and cooling the servers \cite{energy09}, as the main component of the operational cost in a cloud. Other costs, \emph{e.g.}, space rental and labour, remain relatively fixed for a long time, and are of less interest. Given that electricity prices vary at different locations and from time to time \cite{rao-infocom10}\cite{ferc}, we model the operational cost $\beta_{i}(t)$ in each cloud $i$ as a general ergodic process over time, varying across time slots between $\beta_i^{(min)}$ and $\beta_{i}^{(max)}$.

Each cloud strategically decides the number of active servers at each time, to optimize its profit. Let $n_i^m(t)$ be the number of active servers provisioning type-$m$ VMs at cloud $i$ in $t$. The available server capacities at each cloud $i\in [1, F]$ constrain the feasible job scheduling at time $t$:

\vspace{-4mm}{\small\begin{align}
&\sum_{j\in [1,F]}\sum_{s:m_s = m, s\in [1, S]}g_s \mu_{ji}^s(t)\leq C_i^m \cdot n_i^m(t),\ \forall m\in [1, M], \label{eqn:capacity1}\\
&n_i^m(t)\leq N_i^m,\ \forall m\in [1, M].\label{eqn:capacity2}
\end{align}}\vspace{-6mm}

\noindent (\ref{eqn:capacity1}) states that the overall demand for type-$m$ VMs in cloud $i$ from itself and other clouds should be no larger than the maximum number of available type-$m$ VMs on the active servers in cloud $i$. Here $g_s\mu_{ji}^s(t)$ is the total number of VMs needed by type-$s$ jobs scheduled from cloud $j$ to cloud $i$ in $t$. Motivated by practical job execution efficiency, we only consider scheduling a job to VMs from a single cloud, but not VMs across different clouds.  (\ref{eqn:capacity2}) ensures that the number of active servers is limited by the total number of on-premise servers of the corresponding VM configuration at each cloud.

\subsection{Inter-cloud VM Trading with Double Auction}\label{subsec:auction}

In an inter-cloud resource market, VMs constitute the items for trading. For each type of VMs, multiple clouds may have them on sale while multiple other clouds can request them. A double auction is a natural fit to implement efficient trading in this case, allowing both selling and buying clouds to actively participate in pricing, on behalf of their own benefits. In our dynamic system, a {\em multi-unit double auction} is carried out among the clouds at the beginning of each time slot, deciding the VM trades within that time slot.

\noindent \textbf{Buyers \& Sellers}: A cloud can be both a buyer and a seller. A buy-bid $<b_{i}^m(t), \gamma_{i}^m(t)>$ records the unit price and maximum quantity at which cloud $i$ is willing to buy VMs of type $m$, in $t$. Similarly, a sell-bid $<s_{i}^m(t), \eta_{i}^m(t)>$ records the unit price and maximum quantity at which cloud $i$ is willing to sell VMs of type $m$ in $t$.

Let $\tilde{b}_i^m(t)$ and $\tilde{s}_{i}^m(t)$ be cloud $i$'s true valuation of buying and selling a type-$m$ VM respectively (the max/min price it is willing to pay/accept). Similarly, let $\tilde{\gamma}_i^m(t)$ and $\tilde{\eta}_{i}^m(t)$ be cloud $i$'s true valuation of the quantity to buy and sell VMs of type $m$ respectively (the maximum volume of VMs it is willing to purchase/sell). A cloud $i$ may strategically manipulate the bid prices and volumes, in the hope of maximizing its profit. We show in Sec.~\ref{sec:analysis} that the double auction proposed in Sec.~\ref{sec:algorithm} is truthful, such that each bid price reveals the true valuation.

\noindent \textbf{Auctioneer}: We assume that there is a broker in the cloud federation, assuming the role of the auctioneer. After collecting all the buy and sell bids, the auctioneer executes a double auction to be detailed in Sec.~\ref{sec:auction_mechanism}, to decide the set of successful buy and sell bids, their clearing prices and the numbers of VMs to trade in each type. Let $\hat{b}_i^m(t)$ be the actual charge price for cloud $i$ to buy one type-$m$ VM, and $\hat{\gamma}_{i}^m(t)$ be the actual number of VMs purchased. Similarly, let $\hat{s}_{i}^m(t)$ be the actual income cloud $i$ receives for selling one type-$m$ VM, and $\hat{\eta}_{i}^m(t)$ be the actual number of  VMs sold. 

Let $\alpha_{ij}^m(t)$ be the number of type-$m$ VMs that cloud $i\in [1,F]$ purchases from cloud $j\in [1,F]$ in $t$, as decided by the auctioneer: 

\vspace{-4mm}{\small\begin{align}
\hat{\gamma}_{i}^m(t)& = \sum_{j\in [1,F], j\neq i}\alpha_{ij}^m(t),\ \forall m\in [1, M],\label{eqn:buy-cons}\\
\hat{\eta}_{i}^m(t)& = \sum_{j\in [1, F], j\neq i}\alpha_{ji}^m(t),\ \forall m\in [1, M]. \label{eqn:sell-cons}
\end{align}}\vspace{-4mm}

Since VMs are purchased for serving jobs, the job scheduling decisions $\mu_{ij}^s(t)$ at each cloud $i\in [1,F]$, are related to the number of VMs it purchases:

\vspace{-4mm}{\small\begin{align}
\sum_{s:s\in [1, S], m_s=m}g_s\cdot& \mu_{ij}^s(t)= \alpha_{ij}^m(t),\nonumber\\
& \forall m\in [1, M], \forall i,j\in [1, F], i\neq j.\label{eqn:capacity3}
\end{align}}\vspace{-5mm}

Three economic properties are desirable for the auctioneer's mechanism.
(i) \emph{Truthfulness}: Bidding true valuations is a dominant strategy, and consequently, both bidder strategies and auction design are simplified. (ii) \emph{Individual Rationality}: Each cloud obtains a non-negative profit by participating in the auction. (iii) \emph{Ex-post Budget Balance}: The auctioneer has a non-negative surplus, {\em i.e.}, the total payment from all winning buy-bids is no less than the total charge for all winning sell-bids in each time slot.

%
%
%

\begin{table}[!t]
\centering
\caption{Notation: input quantities and intermediate variables}\label{table:notation1}\vspace{-3mm}
\begin{tabular}{|p{0.6cm}|p{3cm}||p{0.6cm}|p{3cm}|}
\hline
  $F$ & \# of clouds &  $S$ & \# of service types \\\hline
  $M$ & \# of VM types & $m_s$  & VM type of service type $s$ \\\hline
  $d_s$ & Max.~response delay of service type $s$ & $g_s$  & \# of VMs required by service type $s$ \\\hline
\end{tabular}
\begin{tabular}{|p{0.90cm}|p{7.25cm}|}
    $r_i^s(t)$ & \# of type-$s$ jobs arrived at cloud $i$, slot $t$ \\\hline
    $R_i^s$ & Max.~\# of type-$s$ jobs arrived at cloud $i$ per slot \\\hline
    $p_i^{s}(t)$ & Service price for each job of type $s$ at cloud $i$, slot $t$    \\\hline
    $p_i^{s(max)}$ & Max.~service price for each type-$s$ job at cloud $i$ per slot\\\hline
    $\beta_{i}(t)$ & Cost for operating an active server at cloud $i$, slot $t$ \\\hline
    $\beta_{i}^{(min)}$ & Min.~cost for operating an active server at cloud $i$ per slot \\\hline
    $\beta_{i}^{(max)}$ & Max.~cost for operating an active server at cloud $i$ per slot \\\hline
    $\xi_i^s$ & Penalty for dropping a type-$s$ job at cloud $i$\\\hline
    $D_i^{s(max)}$ & Max.~\# of type-$s$ jobs cloud $i$ drops per slot \\\hline
    $C_{i}^m$ & Max.~\# of type-$m$ VMs an active server at cloud $i$ provisions \\\hline
    $N_{i}^m$ & Total \# of servers provisioning type-$m$ VMs at cloud $i$ \\\hline
    $Q_i^s(t)$ & Length of queue buffering type-$s$ jobs at cloud $i$, slot $t$ \\\hline
    $Z_i^s(t)$ & Length of virtual queue of type-$s$ jobs at cloud $i$, slot $t$ \\\hline
    $\epsilon_s$ & Constant positive parameter for $Z_i^s(t)$, $\forall i\in [1, F]$\\\hline
    $Q_i^{s(max)}$ & Maximum length of queue $Q_i^s(t)$ \\\hline
    $Z_i^{s(max)}$ & Maximum length of virtual queue $Z_i^s(t)$ \\\hline
    $V$ & User-defined constant positive parameter for dynamic algorithm\\\hline
\end{tabular}
\vspace{-6mm}
\end{table}

\begin{table}[!t]
\centering
\caption{Notation: decision variables at individual clouds}\label{table:notation2}\vspace{-3mm}
\begin{tabular}{|p{0.90cm}|p{7.25cm}|}
\hline
    $\mu_{ij}^s(t)$ & \# of type-$s$ jobs scheduled from cloud $i$ to cloud $j$, slot $t$\\\hline
    $n_{i}^m(t)$ & \# of active servers providing type-$m$ VMs at cloud $i$, slot $t$ \\\hline
    $D_i^s(t)$ & \# of dropped type-$s$ jobs at cloud $i$, slot $t$ \\\hline
    $\tilde{s}_{i}^m(t)$ & True value of selling one type-$m$ VM from cloud $i$, slot $t$    \\\hline
    $\tilde{\eta}_{i}^m(t)$ & True value of volume to sell type-$m$ VMs from cloud $i$, slot $t$  \\\hline
    $s_{i}^m(t)$ & Bid price for selling one type-$m$ VM from cloud $i$, slot $t$    \\\hline
    $\eta_{i}^m(t)$ & Max.~\# of type-$m$ VMs cloud $i$ can sell, slot $t$    \\\hline
    $\tilde{b}_{i}^m(t)$ & True value of buying one type-$m$ VM by cloud $i$, slot $t$    \\\hline
    $\tilde{\gamma}_{i}^m(t)$ & True value of volume to buy type-$m$ VMs by cloud $i$, slot $t$    \\\hline
    $b_{i}^m(t)$ & Bid price for buying one type-$m$ VM by cloud $i$, slot $t$    \\\hline
    $\gamma_{i}^m(t)$ & Max.~\# of type-$m$ VMs cloud $i$ can buy, slot $t$    \\\hline
\end{tabular}
 \vspace{-3mm}
\end{table}

\begin{table}[!t]
\centering
\caption{Notation: decision variables at the auctioneer}\label{table:notation3}\vspace{-3mm}
\begin{tabular}{|p{0.90cm}|p{7.25cm}|}
\hline
    $\hat{s}_{i}^m(t)$ & Actual price of selling one type-$m$ VM from cloud $i$, slot $t$    \\\hline
    $\hat{\eta}_{i}^m(t)$ & Actual \# of type-$m$ VMs sold from cloud $i$, slot $t$    \\\hline
    $\hat{b}_{i}^m(t)$ & Actual price of buying one type-$m$ VM by cloud $i$, slot $t$    \\\hline
    $\hat{\gamma}_{i}^m(t)$ & Actual \# of type-$m$ VMs bought by cloud $i$, slot $t$    \\\hline
    $\alpha_{ij}^m(t)$ & Actual \# of type-$m$ VMs sold from cloud $j$ to $i$, slot $t$ \\\hline
    $\theta_j^m(t)$ & The $j^{th}$ highest buy-bid price for type-$m$ VMs at auctioneer \\\hline
    $\vartheta_j^m(t)$ & The $j^{th}$ lowest sell-bid price for type-$m$ VMs at auctioneer \\\hline
    $L_j^m(t)$ & Max.~\# of type-$m$ VMs to sell, in sell-bid with $j^{th}$ lowest price at auctioneer in $t$\\\hline
\end{tabular}
\vspace{-6mm}
\end{table}

\subsection{Individual Selfishness}\label{subsec:selfishness}

Each cloud in the federation aims to maximize its time-averaged profit (revenue minus cost) over the long run of the system, while striking to fulfill the resource and SLA requirements of each job.

\vspace{1mm}
\noindent\textbf{Revenue}: A cloud has two sources of revenue: i) job service charges paid by its customers, and ii) the proceeds from VM sales. The time-averaged revenue of cloud $i\in [1,F]$ by undertaking different types of jobs from its customers is

\vspace{-4mm}{\small\begin{align}
\Phi_1^i &= \lim_{T\rightarrow \infty}\frac{1}{T}\sum_{t=0}^{T-1}\sum_{s\in [1, S]}\mathbb{E}\{p_i^s(t)\cdot r_i^s(t)\}.\label{eqn:revenue1}
\end{align}}\vspace{-4mm}

\noindent We assume the front-end charges, $p_i^s(t)$, from a cloud to its customers, are given. Hence, this part of the revenue is fixed in each time slot. The time-averaged income of cloud $i\in [1, F]$ from selling VMs to other clouds is:

\vspace{-4mm}{\small\begin{align}
\Phi_2^i &= \lim_{T\rightarrow \infty}\frac{1}{T}\sum_{t=0}^{T-1}\sum_{m\in [1, M]}\mathbb{E}\{\hat{s}_i^m(t)\cdot \hat{\eta}_i^m(t)\}.\label{eqn:revenue2}
\end{align}}\vspace{-4mm}

\noindent Cloud $i$ can control this income by adjusting its sell-bids, \emph{i.e.}, $s_i^m(t)$ and $\eta_i^m(t)$, $\forall m\in [1, M]$, at each time.

\vspace{1mm}
\noindent\textbf{Cost}: The cost of cloud $i$ consists of three parts: i) operational costs incurred for running its active servers, ii) the penalties for dropping jobs, and iii) the expenditure on buying VMs from other clouds. The time-averaged cost for operating servers at each cloud $i\in [1, F]$ is decided by the number of active servers in each time, \emph{i.e.},

\vspace{-4mm}{\small\begin{align}
\Psi_1^i &=  \lim_{T\rightarrow \infty}\frac{1}{T}\sum_{t=0}^{T-1}\mathbb{E}\{\beta_{i}(t)\cdot \sum_{m=1}^M n_{i}^m(t)\}.\label{eqn:cost1}
\end{align}}\vspace{-4mm}

The time-averaged penalty at each cloud $i\in [1, F]$ is determined by the number of dropped jobs over time, \emph{i.e.}, $D_i^s(t),\ \forall s\in [1, S]$, $t\in[0,T-1]$:

\vspace{-4mm}{\small\begin{align}
\Psi_2^i &= \lim_{T\rightarrow \infty}\frac{1}{T}\sum_{t=0}^{T-1}\sum_{s\in [1,S]}\mathbb{E}\{\xi_i^s\cdot D_i^{s}(t)\}.\label{eqn:cost2}
\end{align}}\vspace{-4mm}

The time-averaged expenditure for VM purchases is decided by the actual VM trading prices and numbers, as decided by the buy-bids ($b_i^m(t),\gamma_i^m(t))$ from cloud $i\in [1, F]$:

\vspace{-4mm}{\small\begin{align}
\Psi_3^i & = \lim_{T\rightarrow \infty}\frac{1}{T}\sum_{t=0}^{T-1}\mathbb{E}\{\sum_{m=1}^{M}\hat{b}_i^m(t)\cdot \hat{\gamma}_i^m(t)\}.\label{eqn:cost3}
\end{align}}\vspace{-5mm}

\vspace{1mm}
\noindent\textbf{Profit Maximization}: The profit maximization problem at cloud $i\in [1, F]$ can be formulated as follows:

\vspace{-5mm}{\small\begin{align}
\max&~~~~\Phi_1^i + \Phi_2^i - \Psi_1^i - \Psi_2^i - \Psi_3^i \label{eqn:profit-max}\\
s.t.&~~~~\text{Constraints (\ref{eqn:drop-cons})-(\ref{eqn:capacity3})}.\nonumber
\end{align}}\vspace{-5mm}



\vspace{-2mm}
\subsection{Social Welfare}\label{subsec:social}

Social welfare is the overall profit of the cloud federation: 

\vspace{-4mm}{\small\begin{align*}
\sum_{i\in [1, F]}(\Phi_1^i+\Phi_2^i-\Psi_1^i-\Psi_2^i-\Psi_3^i).
\end{align*}}\vspace{-4mm}

Since the income and expenditure due to VM trades among the clouds cancel each other, the formula above equals $\sum_{i\in [1, F]}(\Phi_1^i - \Psi_1^i - \Psi_2^i)$. The social welfare maximization problem is:

\vspace{-5mm}{\small\begin{align}
\max&~~~~\sum_{i\in [1, F]}(\Phi_1^i - \Psi_1^i - \Psi_2^i) \label{eqn:social-max}\\
s.t.&~~~~\text{Constraints (\ref{eqn:drop-cons})-(\ref{eqn:capacity2})},\ \forall i\in [1, F]\nonumber
\end{align}}\vspace{-6mm}


\noindent which globally optimizes server provisioning and job scheduling in the federation and maximally serves all the incoming jobs at the minimum cost, regardless of the specific inter-cloud VM trading mechanism.

When a double auction mechanism is truthful, individual rational and ex-post budget balancing, it is shown that efficiency in terms of social welfare maximization cannot be achieved concurrently \cite{myerson-1983}. We hence make a necessary compromise in social welfare in our auction design, {\em i.e.}, the sum of maximal individual profits derived by (\ref{eqn:profit-max}) will be smaller than the optimal social welfare from (\ref{eqn:social-max}). Nevertheless, we will show in Sec.~\ref{sec:analysis} and Sec.~\ref{sec:simulation} that our mechanisms still manages to achieve a satisfactory social welfare in the long run.

\vspace{1mm}
Tables \ref{table:notation1}, \ref{table:notation2} and \ref{table:notation3} summarize important notation in the paper, for ease of reference.

\section{Dynamic individual-profit maximization algorithm}\label{sec:algorithm}

We next present a dynamic algorithm for each cloud to trade VMs and scheduling jobs/servers, 
 which is in fact applicable under any truthful, individual-rational and ex-post budget balanced double auction mechanism. We will also tailor a double auction mechanism on the auctioneer in the next section. Fig.~\ref{fig:flow} illustrates the relation among these algorithm modules.

\vspace{-4mm}
\begin{figure}[H]
  \centering
  \includegraphics[width=0.6\columnwidth]{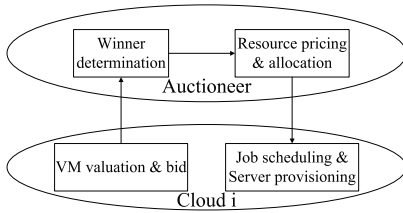}\\\vspace{-4mm}
  \caption{Key algorithm modules.}\label{fig:flow}\vspace{-4mm}
\end{figure}


The goal of the dynamic algorithm at each cloud $i$ is to maximize its time-averaged profit, {\em i.e.}, to solve optimization (\ref{eqn:profit-max}), by dynamically making decisions in each time slot. We apply the \emph{drift-plus-penalty} framework in Lyapunov optimization theory \cite{book2010}, and derive a one-shot optimization problem 
 to be solved by cloud $i$ in each time slot $t$ as follows. We will prove in Sec.~\ref{sec:analysis} that by optimally solving the one-shot optimization 
 at each cloud during each time slot, the dynamic algorithm can achieve a time-averaged individual profit arbitrarily close to its offline optimum (computed with complete knowledge in the entire time span), for each cloud.

\subsection{The One-shot Optimization Problem}

Define the set of queues at cloud $i$ in each time slot $t$ as

\vspace{-5mm}{\small
\begin{align*}
\Theta_i(t)= \{Q_i^s(t), Z_i^s(t)| s\in [1, S]\}.
\end{align*}}\vspace{-6mm}

Define the Lyapunov function as follows:

\vspace{-4mm}{\small
\begin{align*}
L(\Theta_i(t)) = \frac{1}{2}\sum_{s\in [1, S]}[(Q_i^s(t))^2 + (Z_i^s(t))^2].
\end{align*}}\vspace{-4mm}

Then the one-slot conditional Lyapunov drift \cite{book2010} is

\vspace{-4mm}{\small
\begin{align*}
\Delta(\Theta_i(t))= L(\Theta_i(t+1)) - L(\Theta_i(t)). 
\end{align*}}\vspace{-6mm}

Squaring the queuing laws (\ref{eqn:queue1}) and (\ref{eqn:queue2}), we can derive the following inequality (details can be found in Appendix \ref{appendix:drift}):

\vspace*{-5mm}{\small\begin{align}
&\Delta(\Theta_i(t))-V\cdot [\sum_{m\in [1, M]} [\hat{s}_i^m(t) \hat{\eta}_i^m(t) - \hat{b}_i^m(t) \hat{\gamma}_i^m(t)-\beta_i(t) n_i^m(t)]\notag \\ &+ \sum_{s\in [1, S]}[p_i^s(t)\cdot r_i^s(t)- D_i^s(t) \xi_i^s]] \notag\\
\leq&  B_i + \sum_{s\in [1, S]}[Q_i^s(t) r_i^s(t) + Z_i^s(t) \epsilon_s - V p_i^s(t)\cdot r_i^s(t)] \notag\\
&- \varphi_1^i(t)-\varphi_2^i(t)-\varphi_3^i(t),\label{eqn:drift-plus-penalty1}
\end{align}}\vspace*{-5mm}

\noindent where $V>0$ is a user-defined positive parameter for gauging the optimality of time-averaged profit, {\small $B_i=\frac{1}{2}\sum_{s\in [1, S]}[(\sum_{j=1}^F C_j^{m_s}N_j^{m_s}/g_s+D_i^{s(max)})^2 + (R_i^s)^2+(\epsilon_s)^2 + (D_i^{s(max)}+\sum_{j=1}^F C_j^{m_s}N_j^{m_s}/g_s)^2]$} is a constant,
and

\vspace{-4mm}
{\small \begin{align*}
&\varphi_1^i(t) = V \sum_{m\in [1, M]} [\hat{s}_i^m(t) \hat{\eta}_i^m(t) - \hat{b}_i^m(t) \hat{\gamma}_i^m(t)-\beta_i(t) n_i^m(t)],\\
&\varphi_2^i(t) = \sum_{s=\in [1, S]}\sum_{j\in [1, F]}\mu_{ij}^s(t) [Q_i^s(t) + Z_i^s(t)],\\
%
&\varphi_3^i(t) = \sum_{s\in [1, S]}D_i^s(t)[Q_i^s(t)+Z_i^s(t)-V\cdot \xi_i^s].
\end{align*}}\vspace{-4mm}

Based on the drift-plus-penalty framework \cite{book2010}, a dynamic algorithm can be derived for each cloud $i$, which observes the job and virtual queues ($\Theta_i(t)$), job arrival rates ($r_i^s(t),\ \forall s\in [1,S]$), the current cost for server operation ($\beta_i(t)$) in each time slot, and minimizes the RHS of the inequality (\ref{eqn:drift-plus-penalty1}), such that a lower bound for time-averaged profit of cloud $i$ 
is maximized. Note that $B_i + \sum_{s\in [1, S]}[Q_i^s(t) r_i^s(t) + Z_i^s(t) \epsilon_s - V p_i^s(t)\cdot r_i^s(t)]$ in the RHS of (\ref{eqn:drift-plus-penalty1})  is fixed in time slot $t$. Hence, to maximize a lower bound of the time-averaged profit for cloud $i$, the dynamic algorithm
should solve the one-shot optimization problem in each time slot $t$ as follows:

\vspace{-4mm}
{\small \begin{align}
\max&~~~~\varphi_1^i(t) + \varphi_2^i(t) + \varphi_3^i(t)\label{eqn:profit-oneslot} \\
s.t.&~~~~\text{Constraints (\ref{eqn:drop-cons}), (\ref{eqn:capacity1})-(\ref{eqn:capacity3}).}\nonumber
\end{align}}\vspace{-5mm}



The maximization problem in (\ref{eqn:profit-oneslot}) can be decoupled into two independent optimization problems:

\vspace{-4mm}{\small
\begin{align}
\max&~~\varphi_1^i(t)+\varphi_2^i(t)\label{eqn:profit-oneshot1}~~~~s.t.~~\text{Constraints (\ref{eqn:capacity1})-(\ref{eqn:capacity3})},
\end{align}}\vspace{-5mm}

\noindent which is related to optimal decisions on i) buy/sell bids for different types of VMs, and ii) scheduling of active servers and jobs to these servers; and

\vspace{-4mm}{\small
\begin{align}
\max&~~\varphi_3^i(t)\label{eqn:profit-oneshot2}~~~~s.t.~~\text{Constraint (\ref{eqn:drop-cons})},
\end{align}}\vspace{-5mm}

\noindent which is related to optimal decisions on iii) jobs to drop. In the following, we design algorithms to derive the optimal decisions based on problem (\ref{eqn:profit-oneshot1}) 
 and problem (\ref{eqn:profit-oneshot2}). 

\subsection{VM Valuation and Bid}

Optimization problem (\ref{eqn:profit-oneshot1}) is related to the actual charges that cloud $i$ pays for each type of VMs purchased, $\hat{b}_i^m(t)$ and $\hat{s}_i^m(t)$ ($\forall m\in [1, M]$), and the actual numbers of traded VMs, $\hat{\gamma}_i^m(t)$ and $\hat{\eta}_i^m(t)$ ($\forall m\in [1, M]$), from the double auction. These values are determined by the auctioneer according to buy-bids $({b}_i^m(t), {\gamma}_i^m(t))$ and sell-bids $({s}_i^m(t),{\eta}_i^m(t))$ submitted by all clouds, and its double auction mechanism. That is, each cloud $i$ first proposes its buy-bids and sell-bids to the auctioneer, and then receives the auction results, based on which the \emph{job scheduling and server provisioning} decisions are made. We first investigate how each cloud proposes its buy-bids and sell-bids, and then decide optimal job scheduling and server provisioning in Sec.~\ref{sec:scheduling}.

A truthful double auction is employed at the auctioneer, where sellers and buyers bid their {\em true values} of the prices and quantities, in order to maximize their individual utilities. (\ref{eqn:profit-oneshot1}) is the utility maximization problem for each cloud. If we can find true values of each cloud $i$, $\tilde{b}_i^m(t), \tilde{\gamma}_i^m(t)$, $\tilde{s}_i^m(t)$ and $\tilde{\eta}_i^m(t)$, and let the cloud bid using these values, the achieved utility in (\ref{eqn:profit-oneshot1}) is guaranteed to be the largest, as compared to bidding any other values.


We decide the true values of the bids for each cloud $i$, according to their definitions in double auctions \cite{zhou-infocom09}\cite{xu-infocom10}. The true value of the price to buy (sell) a type-$m$ VM, $\tilde{b}_i^m(t)$ ($\tilde{s}_i^m(t)$), is such a value that, if a VM is purchased (sold) at a price (i) equal to this value, then cloud $i$'s profit remains the same, compared to not obtaining the VM; (ii) higher than this value, a profit loss (gain) at cloud $i$ occurs; and (iii) lower than this value, a profit gain (loss) results. In a multi-unit double auction, the true value of the maximum number of type-$m$ VMs cloud $i$ can buy (sell), $\tilde{\gamma}_i^m(t)$ ($\tilde{\eta}_i^m(t)$), is the maximum number of type-$m$ VMs the cloud is willing to buy (sell) at the true value of the price, {\em i.e.}, $\tilde{b}_i^m(t)$ ($\tilde{s}_i^m(t)$).

Using the above rationale and based on problem (\ref{eqn:profit-oneshot1}), the true values of the buy/sell prices for cloud $i$ can be derived as (detailed derivation steps are given in Appendix \ref{appendix:truevalue})

\vspace{-6mm}{\small\begin{align}
\tilde{b}_{i}^m(t) = \frac{Q_i^{s_m^*}(t)+Z_i^{s_m^*}(t)}{V\cdot g_{s_m^*}}, \label{eqn:true-buy}
\end{align}}\vspace{-5mm}

\noindent and

\vspace{-4mm}{\small\begin{align}
\tilde{s}_{i}^m(t)=\begin{cases} \frac{Q_i^{s_m^*}(t)+Z_i^{s_m^*}(t)}{V\cdot g_{s_m^*}}& \text{if }\frac{Q_i^{s_m^*}(t)+Z_i^{s_m^*}(t)}{V\cdot g_{s_m^*}}> \beta_{i}(t)/C_{i}^m\\ \beta_{i}(t)/C_{i}^m & \text{Otherwise}\end{cases},\label{eqn:true-sell}
\end{align}}\vspace{-5mm}

\noindent respectively, where

\vspace{-4mm}{\small\begin{align}
s_m^* = arg\max_{s'\in [1, S], m_{s'}=m}\{W_i^{s'}(t)\},\label{eqn:max-weight}
\end{align}}\vspace{-5mm}

\vspace{-4mm}{\small\begin{align}
\textrm{and}\hspace{1cm}W_i^{s'}(t) = \frac{Q_i^{s'}(t)+Z_i^{s'}(t)}{g_{s'}}.\label{eqn:weight}
\end{align}}\vspace{-5mm}

\noindent Here, $W_i^{s'}(t)$ denotes the weight for scheduling one type-$s'$ job (to run on type-$m_{s'}$ VM(s)) by cloud $i$ in $t$, and $s_m^*$ specifies the job type with the largest weight (ties broken arbitrarily), among all types of jobs requiring type-$m$ VMs. $W_i^{s'}(t)$ is determined by the following factors: (i) the sum of queue backlogs, $Q_i^{s'}(t)+Z_i^{s'}(t)$, representing the level of urgency for scheduling type-$s'$ jobs in $t$, since $Q_i^{s'}(t)$ is the number of unscheduled type-$s'$ jobs and $Z_i^{s'}(t)$ reflects the cumulated response delay; (ii) the number of concurrent VMs each type-$s'$ job requires, $g_{s'}$, which decides the job-scheduling difficulty. 

The intuition behind (\ref{eqn:true-buy}) and (\ref{eqn:true-sell}) includes: (i) the true value of the price to buy a type-$m$ VM depends on the combined effect of urgency and difficulty for scheduling jobs requiring this type of VMs, and is computed based on the maximum weight that any type of jobs requiring type-$m$ VMs may achieve; (ii) the true value of the price to sell one type-$m$ VM from cloud $i$ is the same as that of the price to buy, if the latter exceeds the current cost of operating a type-$m$ VM in the cloud; otherwise, it is set to the operational cost.

The true values of the number of type-$m$ VMs to buy and to sell at cloud $i$ are

\vspace{-4mm}{\small\begin{align}
\tilde{\gamma}_{i}^m(t) = \sum_{j\in [1, F]}C_{i}^m \cdot N_{i}^m,\label{eqn:true-vol-b}
\end{align}}\vspace{-5mm}

\vspace{-4mm}{\small\begin{align}
\textrm{and}\hspace{1cm}\tilde{\eta}_{i}^m(t) = C_{i}^m \cdot N_{i}^m,\label{eqn:true-vol-s}
\end{align}}\vspace{-5mm}

\noindent respectively. They state that the maximum number of type-$m$ VMs cloud $i$ is willing to buy (sell) at the price in (\ref{eqn:true-buy}) (in (\ref{eqn:true-sell})), is the number of all potential type-$m$ VMs in the federation. The rationale is as follows: The clearing price for transactions of type-$m$ VMs in the double auction is at most the buyer's true value in (\ref{eqn:true-buy}) and at least the seller's true value in (\ref{eqn:true-sell}), if the corresponding buy/sell bids are successful. By definition of the true value, if the actual charge per VM is lower (higher) than the true value, a profit gain happens at the buyer (seller), and the more VMs purchased (sold), the larger the profit gain. Therefore, a cloud is willing to buy or sell at the largest quantity possible, for profit maximization.\footnote{It may appear counter-intuitive that a cloud is willing to buy all type-$m$ VMs in the federation, 
regardless of its number of unscheduled jobs requiring type-$m$ VMs, \emph{i.e.}, $\sum_{s\in [1, S], m_s=m}Q_i^s(t)$. Interestingly, our proof in Sec.~\ref{sec:analysis} shows that bidding so in each time slot can achieve a time-averaged profit over the long run that approximates the offline optimum, and our simulation in Sec.~\ref{sec:simulation} shows that it performs better as compared to a bidding strategy that asks for the exact number of VMs to serve the unscheduled jobs.}



To conclude, in each time slot $t$, cloud $i$ submits its bids as $b_i^m(t)=\tilde{b}_i^m(t)$, $s_i^m(t)=\tilde{s}_i^m(t)$, $\gamma_i^m(t)=\tilde{\gamma}_{i}^m(t)$ and $\eta_i^m(t)=\tilde{\eta}_i^m(t)$, for each type of VMs $m\in [1,M]$.

\subsection{Server Provisioning, Job scheduling and Dropping}
\label{sec:scheduling}

After receiving results of the double auction (actual charges $\hat{b}_i^m(t)$, $\hat{s}_i^m(t)$, $\forall m\in [1, M]$, and the actual numbers of traded VMs $\hat{\gamma}_i^m(t)$, $\hat{\eta}_i^m(t)$, $\forall m\in [1, M]$, $\alpha_{ji}^{m_s}(t),\forall s\in[1,S],\forall j\in[1,F]$),
 cloud $i$ schedules its jobs on its local servers and (potentially) purchased VMs from other clouds, decides job drops and the number of active servers to provision, by solving optimization problems (\ref{eqn:profit-oneshot1}) and (\ref{eqn:profit-oneshot2}).

\subsubsection{Server provisioning}

We start with deriving $n_i^m(t)$, $\forall m\in [1, M]$, by assuming known values of $\hat{s}_i^m(t)$, $\hat{\eta}_i^m(t)$, $\hat{b}_i^m(t)$, $\hat{\gamma}_i^m(t)$, $\alpha_{ij}^m(t)$ and $\mu_{ij}^s(t)$ (we will present the value of $n_i^m(t)$ in terms of these variables). In this case, problem (\ref{eqn:profit-oneshot1}) is equivalent to the following minimization problem:

\vspace{-4mm}{\small
\begin{align}
\min&~~~~V \beta_i(t) \sum_{m\in [1, M]}n_i^m(t)\notag\\
s.t.&~~~~\text{Constraint (\ref{eqn:capacity1}), (\ref{eqn:capacity2}) and (\ref{eqn:capacity3})}.\notag
\end{align}}\vspace{-6mm}

Since $V \beta_i(t) \ge 0$, the best strategy is to assign the minimal feasible value to $n_i^m(t)$, $\forall m\in[1,M]$, that satisfies constraints (\ref{eqn:capacity1}) and (\ref{eqn:capacity3}), which can be combined into

\vspace{-4mm}{\small
\begin{align*}
\sum_{s\in [1,S], m_s=m}\mu_{ii}^s(t)\cdot g_s + \sum_{j\neq i}\alpha_{ji}^m(t) \leq C_i^m n_i^m(t).
\end{align*}}\vspace{-4mm}

\noindent Hence, the optimal number of activated servers at cloud $i$ to provision type-$m$ VM can be calculated as

\vspace{-4mm}{\small\begin{align}
n_{i}^m(t) = (\sum_{s\in [1,S], m_s=m}\mu_{ii}^s(t)\cdot g_s + \sum_{j\neq i}\alpha_{ji}^m(t))/C_{i}^m.\label{eqn:server-provision}
\end{align}}\vspace{-4mm}

\noindent These many servers can provide enough type-$m$ VMs for serving local jobs and selling to other clouds.

\subsubsection{Job scheduling}

We now derive $\mu_{ij}^s(t)$, $\forall j\in [1,F]$, $s\in [1, S]$, by assuming known values of $\hat{s}_i^m(t)$, $\hat{\eta}_i^m(t)$, $\hat{b}_i^m(t)$, $\hat{\gamma}_i^m(t)$ and $\alpha_{ij}^m(t)$, with $n_i^m(t)$ given in Eqn.~(\ref{eqn:server-provision}). Problem (\ref{eqn:profit-oneshot1}) is equivalent to the following maximization problem:

\vspace{-4mm}{\small
\begin{align}
\max&~~~~\sum_{s\in [1, S]}\sum_{j\in [1, F]}\mu_{ij}^s(t) [Q_i^s(t) + Z_i^s(t)]\nonumber\\
&~~~~-V \beta_i \sum_{s\in [1,S], m_s=m}\mu_{ii}^s(t)\cdot \frac{g_s}{C_i^m}\notag\\
s.t.&~~~~\text{Constraint (\ref{eqn:capacity1}), (\ref{eqn:capacity2}) and (\ref{eqn:capacity3})}.\notag
\end{align}}\vspace{-6mm}


This is a maximum-weight scheduling problem, with $Q_i^s(t)+Z_i^s(t)$ as the per-job scheduling weight for each $\mu_{ij}^s(t)$ ($j\neq i$) and $Q_i^s(t)+Z_i^s(t)-\frac{V\beta_i(t)  g_s}{C_i^m}$ as the per-job scheduling weight for each $\mu_{ii}^s(t)$. There are two cases:

\vspace{1mm}
$\triangleright$ $j = i$: In this case, by combining constraints (\ref{eqn:capacity1}), (\ref{eqn:capacity2}) and (\ref{eqn:capacity3}), we have

\vspace{-4mm}{\small
\begin{align*}
\sum_{s:m_s=m, s\in [1, S]}g_s\mu_{ii}^s(t)\leq C_i^m N_i^m-\sum_{j\neq i}\alpha_{ji}^m(t).
\end{align*}}\vspace{-4mm}

Based on the above maximum-weight problem, we know that 
 the best strategy is to assign all the remaining type-$m_s$ VMs in cloud $i$, $C_i^{m_s} n_i^{m_s}(t)-\sum_{j\neq i}\alpha_{ji}^{m_s}(t)$ (the maximum number of on-premise type-$m_s$ VMs minus those sold to other clouds), to serve its own jobs of service type $s_{m_s}^*$ with the largest per-VM scheduling weight $\frac{Q_i^s(t)+Z_i^s(t)}{g_s}-\frac{V\beta_i(t)}{C_i^{m_s}}$ if it is positive (equivalently, the largest $\frac{Q_i^s(t)+Z_i^s(t)}{g_s}$ 
 if $\frac{Q_i^s(t)+Z_i^s(t)}{g_s}>\frac{V\beta_i(t)}{C_i^{m_s}}$), among all job types requiring type-$m_s$ VMs. Otherwise, cloud $i$ does not serve any jobs using its own servers in $t$. Hence, we derive the optimal number of cloud $i$'s type-$s$ jobs scheduled to run on the cloud's local servers as 

\vspace{-4mm}{\small\begin{align}
\mu_{ii}^s(t)=\begin{cases}\frac{C_{i}^{m_s}\cdot N_{i}^{m_s}- \sum_{j\neq i}\alpha_{ji}^{m_s}(t)}{g_s} & \text{if }\frac{Q_i^s(t)+Z_i^s(t)}{ \cdot g_s }> \frac{V\beta_{i}(t)}{C_{i}^{m_s}}\\ & \text{and }\ s=s_{m_s}^*\\ 0 & \text{Otherwise}\end{cases}.\label{eqn:schedule1}
\end{align}}\vspace{-5mm}

\vspace{1mm}
$\triangleright$ $j\neq i$: $\mu_{ij}^s(t)$ can be directly derived by $\alpha_{ij}^{m_s}(t)$, which is the number of type-$m_s$ VMs cloud $i$ purchased from cloud $j$ (constraint (\ref{eqn:capacity1}) is satisfied by our server provisioning decision in Eqn.~(\ref{eqn:server-provision}), and constraint (\ref{eqn:capacity2}) is met by Eqn.~(\ref{eqn:schedule1}) and (\ref{eqn:server-provision})), based on constraint (\ref{eqn:capacity3}).
Similar to the previous case, we know that 
 the best strategy is to assign all the type-$m_s$ VMs purchased, $\alpha_{ij}^{m_s}(t)$, to serve jobs of service type $s_{m_s}^*$ with the largest per-VM scheduling weight $\frac{Q_i^s(t)+Z_i^s(t)}{g_s}$, as defined in Eqn.~(\ref{eqn:max-weight}) and (\ref{eqn:weight}). Hence, we derive the optimal solution to the number of type-$s$ jobs to run at cloud $j(\neq i)$  
 as

\vspace{-4mm}{\small\begin{align}
\mu_{ij}^s(t)=\begin{cases}\alpha_{ij}^{m_s}(t)/g_s& \text{if }s=s_{m_s}^*\\ 0 & \text{Otherwise}\end{cases}.\label{eqn:schedule2}
\end{align}}\vspace{-4mm}



\subsubsection{Job dropping}

Problem (\ref{eqn:profit-oneshot2}) is a maximum-weight problem with weight $Q_i^s(t)+Z_i^s(t)-V\cdot \xi_i^s$ for job-dropping decision variable $D_i^s(t)$, $\forall s\in [1,S]$, in the objective function. If the weight $Q_i^s(t)+Z_i^s(t)-V\cdot \xi_i^s> 0$ ({\em i.e.}, if the level of urgency for scheduling type-$s$ jobs $Q_i^s(t)+Z_i^s(t)$ exceeds the weighted job-drop penalty $V\cdot \xi_{i}^s$), type-$s$ jobs in queue $Q_i^s$ should be dropped at the maximum rate, \emph{i.e.}, $D_i^s(t)=D_i^{s(max)}$, in order to maximize the objective function value; otherwise, there is no drop, {\em i.e.}, $D_i^s(t)=0$. Therefore, the optimal number of type-$s$ jobs dropped by cloud $i$ in $t$ is

\vspace{-4mm}{\small\begin{align}
D_i^s(t)=\begin{cases}D_i^{s(max)} & \text{if }Q_i^s(t)+Z_i^s(t)> V\cdot \xi_{i}^s\\ 0 & \text{Otherwise}\end{cases}.\label{eqn:drop}
\end{align}}\vspace{-6mm}


\vspace{2mm}
In the above results, we note that the derived job scheduling and drop numbers do not need to be bounded by the number of unscheduled jobs in the corresponding job queue, {\em i.e.}, $\mu_{ij}^s(t)$ and $D_i^s(t)$ are not required to be bounded by $Q_i^{s}(t)$ according to Eqn.~(\ref{eqn:queue1}). Nevertheless, the actual number of jobs to schedule/drop when running the algorithm, is upper bounded by the length of the job queue.

%

\subsection{The Dynamic Algorithm}\label{subsec:derivation-alg}

Alg.~\ref{alg:profit} summarizes the dynamic algorithm for each cloud to carry out in each time slot, in order to maximize its time-averaged profit over the long run.

{\small \vspace{-3mm}
\begin{algorithm}[!h]
\small \caption{Dynamic Profit Maximization Algorithm at cloud $i$ in Time Slot $t$} \label{alg:profit}

\textbf{Input}: $r_i^s(t)$, $Q_i^s(t)$, $Z_i^s(t)$, $g_s$, $m_s$, $\xi_i^s$, $C_i^m$, $N_i^m$ and $\beta_i(t)$, $\forall s\in [1, S]$.

\textbf{Output}: $b_i^m(t)$, $s_i^m(t)$, $\gamma_i^m(t)$, $\eta_i^m(t)$, $D_i^s(t)$, $\mu_{ij}^s(t)$ and $n_i^m(t)$, $\forall m\in [1, M], s\in [1, S], j\in [1, F]$.
\begin{algorithmic}[1]
\State \textbf{VM valuation and bid}: Decide $b_i^m(t)$, $s_i^m(t)$, $\gamma_i^m(t)$ and $\eta_i^m(t)$ with Eqn.~(\ref{eqn:true-buy})-(\ref{eqn:true-vol-s});

\State \textbf{Server provisioning, job scheduling and dropping}: Decide $\mu_{ij}^s(t)$, $D_i^s(t)$ and $n_i^m(t)$ with Eqn.~(\ref{eqn:schedule2}), (\ref{eqn:schedule1}), (\ref{eqn:drop}) and (\ref{eqn:server-provision});

\State Update $Q_i^s(t)$ and $Z_i^s(t)$ with Eqn.~(\ref{eqn:queue1}) and (\ref{eqn:queue2}).

\end{algorithmic}
\end{algorithm}
}\vspace{-3mm}


We analyze the computation and communication complexities of Alg.~\ref{alg:profit} as follows.

\vspace{1mm}
\noindent \textbf{Computation complexity}: We study the computation complexity for each algorithm module respectively.

\vspace{1mm}
\noindent $\triangleright$ \emph{VM valuation and bid}: The algorithm should first calculate the value of $s_m^*$ for each VM type $m\in [1, M]$ with Eqn.~(\ref{eqn:max-weight}) by comparing the weights $W_i^{s'}(t)$ among different types of jobs. In fact, the weight for each job type $s\in [1, S]$ is only evaluated once since it is only involved in the calculation of $s_m^*$ where $m=m_s$. Hence, the computation overhead to find $s_m^*,\ \forall m\in [1, M]$, is $O(S)$. Based on the value of $s_m^*$, the buy/sell bids of type-$m$ VMs can be decided by Eqn.~(\ref{eqn:true-buy})-(\ref{eqn:true-vol-s}) in constant time. For all $M$ VM types, the computation overhead is $O(M)$. Hence, the overall computation complexity for this algorithm module is $O(S+M)$.

\vspace{1mm}
\noindent $\triangleright$ \emph{Server provisioning, job scheduling and dropping}: With $s_m^*,\ \forall m\in [1, M]$, calculated in the above algorithm module, we can directly know the value of $s_{m_s}^*,\ \forall s\in [1,S]$. Then, the job scheduling decision $\mu_{ij}^s(t)$ for job type $s$ can be made in constant time based on Eqn.~(\ref{eqn:schedule2}) and (\ref{eqn:schedule1}). For all $S$ job types, the computation overhead is $O(S)$.

The server provisioning decisions can be found in constant time based on the job scheduling decisions and the auction results, according to Eqn.~(\ref{eqn:server-provision}) for type-$m$ VMs. For all $M$ VM types, the computation overhead is $O(M)$.

Job dropping is also decided in constant time for type-$s$ jobs based on Eqn.~(\ref{eqn:drop}). For all $S$ job types, the computation complexity is $O(S)$.

\vspace{1mm}
\noindent $\triangleright$ \emph{Queue update}: For each job type $s$, the job queue $Q_i^s(t)$ and virtual queue $Z_i^s(t)$ can be updated in constant time based on Eqn.~(\ref{eqn:queue1}) and (\ref{eqn:queue2}). Hence, for all $S$ job types, the computation overhead is $O(S)$.

\vspace{1mm}
In summary, the computation complexity of Alg.~\ref{alg:profit} is $O(S+M)$.

\vspace{1mm}
\noindent \textbf{Communication complexity}: The input to Alg.~\ref{alg:profit} is mostly derived from local information. There is no direct information exchange among individual clouds. The only communication overhead is incurred when a cloud sends its VM bids to the auctioneer and receives the auction results for each VM type. Since there are $M$ VM types, the communication complexity is $O(M)$ for each cloud.

\section{Double Auction Mechanism}
\label{sec:auction_mechanism}

We next design a double auction mechanism for inter-cloud VM trading, which not only is truthful, individual rational and ex-post budget balanced, but also can enable satisfactory social welfare (Theorems \ref{theorem:truthfulness}-\ref{theorem:budget} and \ref{theorem:welfare}, Sec.~\ref{sec:analysis}).






The true values of buy and sell bids at each participating cloud (Eqn.~(\ref{eqn:true-buy})-(\ref{eqn:true-vol-s})) are not related to the detailed auction mechanism. The true values of the maximum numbers of VMs a cloud is willing to trade ($\hat{\gamma}_i^m(t)$ and $\hat{\eta}_i^m(t)$ in (\ref{eqn:true-vol-b}) and (\ref{eqn:true-vol-s})) are time-independent constants determined by system parameters $C_i^m$ and $N_i^m$. These parameters, and thus $\hat{\gamma}_i^m(t)$ and $\hat{\eta}_i^m(t)$, are easily known to other clouds, and hence it is not meaningful for a buyer/seller to bid otherwise. We correspondingly design a double auction where ${\gamma}_i^m(t)$ in each buy-bid is fixed to the value in (\ref{eqn:true-vol-b}) and ${\eta}_i^m(t)$ in each sell-bid is always the value in (\ref{eqn:true-vol-s}), while the buy/sell prices, $b_i^m(t)$'s and $s_i^m(t)$'s, can be decided by the respective buyers/sellers.

%
%

The following mechanism is carried out by the auctioneer at the beginning of each time slot $t$, to decide the actual trading price and number for each type of VMs $m\in [1, M]$.


\vspace{1mm}
\noindent \textbf{1. Winner Determination}:
The auctioneer sorts all received buy-bids for type-$m$ VMs in descending order in the buy prices. Let $\theta_j^m(t)$ be the $j^{th}$ highest. Two buy-bids with the \emph{largest} and \emph{second largest} prices, $\theta_1^m(t)$, $\theta_2^m(t)$, are identified (ties broken arbitrarily). The sell-bids for type-$m$ VMs are sorted in ascending order in the sell prices. Let $\vartheta_j^m(t)$ be the $j^{th}$ lowest, with $L_j^m(t)$ as the corresponding maximum number of VMs to sell, such that $\vartheta_1^m(t)\leq \vartheta_2^m(t)\leq \ldots \leq \vartheta_N^m(t)$. Let $j'$ be the critical index in the sorted sequence of sell-bids, such that $\vartheta_{j'}^m(t)$ is the largest sell price not exceeding $\theta_2^m(t)$, \emph{i.e.},

    \vspace{-4mm}{\small\begin{align}
    \vartheta_{j'}^m(t) \leq \theta_2^m(t),\ \text{and }\vartheta_{j'+1}^m(t) > \theta_2^m(t).\label{eqn:j'}
    \end{align}}\vspace{-5mm}

If there are at least two sell-bids $\vartheta_1^m(t)$ and $\vartheta_2^m(t)$ no higher than the second largest buy price $\theta_2^m(t)$, the highest buy-bid $\theta_1^m(t)$ wins, and the sell-bids with the lowest to the ${(j'-1)}^{\mbox{th}}$ lowest sell prices ($\vartheta_j^m(t)\leq \vartheta_{j'}^m(t)$, not including $j'$) win. Otherwise, no buy/sell bid wins.

\vspace{1mm}
\noindent \textbf{2. Pricing and Allocation}: It is a NP-hard problem to clear the double auction market with discriminatory prices \cite{sandholm-IJICAI01}. We apply a uniform clearing price to winning buy/sell bids of type-$m$ VMs, as follows.

\noindent $\triangleright$ The price charged to each buyer cloud $i$ of type-$m$ VMs is

\vspace{-5mm}{\small\begin{align}
\hat{b}_i^m(t) = \begin{cases}\theta_{2}^m(t) & \text{if bid }b_i^m(t)\text{ wins},\\ 0 & \text{otherwise.}\end{cases} \label{eqn:buy-price}
\end{align}}\vspace{-4mm}

\noindent $\triangleright$ The price paid to each seller cloud $i$ of type-$m$ VMs is

\vspace{-5mm}{\small\begin{align}
\hat{s}_{i}^m(t) = \begin{cases}\vartheta_{j'}^m(t) & \text{if bid }s_i^m(t)\text{ wins},\\ 0 & \text{otherwise.}\end{cases} \label{eqn:sell-price}
\end{align}}\vspace{-4mm}

\noindent $\triangleright$ The number of type-$m$ VMs bought by cloud $i$ is

\vspace{-5mm}{\small\begin{align}
\hat{\gamma}_i^m(t) = \begin{cases}\sum_{j=1}^{j'-1}L_j^m(t) & \text{if bid }b_i^m(t)\text{ wins},\\ 0 & \text{otherwise.}\end{cases} \label{eqn:buy-vol}
\end{align}}\vspace{-4mm}

\noindent $\triangleright$ The number of type-$m$ VMs sold by cloud $i$ is

\vspace{-5mm}{\small\begin{align}
\hat{\eta}_{i}^m(t) = \begin{cases}\eta_{i}^m(t) & \text{if bid }s_i^m(t)\text{ wins},\\ 0 & \text{otherwise.}\end{cases} \label{eqn:sell-vol}
\end{align}}\vspace{-4mm}

\noindent $\triangleright$ The number of type-$m$ VMs sold from cloud $j$ to cloud $i$ is

\vspace{-4mm}{\small\begin{align}
\alpha_{ij}^m(t) = \begin{cases}\eta_{j}^m(t) & \text{if bids }b_i^m(t)\text{ and }s_j^m(t)\text{ win},\\ 0 & \text{otherwise.}\end{cases} \label{eqn:buy-sell}
\end{align}}\vspace{-4mm}

For example, consider a federation of 4 clouds with buy and sell prices bid in Table \ref{table:example}, each seeking to buy/sell one VM. Clouds 2 and 3 bid the two largest buy prices \$$20$ and \$$15$, which are higher than sell prices from clouds 1 and 4. Hence the buyer cloud 2 and the seller cloud 4 win, while the clearing buy and sell prices are \$$15$ and \$$13$, respectively.

\vspace{-3mm}
\begin{table}[h]
  \centering
  \caption{Double auction bids: an illustrative example}\label{table:example}\vspace{-3mm}
  \begin{tabular}{|c|c|c|c|c|}
    \hline
     & Cloud 1 & Cloud 2 & Cloud 3 & Cloud 4 \\\hline
    Buy-bid & \$10 & \$20 & \$15 & \$8 \\\hline
    Sell-bid & \$13 & \$22 & \$16 & \$9\\
    \hline
  \end{tabular}
\end{table}\vspace{-4mm}

\section{Dynamic Social-Welfare Maximization Algorithm: a Benchmark}\label{sec:welfare}

We also present a dynamic algorithm that maximizes the time-averaged social welfare in the federation (optimization problem (\ref{eqn:social-max})), and its derivation steps based on the Lyapunov optimization framework. This algorithm is used as a benchmark to examine the efficiency of Alg.~\ref{alg:profit} in social welfare.



\subsection{Derivation Details}

Similar to the derivation of Alg.~\ref{alg:profit}, we first derive a one-shot optimization problem (\ref{eqn:profit-oneslot2}) for the federation to solve based on the \emph{drift-plus-penalty} framework of Lyapunov optimization, and then derive the dynamic benchmark algorithm to solve it optimally in each time slot.

In each time slot $t$, define the set of queues $\Theta(t)$ in the federation as

\vspace{-5mm}{\small
\begin{align*}
\Theta(t)= \{Q_i^s(t),Z_i^s(t)|i\in [1, F], s\in [1, S]\}.
\end{align*}}\vspace{-6mm}

Define the Lyapunov function as follows:

\vspace{-4mm}{\small
\begin{align*}
L(\Theta(t)) = \frac{1}{2}\sum_{i\in[1, F]}\sum_{s\in [1, S]}[(Q_i^s(t))^2 + (Z_i^s(t))^2].
\end{align*}}\vspace{-4mm}

Then the one-slot conditional Lyapunov drift is

\vspace{-5mm}{\small
\begin{align*}
\Delta(\Theta(t))= L(\Theta(t+1)) - L(\Theta(t)). 
\end{align*}}\vspace{-6mm}

Squaring the queuing laws Eqn.~(\ref{eqn:queue1}) and (\ref{eqn:queue2}), we can derive the following inequality (details can be found in Appendix \ref{appendix:drift2})

\vspace*{-5mm}{\small\begin{align}
&\Delta(\Theta(t))+V\cdot \sum_{i\in [1, F]}[\sum_{m\in [1, M]} [\beta_i(t) n_i^m(t)] + \sum_{s\in [1, S]}D_i^s(t) \xi_i^s \notag\\
&-\sum_{s\in [1, S]} p_i^s(t)r_i^s(t)]\notag\\
\leq&  B + \sum_{i\in [1, F]}\sum_{s\in [1, S]}[Q_i^s(t) r_i^s(t) + Z_i^s(t) \epsilon_s - V p_i^s(t)r_i^s(t)]\notag\\
&- \varphi_1(t)-\varphi_2(t),\label{eqn:drift-plus-penalty2}
\end{align}}\vspace*{-6mm}

\noindent where $V>0$ is a user-defined positive parameter for gauging the optimality of the time-averaged social welfare, 
 $B=\sum_{i\in [1, F]}B_i$ is a constant with $B_i=\frac{1}{2}\sum_{s\in [1, S]}[[\sum_{j=1}^F C_j^{m_s}N_j^{m_s}/g_s+D_i^{s(max)}]^2 + [R_i^s]^2+[\epsilon_s]^2 + [D_i^{s(max)}+\sum_{j=1}^F C_j^{m_s}N_j^{m_s}/g_s]^2]$, and 

\vspace{-4mm}
{\small \begin{align}
\varphi_1(t) =& \sum_{i\in [1, F]}[\sum_{s\in [1, S]}[Q_i^s(t) + Z_i^s(t)]\cdot \sum_{j\in [1, F]}\mu_{ij}^s(t)\notag\\
&-V \beta_i(t) \sum_{m\in [1, M]}  n_i^m(t)],\label{eqn:varphi1}\\
\varphi_2(t) =& \sum_{i\in [1, F]}\sum_{s\in [1, S]}D_i^s(t)[Q_i^s(t)+Z_i^s(t)-V\cdot \xi_i^s].\label{eqn:varphi2}
\end{align}}\vspace{-4mm}



Based on the drift-plus-penalty framework \cite{book2010}, a dynamic
algorithm can be derived for the federation to observe job and virtual queues $\Theta(t)$, job arrival rates ($r_i^s(t),\ \forall i\in [1,F], s\in [1, S]$), the current cost for server operation ($\beta_i(t),\ \forall i\in [1, F]$) in each
time slot, and minimizes the RHS of the inequality (\ref{eqn:drift-plus-penalty2}),
such that a lower bound for the time-averaged social welfare is maximized. Note that $B + \sum_{i\in [1, F]}\sum_{s\in [1, S]}[Q_i^s(t) r_i^s(t) + Z_i^s(t) \epsilon_s - V p_i^s(t)r_i^s(t)]$ in the RHS of (\ref{eqn:drift-plus-penalty2}) is fixed in time slot $t$. Hence, to maximize a lower bound of the time-averaged social welfare for the federation, the dynamic algorithm should solve the one-shot optimization problem in each time slot $t$ as follows:


\vspace{-5mm}
{\small \begin{align}
\max&~~~~\varphi_1(t) + \varphi_2(t)\label{eqn:profit-oneslot2} \\
s.t.&~~~~\text{Constraints (\ref{eqn:drop-cons}), (\ref{eqn:capacity1})-(\ref{eqn:capacity2})},\ \forall i\in [1, F].\nonumber
\end{align}}\vspace{-6mm}

The maximization problem in (\ref{eqn:profit-oneslot2}) can be decoupled into two independent optimization problems:

\vspace{-5mm}{\small
\begin{align}
\max&~~\varphi_1(t)\label{eqn:profit-oneshot1-1}~~~~s.t.~~\text{Constraint (\ref{eqn:capacity1})-(\ref{eqn:capacity2})},\ \forall i\in [1, F],
\end{align}}\vspace{-6mm}

\noindent which is related to decisions on job scheduling and server provisioning, and

\vspace{-5mm}{\small
\begin{align}
\max&~~\varphi_2(t)\label{eqn:profit-oneshot2-2}~~~~s.t.&~~\text{Constraint (\ref{eqn:drop-cons})},\ \forall i\in [1, F], 
\end{align}}\vspace{-6mm}

\noindent which is related to decisions on job dropping. We note that to maximize social welfare, the decisions that the federation needs to make are not related to any inter-cloud VM trading mechanism, since the income and expenditure due to VM trades among the clouds have canceled each other when calculating the social welfare. We next solve problem (\ref{eqn:profit-oneshot1-1}) and problem (\ref{eqn:profit-oneshot2-2}) to derive the optimal decisions.



\emph{1) Server provisioning:} We start with solving $n_i^m(t)$, $\forall m\in [1, M]$, $i\in[1,F]$, by assuming known values of job scheduling decisions $\mu_{ij}^s(t)$'s, and present the value of the former in terms of the latter. 
 In this case, problem (\ref{eqn:profit-oneshot1-1}) is equivalent to the following minimization problem:

\vspace{-4mm}{\small
\begin{align}
\min&~~~~V \sum_{i\in [1, F]}\beta_i(t) \sum_{m\in [1, M]}n_i^m(t)\notag\\ 
s.t.&~~~~\text{Constraint (\ref{eqn:capacity1}) - (\ref{eqn:capacity2})},\ \forall i\in [1, F].\notag
\end{align}}\vspace{-6mm}

Since $V\beta_i(t)\ge 0$, the best strategy is to assign the minimal feasible value to $n_i^m(t)$, for each VM type $m$ at each cloud $i$, that satisfies constraints (\ref{eqn:capacity1}) and (\ref{eqn:capacity2}). Hence, the optimal number of activated servers at cloud $i$ to provision type-$m$ VM is

\vspace{-4mm}{\small\begin{align}
n_{i}^m(t) = [\sum_{j\in [1,F]}\sum_{s\in[1,S], m_s=m}\mu_{ji}^s(t)\cdot g_s]/C_{i}^m.\label{eqn:server-provision2}
\end{align}}\vspace{-4mm}

\emph{2) Job scheduling:} We next derive $\mu_{ij}^s(t)$, $\forall i\in [1, F]$, $j\in [1,F]$, $s\in [1, S]$, with $n_i^m(t)$ given in Eqn.~(\ref{eqn:server-provision2}). Problem (\ref{eqn:profit-oneshot1-1}) is equivalent to the following maximization problem:

\vspace{-4mm}{\small
\begin{align}
\max&~~~~\sum_{i\in [1, F]}\sum_{s\in [1, S]}\sum_{j\in [1, F]}\mu_{ij}^s(t)\cdot[Q_i^s(t)+Z_i^s(t)-V \beta_j\cdot \frac{g_s}{C_j^{m_s}}]\notag\\ 
s.t.&~~~~\text{Constraint (\ref{eqn:capacity1}) - (\ref{eqn:capacity2})},\ \forall i\in [1, F].\notag
\end{align}}\vspace{-6mm}


This is a maximum-weight scheduling problem, with $Q_i^s(t)+Z_i^s(t)-\frac{V\beta_j(t)  g_s}{C_j^{m_s}}$ as the per-job scheduling weight for each $\mu_{ij}^s(t)$. Combining constraints (\ref{eqn:capacity1}) and (\ref{eqn:capacity2}), we have

\vspace{-4mm}{\small
\begin{align*}
\sum_{i\in [1, F]}\sum_{s:m_s=m, s\in [1, S]}g_s\mu_{ij}^s(t)\leq C_j^m N_j^m,\ \forall j\in [1, F].
\end{align*}}\vspace{-4mm}

The best strategy is to assign all the type-$m$ VMs in cloud $j$ at the number of $C_j^m N_j^m$ to serve jobs of type $\acute{s}_{m}$ of cloud $\acute{i}_{m}$ with the maximum per-VM scheduling weight $\frac{Q_i^s(t)+Z_i^s(t)}{g_s}-\frac{V\beta_j(t)}{C_j^{m_s}}$ if it is positive (equivalently, the largest $\frac{Q_i^s(t)+Z_i^s(t)}{g_s}$ as defined in Eqn.~(\ref{eqn:weight}) and (\ref{eqn:max-weight}) in Sec.~\ref{sec:algorithm} if $\frac{Q_i^s(t)+Z_i^s(t)}{g_s}>\frac{V\beta_j(t)}{C_j^{m_s}}$), among all job types from all clouds requiring type-$m$ VMs. Hence, the optimal solution to the number of type-$s$ jobs of cloud $i$ to run at cloud $j$ is 

\vspace{-4mm}{\small\begin{align}
\mu_{ij}^s(t)=\begin{cases}C_{j}^{m_s}\cdot N_{j}^{m_s} / g_s& \text{if }\frac{Q_i^s(t)+Z_i^s(t)}{ g_s}>  \frac{V\beta_{i}(t)}{C_{i}^{m_s}}\\ & \text{ and }\ <i,s>=<\acute{i}_m, \acute{s}_m>,\\ 0 & \text{Otherwise,}\end{cases}\label{eqn:schedule3}
\end{align}}\vspace{-4mm}

\noindent where

\vspace{-5mm}{\small\begin{align}
<\acute{i}_{m},\acute{s}_{m}> = arg\max_{i\in [1,F],s\in [1, S], m_s=m}\{W_i^s(t)\}, \label{eqn:weight2}
\end{align}}\vspace{-5mm}

\noindent and $W_i^s(t)$ is the weight defined in Eqn.~(\ref{eqn:weight}).

%
%
%
%



\emph{3) Job dropping:} Problem (\ref{eqn:profit-oneshot2-2}) is a maximum-weight problem with weight $Q_i^s(t)+Z_i^s(t)-V\cdot \xi_i^s$ for job-dropping decision $D_i^s(t)$  
 in the objective function. If $Q_i^s(t)+Z_i^s(t)-V\cdot \xi_i^s> 0$, type-$s$ jobs at cloud $i$ should be dropped at the maximum rate; 
 otherwise, there is no drop. 
 Hence, the optimal number of type-$s$ jobs dropped by cloud $i$ in $t$ is 

\vspace{-4mm}{\small\begin{align}
D_i^s(t)=\begin{cases}D_s^{(max)} & \text{if }Q_i^s(t)+Z_i^s(t)> V\cdot \xi_{i}^s\\ 0 & \text{Otherwise.}\end{cases}\label{eqn:drop2}
\end{align}}\vspace{-5mm}

\subsection{The Dynamic Benchmark Algorithm}

Alg.~\ref{alg:social} summarizes the dynamic algorithm for the federation
to carry out ({\em e.g.}, on a centralized controller) in each time slot, in order to maximize its time-averaged
social welfare over the long run.






{\small \vspace{-4mm}
\begin{algorithm}[H]
\small \caption{Dynamic Social Welfare Maximization Algorithm in Time Slot $t$} \label{alg:social}

\textbf{Input}: $r_i^s(t)$, $Q_i^s(t)$, $Z_i^s(t)$, $g_s$, $m_s$, $\xi_i^s$, $C_i^m$, $N_i^m$ and $\beta_i(t)$, $\forall i\in [1, F], s\in [1, S]$.

\textbf{Output}: $D_i^s(t)$, $\mu_{ij}^s(t)$ and $n_i^m(t)$, $\forall i\in [1, F], m\in [1, M], s\in [1, S]$.
\begin{algorithmic}[1]
\State \textbf{Job scheduling and server provisioning}: Decide $\mu_{ij}^s(t)$ and $n_i^m(t)$ with Eqn.~(\ref{eqn:schedule3}) and (\ref{eqn:server-provision2});

\State \textbf{Job dropping}: Decide $D_i^s(t)$ with Eqn.~(\ref{eqn:drop2});

\State Update $Q_i^s(t)$ and $Z_i^s(t)$ with Eqn.~(\ref{eqn:queue1}) and (\ref{eqn:queue2}).

\end{algorithmic}
\end{algorithm}
}\vspace{-5mm}

\section{Performance Analysis}\label{sec:analysis}

We next analyze the performance guarantee provided by our dynamic individual-profit maximization algorithm and the double auction mechanism. \opt{short}{Due to space limit, all detailed proofs can be found in \cite{tech-report}.}

\subsection{Properties of the Double Auction Mechanism}

\opt{long}{
\begin{theorem}[True Valuation]\label{theorem:truevalues}
The VM valuations on buy-bids, \emph{i.e.}, Eqn.~(\ref{eqn:true-buy}) and (\ref{eqn:true-vol-b}), and sell-bids, \emph{i.e.}, Eqn.~(\ref{eqn:true-sell}) and (\ref{eqn:true-vol-s}), are true values.
\end{theorem}

\vspace{1mm}
This theorem is proved based on the definition of the true values and the optimization problem (\ref{eqn:profit-oneslot}) solved in each time slot by each cloud in Appendix \ref{appendix:truevalue}.

}
\begin{theorem}[Truthfulness]\label{theorem:truthfulness}
Bidding truthfully is the dominant strategy of each cloud in the double auction in Sec.~\ref{sec:auction_mechanism}, \emph{i.e.}, no cloud can achieve a higher profit in (\ref{eqn:profit-oneslot}) by bidding with values other than its true values of the buy and sell bids, in Eqn.~(\ref{eqn:true-buy})(\ref{eqn:true-vol-b})(\ref{eqn:true-sell})(\ref{eqn:true-vol-s}).
\end{theorem}

\vspace{1mm}
\opt{long}{
We prove this theorem by contradiction and show that, in all cases, no cloud can do better with problem (\ref{eqn:profit-oneslot}) by bidding untruthfully. Details are in Appendix \ref{appendix:truthfulness}.
}

\begin{theorem}[Individual Rationality]\label{theorem:rationality}
No winning buyer pays more than its buy-bid price, and no winning seller is paid less than its sell-bid price, \emph{i.e.},
$\hat{b}_i^m(t)\leq b_i^m(t) \text{ and }\hat{s}_i^m(t) \geq s_i^m(t),\ \forall i\in [1, F], m\in [1, M].$
\end{theorem}

\vspace{1mm}
\opt{long}{This theorem can be proved based on the winner determination and pricing schemes in our auction mechanism, with details in Appendix \ref{appendix:rationality}.}
Given that the buy-bid (sell-bid) price is the true value of the buyer (seller), this theorem implies that a cloud can receive a non-negative profit gain, if it successfully sells or buys VMs. Hence, {\em a cloud's profit obtained in a federation with potential VM trades with others, is always no lower than that obtained when operating alone.}

%

\begin{theorem}[Ex-post Budget Balance]\label{theorem:budget} 
At the auctioneer, the total payment collected from the buyers is no smaller than the overall price paid to the sellers, {\em i.e.},

\vspace{-4mm}{\small
\begin{align*}
\sum_{i\in [1, F]}[\hat{b}_i^m(t)\cdot \hat{\gamma}_i^m(t)-\hat{s}_i^m(t)\cdot \hat{\eta}_i^m(t)] \geq 0,\ \forall m\in [1, M].
\end{align*}}\vspace{-3mm}
\end{theorem}

\opt{long}{
This theorem is proved based on Eqn.~(\ref{eqn:j'}) - (\ref{eqn:sell-price}),  
 with details in Appendix \ref{appendix:budget}.
}

\subsection{SLA Guarantee}

\begin{lemma}\label{lemma:bounded-queue}
Let
 $Q_i^{s(max)} = V \xi_i^s + R_i^s$ and $Z_i^{s(max)} = V  \xi_i^s+\epsilon_s$. 
 If $D_i^{s(max)} \geq \max\{R_i^s, \epsilon_s\}$, 
 each job queue $Q_i^s(t)$ and each virtual queue $Z_i^s(t)$ are upper-bounded by $Q_i^{s(max)}$ and $Z_i^{s(max)}$, respectively, in $t\in[0,T-1]$, $\forall i\in [1, F], s\in [1, S]$.
\end{lemma}

\vspace{1mm}
This lemma can be proved by analyzing the job drop decision in (\ref{eqn:drop}) and the queue updates in (\ref{eqn:queue1})(\ref{eqn:queue2}). The condition $D_i^{s(max)} \geq \max\{R_i^s, \epsilon_s\}$ ensures that, when the queue lengths grow to satisfy the job drop condition, any further increase on the queues, \emph{e.g.}, $R_i^s$ and $\epsilon_s$, can be balanced by dropping enough number of jobs at the rate of $D_i^{s(max)}$. \opt{long}{Detailed proof is included in Appendix \ref{appendix:bounded-queue}.}

\vspace{1mm}
\begin{theorem}[SLA Guarantee]\label{theorem:SLA}
Each job of type $s\in [1, S]$ is either scheduled or dropped with Alg.~\ref{alg:profit} before its maximum response delay $d_s$, if 
 we set $\epsilon_s = \frac{Q_i^{s(max)}+Z_i^{s(max)}}{d_s}$.
\end{theorem}

\vspace{1mm}
This theorem can be proved based on Lemma \ref{lemma:bounded-queue} and the $\epsilon$-persistence queue techniques \cite{neely-infocom11}. The condition on $\epsilon_s$ is to ensure that 
the queue lengths can grow to satisfy the job drop condition, \emph{i.e.}, $Q_i^s+Z_i^s(t) > V \xi_i^s$, if some jobs remain unscheduled in the last $d_s$ slots. Note that a cloud only drops jobs strategically, to balance the loss due to the job drop penalties and the gain in saving VMs for other jobs. \opt{long}{For more details, please refer to Appendix \ref{appendix:SLA}.}

\subsection{Optimality of Individual Profit and Social Welfare}
\label{sec:optanalysis}

\begin{theorem}[Individual Profit Optimality]\label{theorem:profit}
Let $\Omega_i^*$ be the offline optimum of time-averaged profit of cloud $i\in [1, F]$, obtained in a truthful, individual-rational, ex-post budget-balanced double auction, with complete information on its own job arrivals and prices in the entire time span $[0, T-1]$. The dynamic Algorithm \ref{alg:profit} can achieve a time-averaged profit $\Omega_i$ for cloud $i$ within a constant gap $B_i/V$ to $\Omega_i^*$, \emph{i.e.},

\vspace{-4mm}{\small
\begin{align*}
\Omega_i \geq \Omega_i^* - B_i/V,
\end{align*}}\vspace{-6mm}

\noindent where $V>0$ and {\small $B_i=\frac{1}{2}\sum_{s\in [1, S]}[[\sum_{j=1}^F C_j^{m_s}N_j^{m_s}/g_s+D_i^{s(max)}]^2 + [R_i^s]^2+[\epsilon_s]^2 + [D_i^{s(max)}+\sum_{j=1}^F C_j^{m_s}N_j^{m_s}/g_s]^2]$} is a constant.

\end{theorem}

\vspace{1mm}
The proof to this theorem is rooted in the Lyapunov optimization theory \cite{book2010}. The gap $B_i/V$ can be close to zero by fixing $\epsilon_s$  and increasing $V$. Detailed proof is included in \opt{short}{\cite{tech-report}}\opt{long}{Appendix \ref{appendix:profit}}.


\vspace{1mm}
\begin{theorem}[Social Welfare Optimality of Alg.~\ref{alg:social}]\label{theorem:welfare-alg2}
Let $\Pi^*$ be the offline optimum of the time-averaged social welfare in (\ref{eqn:social-max}), obtained with full information of the federation over the entire time span $[0, T-1]$. The time-averaged social welfare achieved by all clouds by running Alg.~\ref{alg:social}, approaches the offline-optimal social welfare $\Pi^*$, by a constant gap $B/V$, \emph{i.e.},


\vspace{-4mm}{\small
\begin{align*}
\Pi \geq \Pi^* - B/V,
\end{align*}}\vspace{-6mm}

\noindent where $V>0$ and $B=\sum_{i\in [1, F]}B_i$. $B_i$ is defined in Theorem \ref{theorem:profit}, $\forall i\in[1,F]$.
\end{theorem}

\vspace{1mm}
The proof to this theorem is also based on the Lyapunov optimization theory \cite{book2010}. The gap $B/V$ can be close to zero by fixing $\epsilon_s$  and increasing $V$. Detailed proof is included in \opt{short}{\cite{tech-report}}\opt{long}{Appendix \ref{appendix:social}}.

\vspace{1mm}
\begin{theorem}[Asymptotic Optimality in Social Welfare of Alg.~\ref{alg:profit}]\label{theorem:welfare}
Let $\Pi^*$ be the offline optimum of the time-averaged social welfare in (\ref{eqn:social-max}), obtained with full information of the federation over the entire time span $[0, T-1]$. Suppose all clouds are homogenous, \emph{i.e.}, with the same number of servers ($N_i^m$) and the same maximum per-server VM provisioning ($C_i^m$) for each VM type $m$, with i.i.d. service prices, job arrivals and operational costs. When the number of clouds, $F$, grows, the sum of time-averaged profits achieved by all clouds by running Alg.~\ref{alg:profit} under the double auction mechanism in Sec.~\ref{sec:auction_mechanism}, approaches the offline-optimal social welfare $\Pi^*$, by a constant gap $B/V$, \emph{i.e.},


\vspace{-4mm}{\small
\begin{align*}
\Pi \geq \Pi^* - B/V,
\end{align*}}\vspace{-6mm}

\noindent where $V>0$ and $B=\sum_{i\in [1, F]}B_i$. $B_i$ is defined in Theorem \ref{theorem:profit}, $\forall i\in[1,F]$.
\end{theorem}

\vspace{1mm}
To prove the theorem, we demonstrate that, when the number of clouds goes to infinity, the one-shot social welfare obtained with Alg.~\ref{alg:profit} is the same as that achieved by the dynamic benchmark Alg.~\ref{alg:social} in the same time slot. \opt{short}{Details are in \cite{tech-report}.}\opt{long}{Details are in Appendix \ref{appendix:welfare}.}




\section{Performance Evaluation}\label{sec:simulation}


\subsection{Simulation Setup}

We carry out trace-driven simulation studies based on Google cluster-usage data \cite{clusterdata:Wilkes2011}\cite{clusterdata:Reiss2011}, which record jobs submitted to the Google cluster, with information on their resource demands (CPU, RAM, etc.) and relative charges. 
 We translate the data into concrete job arrival rates, resource types and prices, to drive our simulations as follows.

We consider $24$ types of jobs ($<m_s, g_s, d_s>$), $6$ VM types ($m_s$) combined from $\{$\emph{small, median, large}$\}$ \emph{CPU} and $\{$\emph{small, large}$\}$ \emph{Memory}, and two \emph{SLA} levels ($d_s$), corresponding to a larger maximum respond delay and a smaller maximum response delay at half of the former. Each job requires either 1 VM or 2 VMs concurrently ($g_s$).

There are $10$ clouds in the federation. One time slot is $1$ hour. The number of servers in each cloud that provision VMs of each type ranges within $[800, 1000]$. Each server can provide $30$ \emph{small-memory} VMs or $10$ \emph{large-memory} VMs. The VM charge to the customer is decided by multiplying $g_s$ by the relative VM price in the Google data, and then by the unit VM price in the range of [0.05, 0.08] \$/h. The penalty for dropping a job is set to the maximum per-job VM charge in the system. Operational costs are set according to the electricity prices at $10$ different geographic locations provided in \cite{ferc}, which vary on a hourly basis. Each server consumes power at 1 KW/h. 

The number of job arrivals in each hour to the federation is set according to the cumulated job requests of each type submitted to the Google cluster during that hour, in the rough range of $[40000, 90000]$ requests per hour. We randomly assign each arrived job to one of the $10$ clouds, following a heavy-tailed distribution. 
 In operating the virtual queues, we set $\epsilon_s=1000$ for jobs requiring low response delay, and $\epsilon_s=500$ for those of long delays. The maximum number of job drops per hour is $1000$ for all job types.


For comparison purposes, we also implement a simpler heuristic algorithm for each cloud to bid in the double auction and to schedules its jobs/servers: The cloud decides a value for each unscheduled job in a queue as the penalty to drop it if the next time slot is the deadline for scheduling, or the charged price upon its arrival otherwise. The true values of buy/sell prices for a type-$m$ VM at this cloud are set to the same, as the largest average value of jobs in a queue, among all job queues requiring type-$m$ VM(s). The quantity of VMs in a buy-bid is set to the number of unscheduled jobs in the queue with the largest average value as computed above. The quantity of VMs in a sell-bid is the overall number of VMs of the type that the cloud can provide. All VMs purchased via the auction are used to serve jobs from the queue with the maximum average value. 
  A cloud maintains the minimum number of servers to support those jobs, and only drops a job when its maximum response delay is reached.
  
\vspace{-4mm}
\begin{figure}[!h]
\begin{center}
 \subfigure[$V=4\times 10^6$]
 { \includegraphics[angle=0,width=0.46\columnwidth]{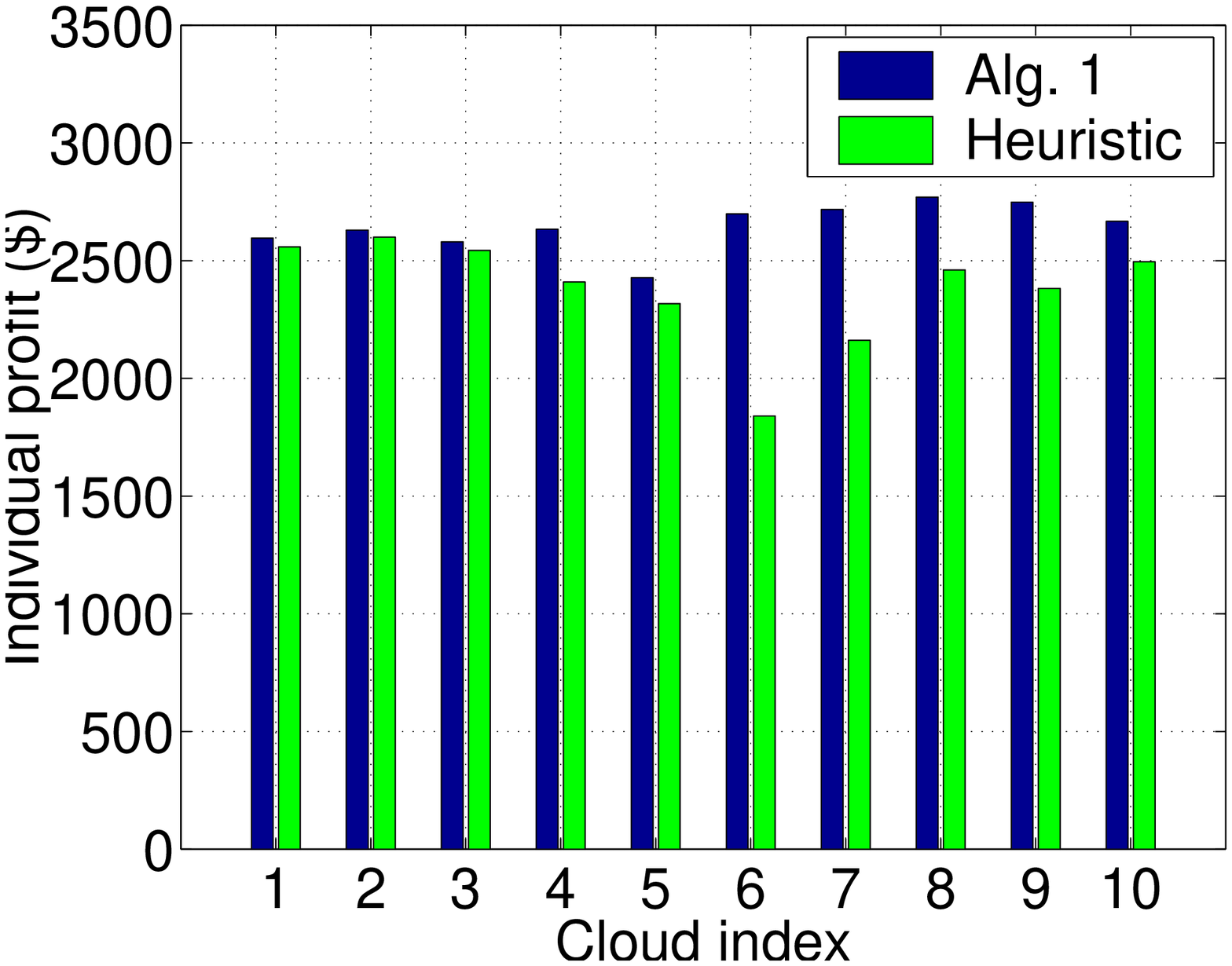}
   \label{fig:ps1}
 }
 \subfigure[$V=5\times 10^6$]
 { \includegraphics[angle=0,width=0.46\columnwidth]{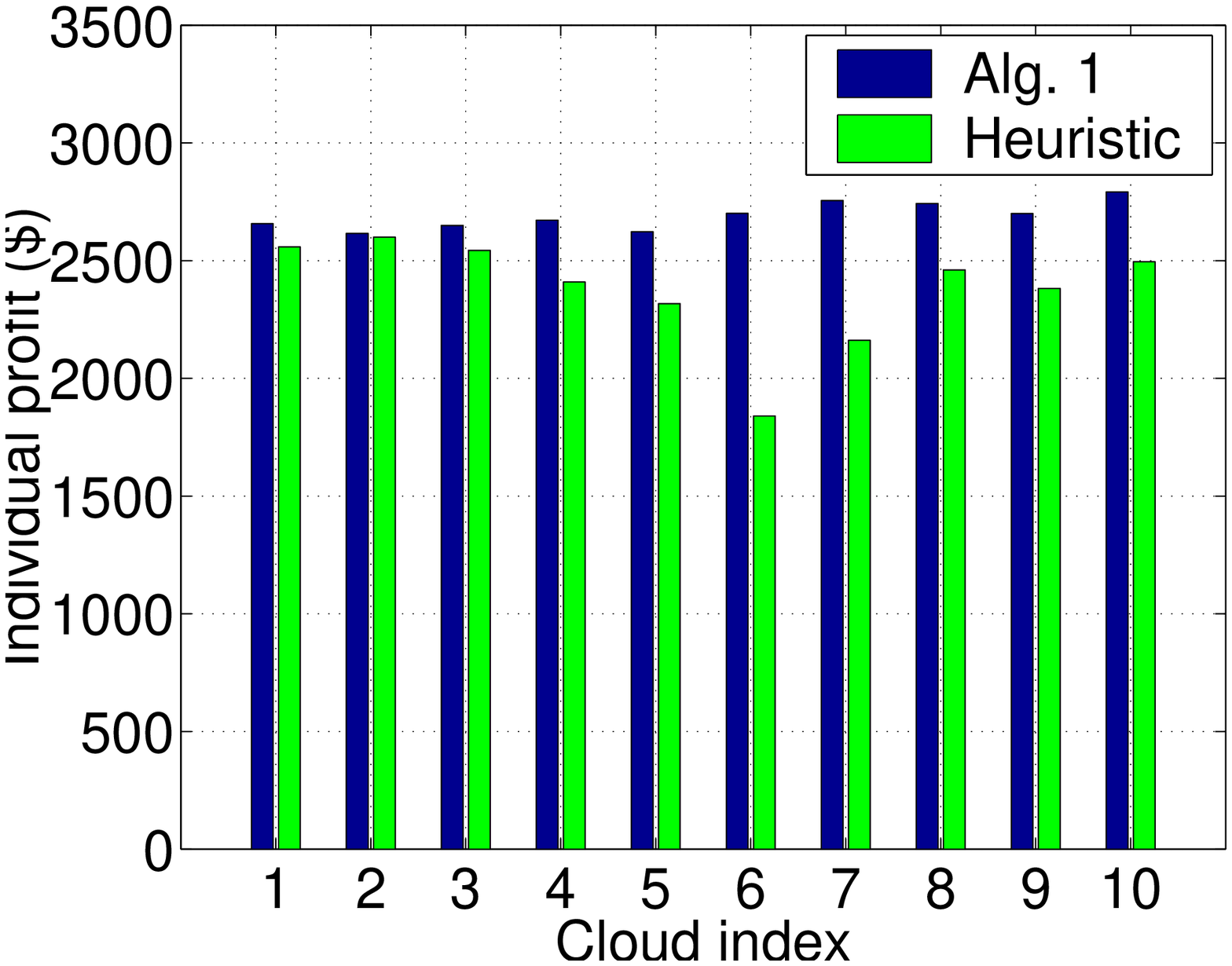}
   \label{fig:ps1}
 }
 \subfigure[$V=6\times 10^6$]
 { \includegraphics[angle=0,width=0.46\columnwidth]{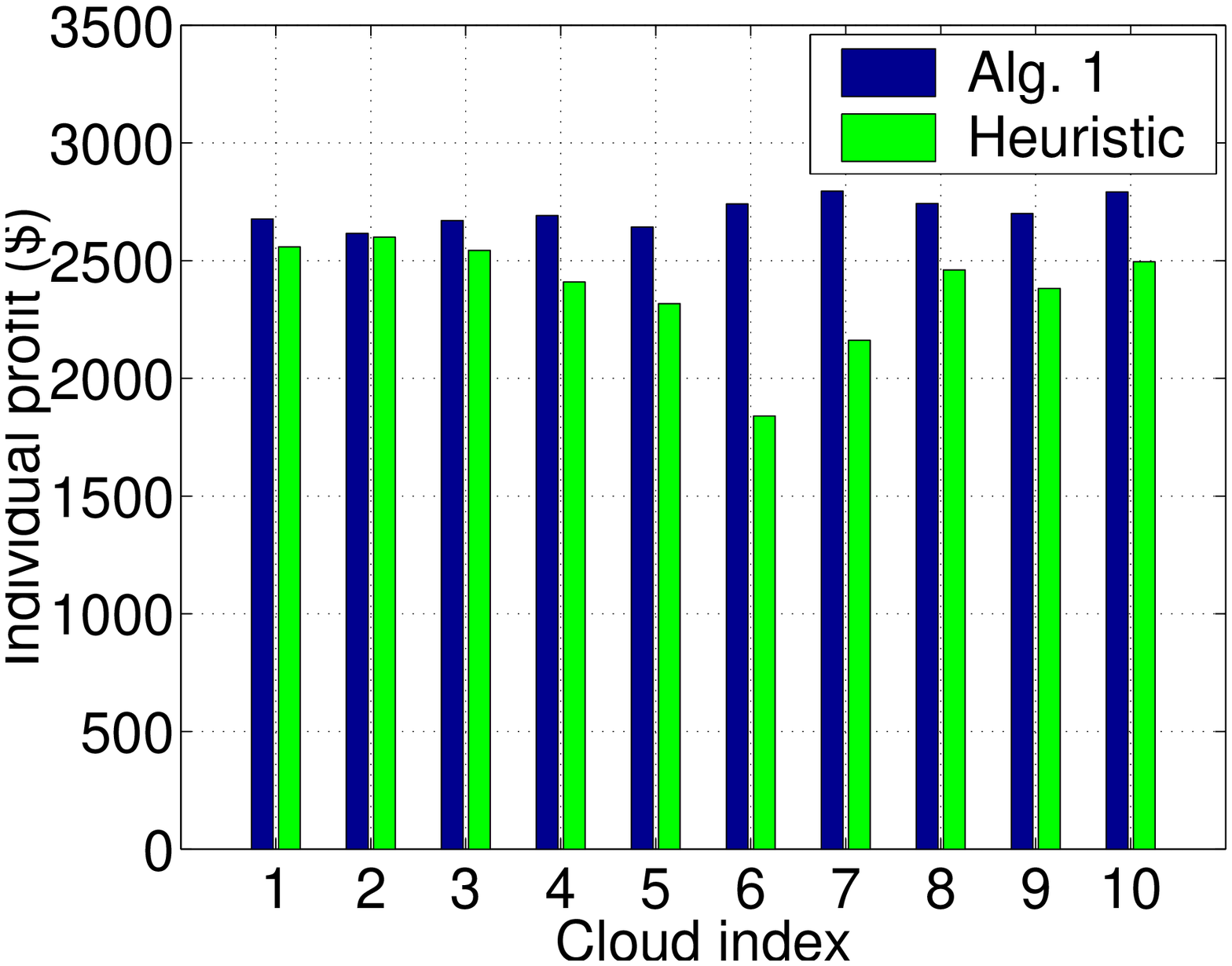}
   \label{fig:ps1}
 }
 \subfigure[$V=7\times 10^6$]
 { \includegraphics[angle=0,width=0.46\columnwidth]{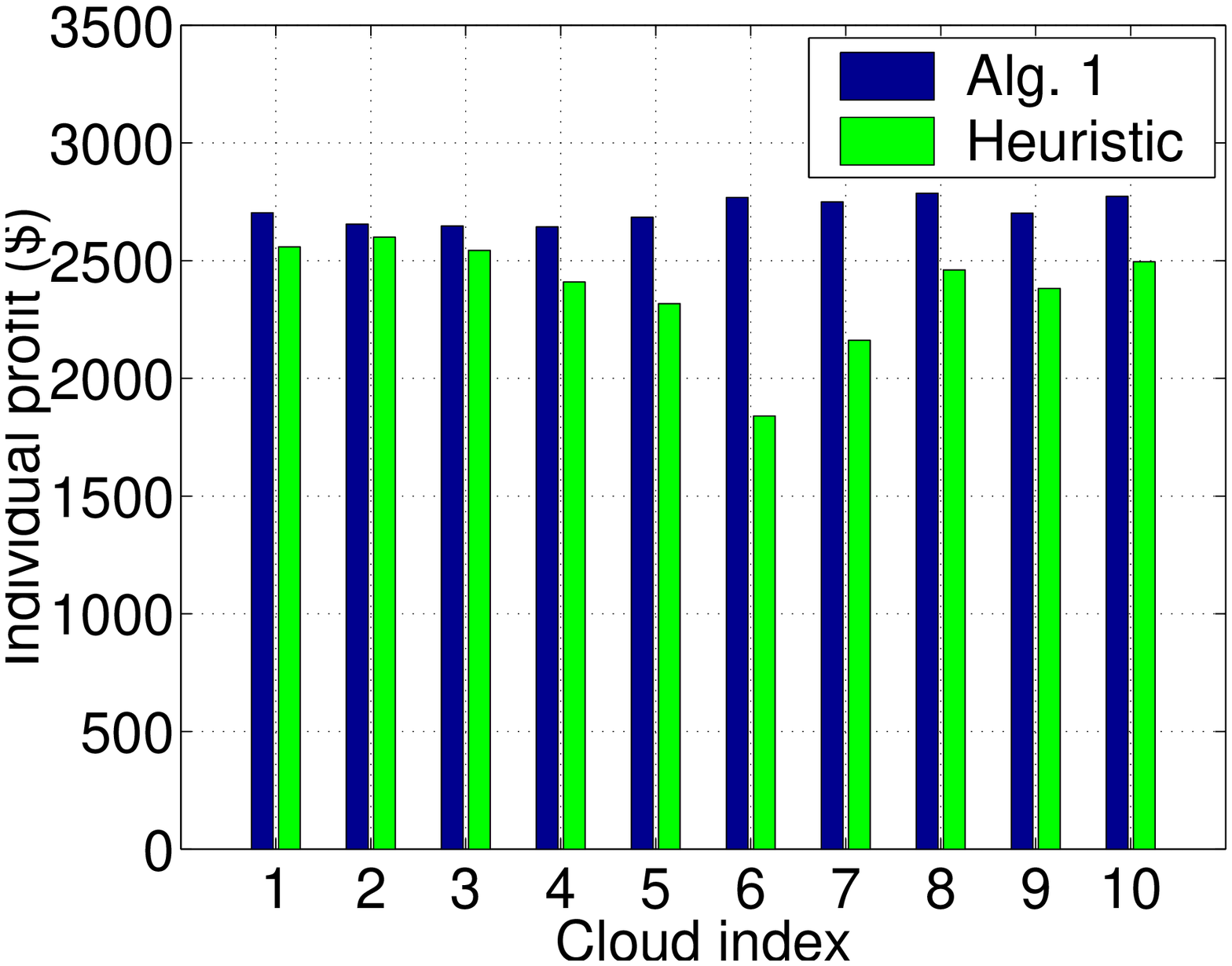}
   \label{fig:ps1}
 }
 \subfigure[$V=8\times 10^6$]
 { \includegraphics[angle=0,width=0.46\columnwidth]{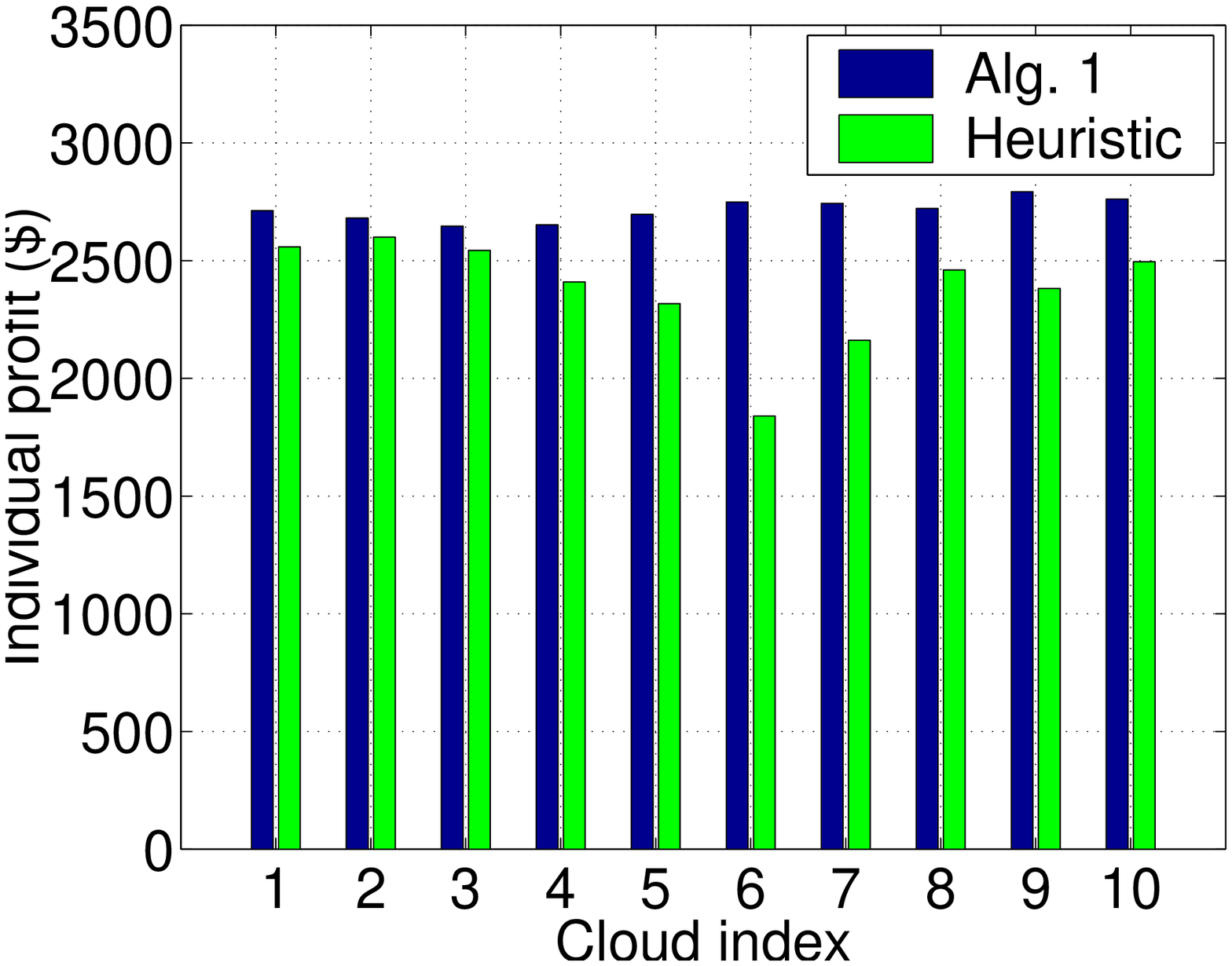}
   \label{fig:ps1}
 }
 \subfigure[$V=9\times 10^6$]
 { \includegraphics[angle=0,width=0.46\columnwidth]{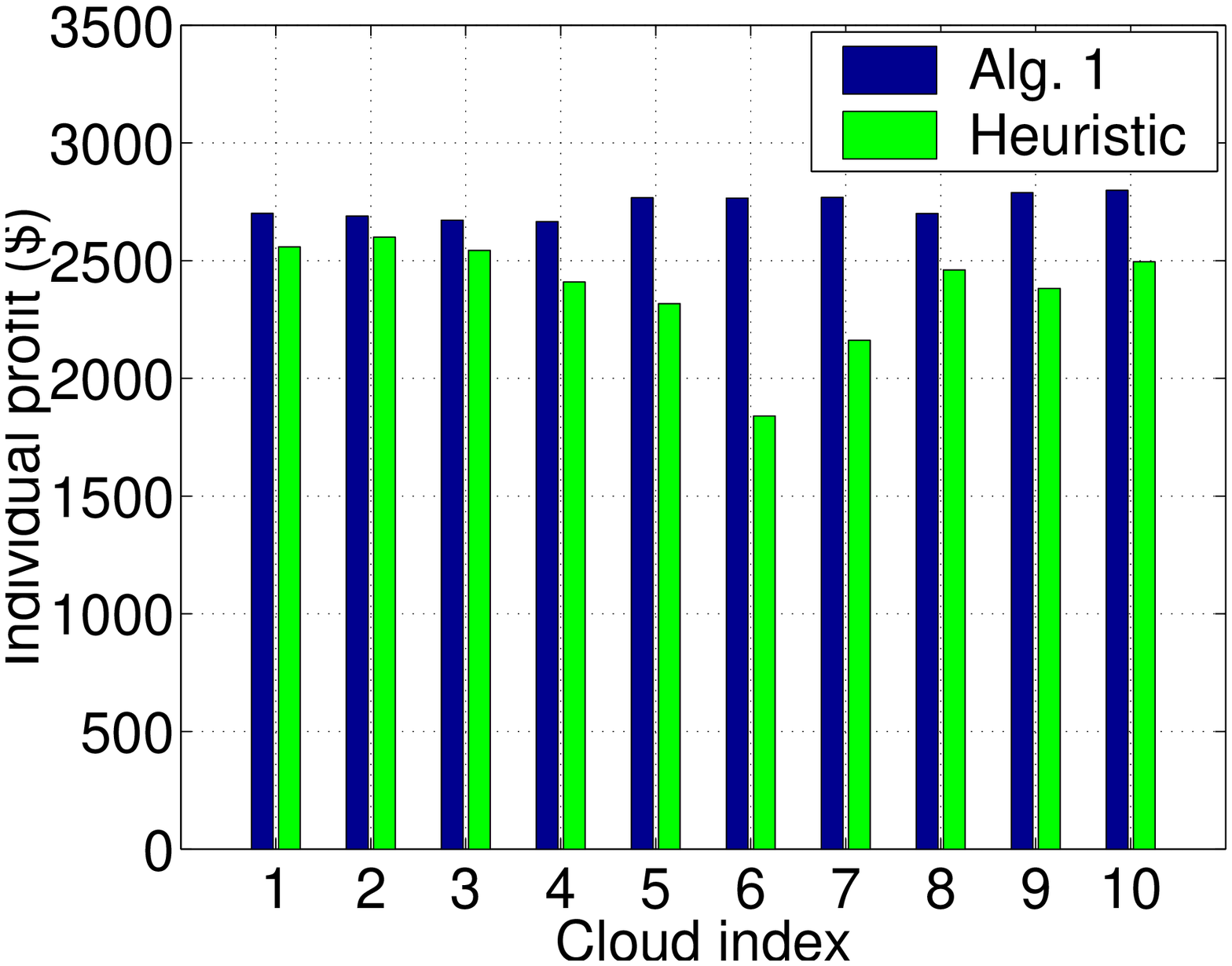}
   \label{fig:ps1}
 }\vspace{-3mm}
 \caption{Comparisons of individual profit with different values of $V$.
}
 \label{fig:profits}
\end{center}
\vspace{-8mm}
\end{figure}

\subsection{Individual Profit and Social Welfare}

We compare the time-averaged profit achieved at each cloud with our dynamic algorithm in Alg.~\ref{alg:profit} and with the heuristic algorithm, after the system has been running for $2000$ hours. \opt{long}{Fig.~\ref{fig:profits} shows that our algorithm can achieve a higher profit than the heuristic, at each of the $10$ clouds, when the value of $V$ is no less than $4\times 10^6$. 
 The observation is that when $V$ is larger, the individual profit with our algorithm is even better, since it is closer to the offline optimum.
 
 We next compare the social welfare achieved with Alg.~\ref{alg:profit}, the heuristic, and the dynamic benchmark Alg.~\ref{alg:social}.
 Fig.~\ref{fig:p2} shows that social welfare achieved with Alg.~\ref{alg:profit} is mostly within $7.7\%$ of that by the benchmark algorithm, even under our heterogenous settings. It outperforms the heuristic by $19.2\%$. The social welfare is larger at larger $V$'s in cases of both  Alg.~\ref{alg:profit} and the benchmark algorithm, verifying Theorems \ref{theorem:profit} and \ref{theorem:welfare} in that they approach the respective offline optimum when $V$ grows. 

}

\opt{short}{Fig.~\ref{fig:p1} shows that our algorithm can achieve a higher profit than the heuristic, at each of the $10$ clouds. More comparisons with different values of $V$ can be found in \cite{tech-report}, 
 with the observation that when $V$ is larger, the individual profit is even better, since it is closer to the offline optimum.

\vspace{-4mm}
\begin{figure}[H]
\begin{center}
 \subfigure[Individual Profit with $V=10^7$]
 { \includegraphics[angle=0,width=0.46\columnwidth]{profit}
   \label{fig:p1}
 }
 \subfigure[Social Welfare]
 { \includegraphics[angle=0,width=0.46\columnwidth]{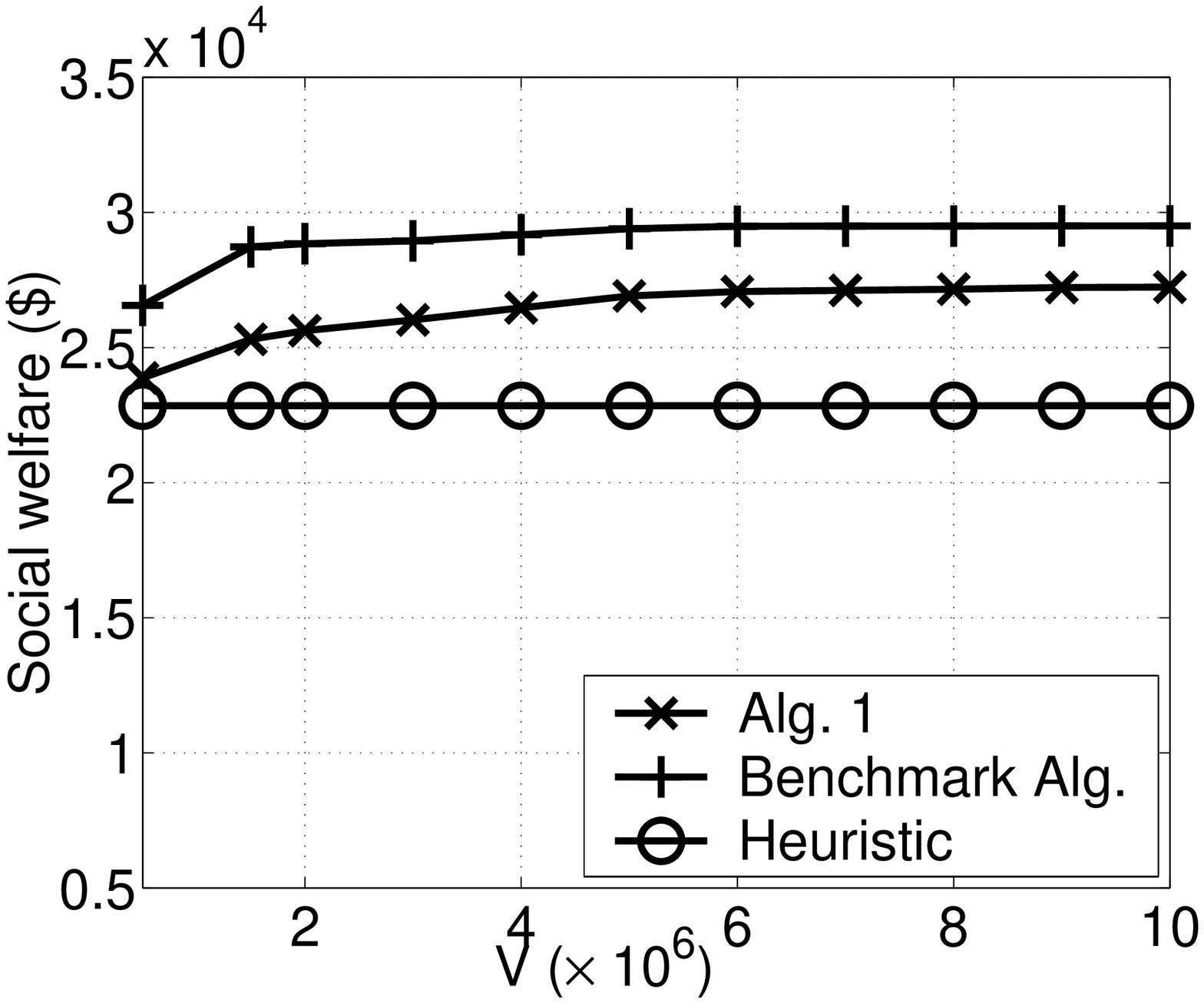}
   \label{fig:p2}
 }\vspace{-3mm}
 \caption{Comparisons of individual profit and social welfare.
}
 \label{fig:profit-welfare}
\end{center}
\vspace{-6mm}
\end{figure}
}

\opt{long}{
\vspace{-4mm}
\begin{figure}[H]
\begin{center}
{ \includegraphics[angle=0,width=0.46\columnwidth]{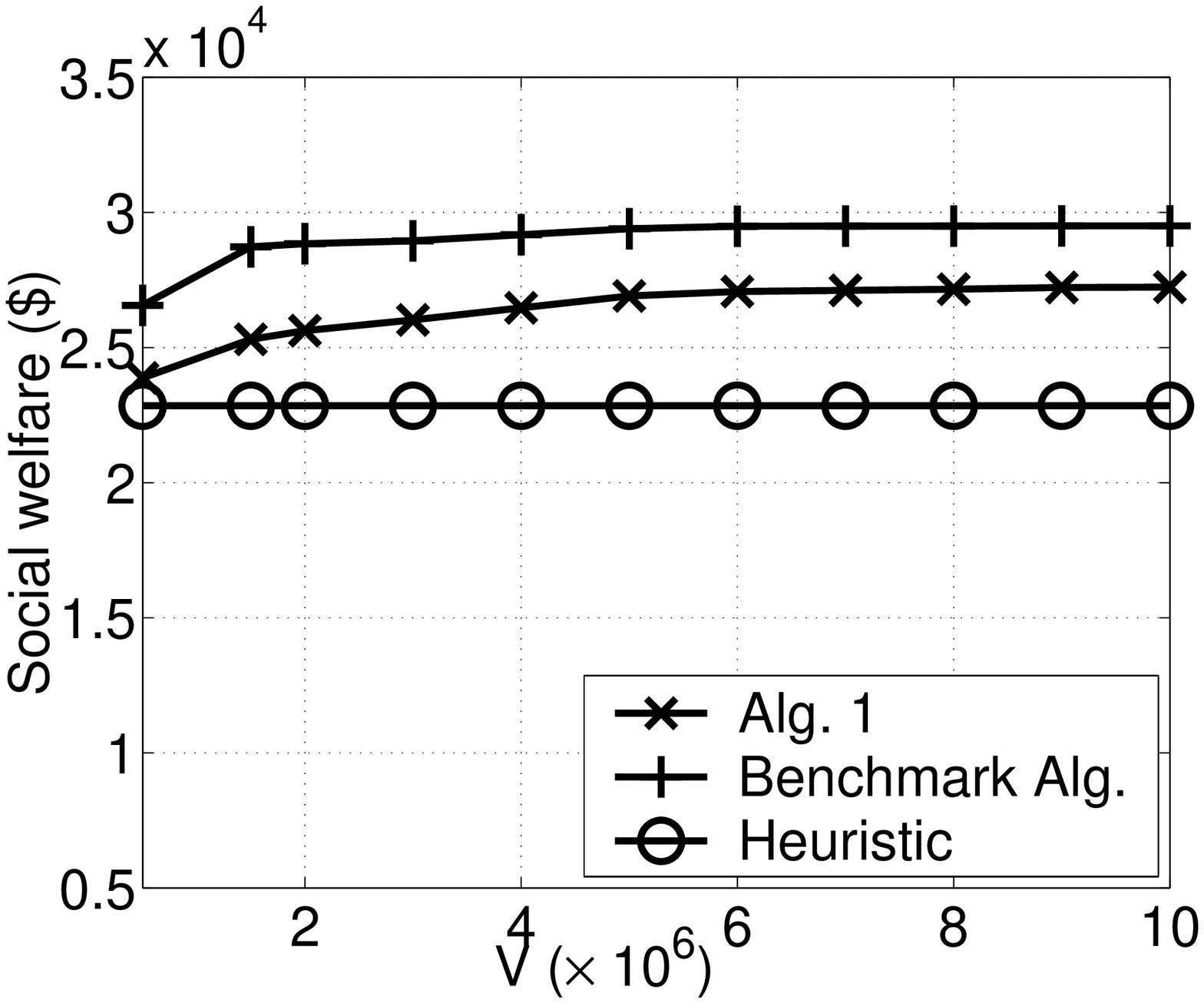}

 }\vspace{-3mm}
 \caption{Comparisons of social welfare.
}
\label{fig:p2}
\end{center}
\vspace{-6mm}
\end{figure}
}

%

\subsection{Response Delay and Job Drop}

We next investigate the scheduling delays experienced by jobs. In our system, a maximum response delay is set as the SLA objective for each type of jobs. Here, we study the average response delay actually experienced by the jobs, when the longer maximum response delay is set to different values. 
Fig.~\ref{fig:w1} shows that both Alg.~\ref{alg:profit} and the benchmark algorithm incur a low average response delay (well ahead of scheduling deadlines), as compared to that of the heuristic. The reasons are: i) the heuristic algorithm always greedily keeps jobs in queues for future scheduling until near the deadline; and ii) both Alg.~\ref{alg:profit} and the benchmark algorithm evaluate the scheduling urgency better than the heuristic does, such that jobs are tended to be served well before the deadlines.

\vspace{-4mm}\begin{figure}[H]
\begin{center}
 \subfigure[Average response delay]
 { \includegraphics[angle=0,width=0.46\columnwidth]{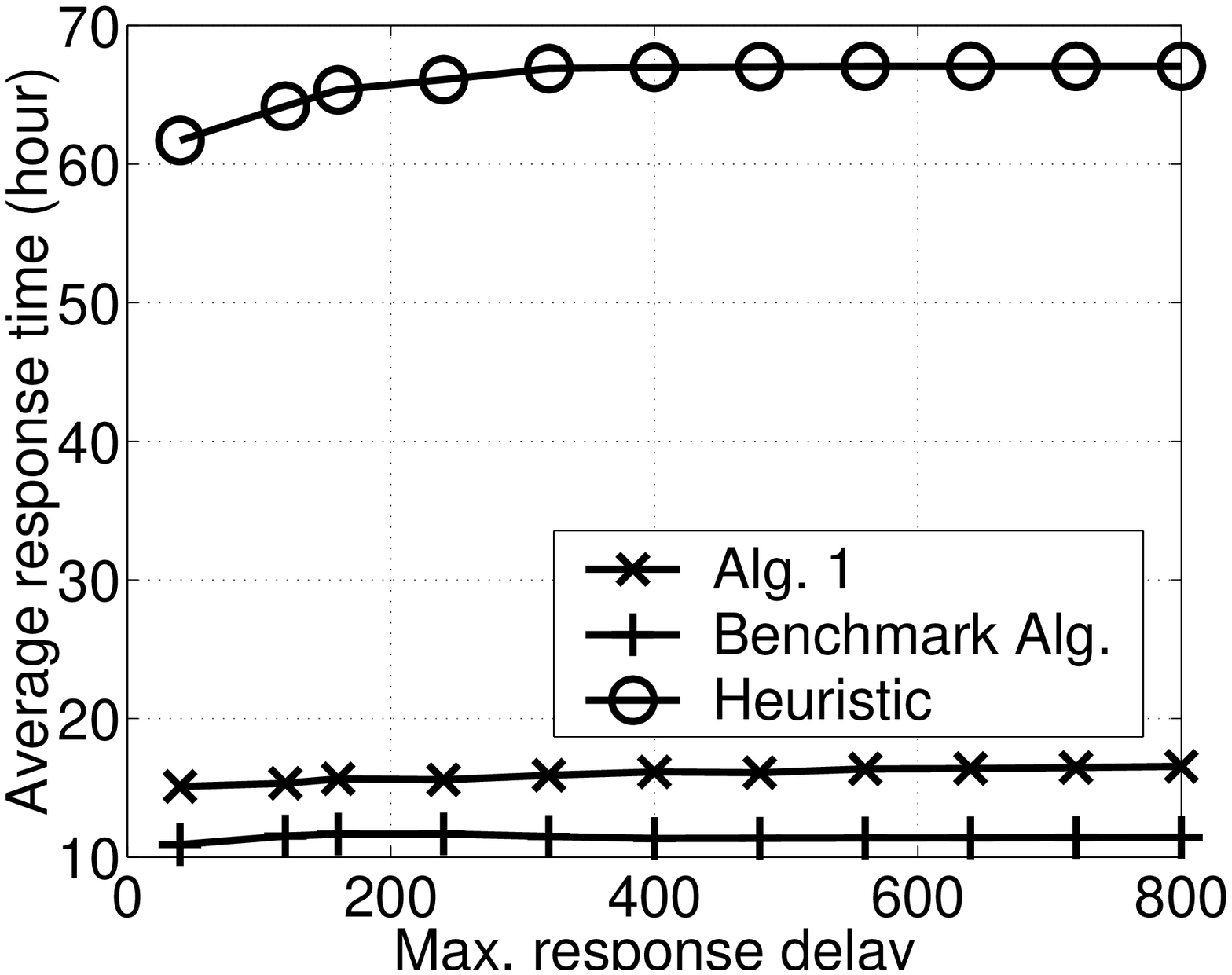}
   \label{fig:w1}
 }
 \subfigure[Averaged job-drop percentage]
 { \includegraphics[angle=0,width=0.46\columnwidth]{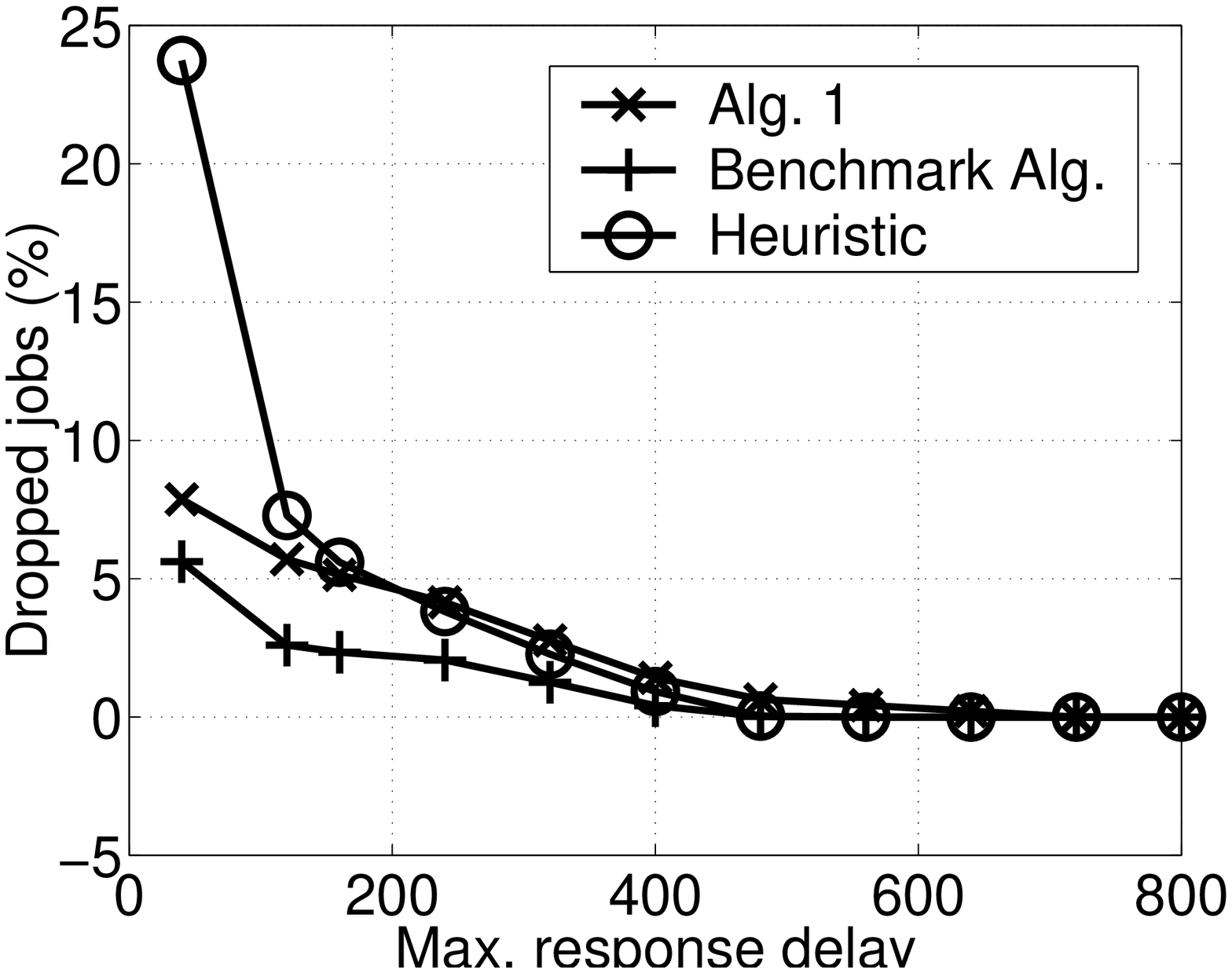}
   \label{fig:w2}
 }\vspace{-3mm}
 \caption{Comparisons of average job scheduling delay and drop percentage.}
 \label{fig:delay-drop}
\end{center}
\vspace{-6mm}
\end{figure}

We also study the percentage of admitted jobs in the entire federation that are eventually dropped with the three algorithms. Fig.~\ref{fig:w2} reveals that the drop rate decreases quickly with the increase of the allowed maximum response delay, and Alg.~\ref{alg:profit} and the benchmark algorithm again outperform the heuristic, due to their well-designed scheduling strategies. 
\section{Conclusion}\label{sec:conclusion} \vspace{-1mm}

This paper investigates both individual-profit maximizing and social-welfare efficient strategies at individual selfish clouds in a cloud federation, in VM trades across cloud boundaries. We tailor a truthful, individual-rational, ex-post budget-balanced double auction as the inter-cloud trading mechanism, and design a dynamic algorithm for each cloud to decide the best VM valuation and bidding strategies, and to schedule job service/drop and server provisioning in the most economic fashion, under time-varying job arrivals and operational costs. The proposed algorithm can obtain a time-averaged profit for each cloud within a constant gap to its offline maximum, as well as a close-to-optimum social welfare in the entire federation, based on both solid theoretical analysis and trace-driven simulation studies under realistic setting.   
 As future work, we are interested in broadening our investigations to front-end job pricing and competition for customers among the clouds, and the connection between front-end charging strategies and inter-cloud trading strategies in a cloud federation.

\vspace{-4mm}

\bibliographystyle{IEEEtran}
\bibliography{infocom13-hxli}

\begin{thebibliography}{10}
\providecommand{\url}[1]{#1}
\csname url@samestyle\endcsname
\providecommand{\newblock}{\relax}
\providecommand{\bibinfo}[2]{#2}
\providecommand{\BIBentrySTDinterwordspacing}{\spaceskip=0pt\relax}
\providecommand{\BIBentryALTinterwordstretchfactor}{4}
\providecommand{\BIBentryALTinterwordspacing}{\spaceskip=\fontdimen2\font plus
\BIBentryALTinterwordstretchfactor\fontdimen3\font minus
  \fontdimen4\font\relax}
\providecommand{\BIBforeignlanguage}[2]{{%
\expandafter\ifx\csname l@#1\endcsname\relax
\typeout{** WARNING: IEEEtran.bst: No hyphenation pattern has been}%
\typeout{** loaded for the language `#1'. Using the pattern for}%
\typeout{** the default language instead.}%
\else
\language=\csname l@#1\endcsname
\fi
#2}}
\providecommand{\BIBdecl}{\relax}
\BIBdecl

\bibitem{IBM09}
B.~Rochwerger, D.~Breitgand, E.~L. E, A.~Galis, K.~Nagin, I.~Llorente,
  R.~Montero, Y.~Wolfsthal, E.~Elmroth, J.~Caceres, M.~Ben-Yehuda, W.~Emmerich,
  and F.~Gal¨¢n, ``The reservoir model and architecture for open federated
  cloud computing,'' \emph{IBM Journal of Research and Development}, vol.~54,
  pp. 535 -- 545, 2009.

\bibitem{gcc09}
E.~Elmroth and L.~Larsson, ``Interfaces for placement , migration, and
  monitoring of virtual machines in federated clouds.'' in \emph{Proc.~of IEEE
  Computer Society GCC'09}, 2009.

\bibitem{rao-infocom10}
L.~Rao, X.~Liu, L.~Xie, and W.~Liu, ``Minimizing electricity cost: Optimization
  of distributed internet data centers in a multi-electricity-market
  environment,'' in \emph{Prof.~of IEEE INFOCOM'10}, 2010.

\bibitem{ren-icdcs12}
S.~Ran, Y.~He, and F.~Xu, ``Provably-efficient job scheduling for energy and
  fairness in geographically distributed data centers,'' in \emph{Prof.~of IEEE
  ICDCS'12}, 2012.

\bibitem{noms10}
R.~Urgaonkar, U.~Kozat, K.~Igarashi, and M.~Neely, ``Dynamic resource
  allocation and power management in virtualized data centers,'' in
  \emph{Prof.~of IEEE/IFIP NOMS'10}, 2010.

\bibitem{yao-infocom12}
Y.~Yao, L.~Huang, A.~Sharma, L.~Golubchik, and M.~Neely, ``Data centers power
  reduction: A two time scale approach for delay tolerant workloads,'' in
  \emph{Proc.~of IEEE INFOCOM'12}, 2012.

\bibitem{Buyya-hpc00}
R.~Buyya, D.~Abramson, and J.~Giddy., ``Nimrod/g: An architecture of a resource
  management and scheduling system in a global computational grid,'' in
  \emph{Proc. of HPC Asia'00}, 2000.

\bibitem{zhou-infocom09}
X.~Zhou and H.~Zheng, ``Trust: A general framework for truthful double spectrum
  auctions,'' in \emph{Proc.~ of IEEE INFOCOM'09}, 2009.

\bibitem{xu-infocom10}
H.~Xu, J.~Jin, and B.~Li, ``A secondary market for spectrum,'' in \emph{Proc.~
  of IEEE INFOCOM'10, Mini Conference}, 2010.

\bibitem{amazon}
\BIBentryALTinterwordspacing
 [Online]. Available: \url{http://aws.amazon.com/ec2}
\BIBentrySTDinterwordspacing

\bibitem{Marian-CCGrid10}
M.~Mihailescu and Y.~M. Teo, ``Dynamic resource pricing on federated clouds,''
  in \emph{Proc.~of IEEE/ACM CCGrid'10}, 2010.

\bibitem{Marian-CCGrid12}
------, ``The impact of user rationality in federated clouds,'' in
  \emph{Proc.~of IEEE/ACM CCGrid'12}, 2012.

\bibitem{gomes-exchange12}
E.~R. Gomes, Q.~B. Vo, and R.~Kowalczyk, ``Pure exchange markets for resource
  sharing in federated clouds,'' \emph{Concurrency Computat.: Pract. Exper.},
  vol.~24, pp. 977 -- 991, 2012.

\bibitem{linode}
\BIBentryALTinterwordspacing
 [Online]. Available: \url{http://www.linode.com/faq.cfm}
\BIBentrySTDinterwordspacing

\bibitem{neely-infocom11}
M.~J. Neely, ``Opportunistic scheduling with worst case delay guarantees in
  single and multi-hop networks,'' in \emph{Proc.~of IEEE INFOCOM'11}, 2011.

\bibitem{energy09}
U.~Hoelzle and L.~A. Barroso, \emph{The Datacenter as a Computer: An
  Introduction to the Design of Warehouse-Scale Machines.}\hskip 1em plus 0.5em
  minus 0.4em\relax Morgan \& Claypool, 2009.

\bibitem{ferc}
\BIBentryALTinterwordspacing
 [Online]. Available: \url{www.ferc.gov}
\BIBentrySTDinterwordspacing

\bibitem{myerson-1983}
R.~B. Myerson and M.~A. Satterthwaite, ``Efficient mechanisms for bilateral
  trading,'' \emph{Journal of Economics Theory}, vol.~29, pp. 265--281, 1983.

\bibitem{book2010}
M.~J. Neely, \emph{Stochastic Network Optimization with Application to
  Communication and Queueing Systems}, J.~Walrand, Ed.\hskip 1em plus 0.5em
  minus 0.4em\relax Morgan\&Claypool Publishers, 2010.

\bibitem{sandholm-IJICAI01}
T.~SANDHOLM and S.~SURI, ``Market clearability,'' in \emph{Proc.~of IJICAI'01},
  2001.

\bibitem{clusterdata:Wilkes2011}
J.~Wilkes, ``More {Google} cluster data,'' Nov. 2011, \uppercase{URL}:
  \url{http://googleresearch.blogspot.com/2011/11/more-google-cluster-data.htm%
l}.

\bibitem{clusterdata:Reiss2011}
C.~Reiss, J.~Wilkes, and J.~L. Hellerstein, ``{Google} cluster-usage traces:
  format + schema,'' Google Inc., Tech. Rep., 2011, revised 2012.03.20. URL:
  \url{http://code.google.com/p/googleclusterdata/wiki/TraceVersion2}.

\bibitem{huang-CI02}
P.~Huang, A.~Scheller-Wolf, and K.~Sycara, ``Design of a multi-unit double
  auction e-market,'' \emph{Computational Intelligence}, vol.~18, pp. 256--617,
  2002.

\end{thebibliography}

\opt{long}{
\begin{appendices}

\vspace{-2mm}\section{Derivation of the one-shot optimization problem for individual profit maximization 
}\label{appendix:drift}

Squaring the queuing laws (\ref{eqn:queue1}) and (\ref{eqn:queue2}), we can derive the following inequality

\vspace{-4mm}{\small
\begin{align*}
\Delta(\Theta_i(t))\leq& \frac{1}{2}\sum_{s\in [1, S]}[[\sum_{j=1}^F \mu_{ij}^s(t)+D_i^s(t)]^2 + [r_i^s(t)]^2 + 2Q_i^s(t)[r_i^s(t)\\
&-\sum_{j=1}^F \mu_{ij}^s(t)-D_i^s(t)] + [\mathbf{1}_{\{Q_i^s(t)>0\}}\epsilon_s]^2 + [D_i^s(t)\\
&+ \mathbf{1}_{\{Q_i^s(t)=0\}}\sum_{j=1}^F C_j^{m_s}N_j^{m_s}/g_s\\ &+ \mathbf{1}_{\{Q_i^s(t)>0\}}\sum_{j=1}^F \mu_{ij}^s(t)]^2\\
& + 2Z_i^s(t)[\mathbf{1}_{\{Q_i^s(t)>0\}}[\epsilon_s-\sum_{j=1}^F \mu_{ij}^s(t)]-D_i^s(t)\\
&-\mathbf{1}_{\{Q_i^s(t)=0\}}\sum_{j=1}^F C_j^{m_s}N_j^{m_s}/g_s]]\\
\leq & \frac{1}{2}\sum_{s\in [1, S]}[[\sum_{j=1}^F C_j^{m_s}N_j^{m_s}/g_s+D_i^{s(max)}]^2 + [R_i^s]^2\\
&+ 2Q_i^s(t)[r_i^s(t)-\sum_{j=1}^F \mu_{ij}^s(t)-D_i^s(t)]\\
&+ [\epsilon_s]^2 + [D_i^{s(max)}+\sum_{j=1}^F C_j^{m_s}N_j^{m_s}/g_s]^2 \\
& + 2Z_i^s(t)[\epsilon_s-\sum_{j=1}^F \mu_{ij}^s(t)-D_i^s(t)]]\\
=& B_i + \sum_{s\in [1, S]}[Q_i^s(t)[r_i^s(t)-\sum_{j=1}^F \mu_{ij}^s(t)-D_i^s(t)]\\ &+ Z_i^s(t)[\epsilon_s-\sum_{j=1}^F \mu_{ij}^s(t)-D_i^s(t)]],
\end{align*}}\vspace{-4mm}

\noindent where {\small $B_i=\frac{1}{2}\sum_{s\in [1, S]}[[\sum_{j=1}^F C_j^{m_s}N_j^{m_s}/g_s+D_i^{s(max)}]^2 + [R_i^s]^2+[\epsilon_s]^2 + [D_i^{s(max)}+\sum_{j=1}^F C_j^{m_s}N_j^{m_s}/g_s]^2]$}.

By applying the drift-plus-penalty framework (or equivalently, drift-minus-profit here), we subtract the weighted one-shot individual profit of cloud $i$ in time $t$, \emph{i.e.}, $V\cdot [\sum_{m\in [1, M]} [\hat{s}_i^m(t) \hat{\eta}_i^m(t) - \hat{b}_i^m(t) \hat{\gamma}_i^m(t)-\beta_i(t) n_i^m(t)]+ \sum_{s\in [1, S]}[p_i^s(t)\cdot r_i^s(t)- D_i^s(t) \xi_i^s]]$, on both sides of the above inequality. Hence, we have the following inequality:


\vspace{-5mm}{\small\begin{align*}
&\Delta(\Theta_i(t))-V\cdot [\sum_{m\in [1, M]} [\hat{s}_i^m(t) \hat{\eta}_i^m(t) - \hat{b}_i^m(t) \hat{\gamma}_i^m(t)-\beta_i(t) n_i^m(t)]\notag \\ &+ \sum_{s\in [1, S]}[p_i^s(t)\cdot r_i^s(t)- D_i^s(t) \xi_i^s]] \notag\\
\leq&  B_i + \sum_{s\in [1, S]}[Q_i^s(t) r_i^s(t) + Z_i^s(t) \epsilon_s - V p_i^s(t)\cdot r_i^s(t)] \notag\\
&- \varphi_1^i(t)-\varphi_2^i(t)-\varphi_3^i(t),
\end{align*}}\vspace{-5mm}

\noindent where $V>0$ is a user-defined positive constant that can be understood as the
weight of profit in the expression.

\opt{short}{

\section{Derivation of dynamic algorithm}\label{appendix:algorithm}

The maximization problem in Eqn.~(\ref{eqn:profit-oneslot}) can be decoupled into two independent optimizations:

\vspace{-4mm}{\small
\begin{align}
\max&~~~~\varphi_1^i(t)+\varphi_2^i(t)\label{eqn:profit-oneshot1}\\
s.t.&~~~~\text{Constraint (\ref{eqn:capacity1})-(\ref{eqn:capacity3})}.\notag
\end{align}}\vspace{-4mm}

\noindent and,

\vspace{-4mm}{\small
\begin{align}
\max&~~~~\varphi_3^i(t)\label{eqn:profit-oneshot2}\\
s.t.&~~~~\text{Constraint (\ref{eqn:drop-cons})}.\notag
\end{align}}\vspace{-4mm}

We first solve problem (\ref{eqn:profit-oneshot1}) to derive the optimal solutions to VM valuation \& bid, job scheduling and server provisioning. Then, we find solution to job dropping with problem (\ref{eqn:profit-oneshot2}).

\vspace{1mm}
\noindent \textbf{Problem (\ref{eqn:profit-oneshot1})}: Before solving the optimization problem, some useful properties of the VM valuation \& bid are discussed based on the truthfulness, individual rationality and ex-post budget balance of the given double auction scheme.

\vspace{1mm}
\noindent \emph{Property 1. } Each buyer/seller bids with the true values, \emph{i.e.}, $b_i^m = \tilde{b}_i^m(t)$, $s_i^m(t)=\tilde{s}_i^m(t)$, $\gamma_i^m(t)=\tilde{\gamma}_i^m(t)$ and $\eta_i^m(t)=\tilde{\eta}_i^m(t)$. This is based on the truthfulness of the double auction mechanism, for which bidding truthfully is always the dominant strategy.

\vspace{1mm}
\noindent \emph{Property 2. } Each buyer/seller pays/charges a price that is no higher/lower than the true value while the number of traded VMs is no larger than the maximum bided value if he/she wins, \emph{i.e.}, $\hat{b}_i^m(t)\leq \tilde{b}_i^m(t)$, $\hat{s}_i^m(t)\geq \tilde{s}_i^m(t)$, $\hat{\gamma}_i^m(t)\leq \tilde{\gamma}_i^m(t)$ and $\hat{\eta}_i^m(t)\leq \tilde{\eta}_i^m(t)$. This is because of the individual rationality of the double auction mechanism, such that the utility improvement obtained from the auction for each cloud is non-negative.

\vspace{1mm}
\noindent \emph{Property 3. } If a cloud wins with both buy-bid and sell-bid, he/she cannot buy a VM with a price strictly higher than that to sell, \emph{i.e.}, we must have $\hat{s}_i^m(t)\geq  \hat{b}_i^m(t)$. Otherwise, there will be a positive utility loss at the cloud by self-trading its own VMs, which violates the individual rationality.

\vspace{1mm}
\noindent \emph{Property 4. } If a cloud wins with both buy-bid and sell-bid, he/she cannot sell a VM with a price strictly higher than that to buy, \emph{i.e.}, we must have $\hat{s}_i^m(t)\leq \hat{b}_i^m(t)$. Otherwise, the auctioneer has to pay a positive price to that cloud in order to compensate the price difference for those inter-cloud traded VMs, which contradicts with the ex-post budget balance at the auctioneer.

\vspace{1mm}
\noindent \emph{Property 5. } With the combination of Property 4 and 5, we have that, if a cloud wins with both buy-bid and sell-bid, he/she buys and sells at the same price, \emph{i.e.}, $\hat{s}_i^m(t)=\hat{b}_i^m(t)$.

\vspace{1mm}
Next, each cloud $i\in [1, F]$ utilizes the above properties to solve the problem (\ref{eqn:profit-oneshot1}).

\vspace{1mm}
\noindent \emph{Step 1 -- Solve $n_i^m(t)$, $\forall m\in [1, M]$}: We start with solving $n_i^m(t)$ by assuming already known feasible assignments to $\hat{s}_i^m(t)$, $\hat{\eta}_i^m(t)$, $\hat{b}_i^m(t)$, $\hat{\gamma}_i^m(t)$, $\alpha_{ij}^m(t)$ and $\mu_{ij}^s(t)$. Next, in steps 2 and 3, we find optimal solutions to other variables iteratively, with values independent of $n_i^m(t)$. In this case, the problem (\ref{eqn:profit-oneshot1}) is equivalent to solve the following minimization problem,

\vspace{-4mm}{\small
\begin{align}
\min&~~~~V \beta_i(t) \sum_{m\in [1, M]}n_i^m(t)\label{eqn:profit-oneshot3}\\
s.t.&~~~~\text{Constraint (\ref{eqn:capacity1}), (\ref{eqn:capacity2}) and (\ref{eqn:capacity3})},\notag
\end{align}}\vspace{-4mm}

Apparently, for each VM type $m$, the best strategy is to assign the minimal feasible value to $n_i^m(t)$ by satisfying constraints (\ref{eqn:capacity1}) and (\ref{eqn:capacity3}), which can be combined into

\vspace{-4mm}{\small
\begin{align*}
\sum_{s\in [1,S], m_s=m}\mu_{ii}^s(t)\cdot g_s + \sum_{j\neq i}\alpha_{ji}^m(t) \leq C_i^m n_i^m(t).
\end{align*}}\vspace{-4mm}

\noindent Hence, we have the optimal solution to $n_i^m(t)$ as in Eqn.~(\ref{eqn:server-provision}) in Sec.~\ref{sec:algorithm}.

\vspace{1mm}
\noindent \emph{Step 2 -- Solve $\mu_{ij}^s(t)$, $\forall j\in [1,F], s\in [1, S]$}: In this step, we assume already known assignments to $\hat{s}_i^m(t)$, $\hat{\eta}_i^m(t)$, $\hat{b}_i^m(t)$, $\hat{\gamma}_i^m(t)$ and $\alpha_{ij}^m(t)$ while $n_i^m(t)$ is defined as in Eqn.~(\ref{eqn:server-provision}). Then, the problem (\ref{eqn:profit-oneshot1}) is equivalent to solve the following optimization,

\vspace{-4mm}{\small
\begin{align}
\max&~~~~\varphi_2^i(t)-V \beta_i \sum_{s\in [1,S], m_s=m}\mu_{ii}^s(t)\cdot \frac{g_s}{C_i^m} \label{eqn:profit-oneshot3}\\
s.t.&~~~~\text{Constraint (\ref{eqn:capacity1}), (\ref{eqn:capacity2}) and (\ref{eqn:capacity3})},\notag
\end{align}}\vspace{-4mm}

\noindent with $\mu_{ij}^s(t)$ as the controlling variables.

This is a maximum weight scheduling problem, with $Q_i^s(t)+Z_i^s(t)$ as the per-job scheduling weight for each $\mu_{ij}^s(t)$ ($j\neq i$) and $Q_i^s(t)+Z_i^s(t)-\frac{V\beta_i(t)  g_s}{C_i^m}$ as the per-job scheduling weight for each $\mu_{ii}^s(t)$. There are two cases:

\vspace{1mm}
\noindent -- $j\neq i$: $\mu_{ij}^s(t)$ is only constrained by $\alpha_{ij}^{m_s}(t)$, which is the number of bought VMs of type $m$ from cloud $j$. Based on constraint (\ref{eqn:capacity3}), we have that, with known $\alpha_{ij}^m(t)$, $\forall j\neq i, m\in [1, M]$, the best strategy is to assign full capacity of $\alpha_{ij}^m(t)$ to the service type $s_m^*$ with the maximum per-VM scheduling weight $\frac{Q_i^s(t)+Z_i^s(t)}{g_s}$ as defined in Eqn.~(\ref{eqn:weight}) and (\ref{eqn:max-weight}) in Sec.~\ref{sec:algorithm}. Then, we have the optimal solution to $\mu_{ij}^s(t)$, $\forall j\in [1, F], j\neq i, s\in [1,S]$ with Eqn.~(\ref{eqn:schedule2}).

\vspace{1mm}
\noindent -- $j = i$: In this case, by combining the constraints (\ref{eqn:capacity1}), (\ref{eqn:capacity2}) and (\ref{eqn:capacity3}), we have that

\vspace{-4mm}{\small
\begin{align*}
\sum_{s:m_s=m, s\in [1, S]}g_s\mu_{ii}^s(t)\leq C_i^m N_i^m-\sum_{j\neq i}\alpha_{ji}^m(t).
\end{align*}}\vspace{-4mm}

Similar with the above reasoning, we have that, with known $\alpha_{ji}^m(t)$, $\forall j\neq i, m\in [1, M]$, the best strategy is to assign full capacity of $C_i^m n_i^m(t)-\sum_{j\neq i}\alpha_{ji}^m(t)$ to the service type $s_m^*$ with the maximum per-VM scheduling weight $\frac{Q_i^s(t)+Z_i^s(t)}{g_s}-\frac{V\beta_i(t)}{C_i^m}$ if it is positive, or equivalently the maximum $\frac{Q_i^s(t)+Z_i^s(t)}{g_s}$ as defined in Eqn.~(\ref{eqn:weight}) and (\ref{eqn:max-weight}) in Sec.~\ref{sec:algorithm} if $\frac{Q_i^s(t)+Z_i^s(t)}{g_s}>\frac{V\beta_i(t)}{C_i^m}$. Then, we have the optimal solution to $\mu_{ii}^s(t)$, $\forall s\in [1,S]$ with Eqn.~(\ref{eqn:schedule2}).

\vspace{1mm}
\noindent \emph{Step 3 -- Find true values $\hat{b}_i^m(t)$, $\hat{\gamma}_i^m(t)$ $\hat{s}_i^m(t)$ and $\hat{\eta}_i^m(t)$, $\forall m\in [1, M]$}: According to property 1, we have that, once the true values are identified, each cloud always bids with true values resulting in the final assignments to $\hat{s}_i^m(t)$, $\hat{\eta}_i^m(t)$, $\hat{b}_i^m(t)$, $\hat{\gamma}_i^m(t)$ and $\alpha_{ij}^m(t)$ decided any given double auction scheme.

With property 2, we know that, the per-VM utility gain for each winner (buy and/or sell) is non-negative. Hence, the best strategy for bidding volume is that each buy/sell always bids for maximum possible VMs. Then, we have the true valuations for buy-bid and sell-bid volumes as in Eqn.~(\ref{eqn:true-vol-b}) and (\ref{eqn:true-vol-s}), respectively.

Finally, we identify the true valuation for bidding prices of each VM type $m\in [1, M]$ case by case:

\vspace{1mm}
\noindent -- Case 1: Cloud $i$ only wins the buy-bid. In this case, we have that i) all bought type-$m$ VMs are from other clouds and should be used for job scheduling according to constraint Eqn.~(\ref{eqn:capacity3}); and ii) $\hat{s}_i^m(t)=0$ and $\hat{\eta}_i^m(t)=0$.

A nice property of problem (\ref{eqn:profit-oneshot1}) is that, all variables related with VM type $m$, \emph{i.e.}, $b_i^m(t)$, $\gamma_i^m$, $\hat{b}_i^m(t)$, $\hat{\gamma}_i^m(t)$, $\alpha_{ij}^m(t)$, $n_i^m(t)$ and $\mu_{ij}^s(t)$ with $m_s=m$, are independent from those related with other VM types. Hence, the optimal solutions to variables related with VM type $m$ can be found by solving the following optimization problem,

\vspace{-4mm}{\small
\begin{align}
\max&~~~~\sum_{s: m_s=m, s\in [1, S]}\sum_{j\in [1, F]}\mu_{ij}^s(t) [Q_i^s(t)+ Z_i^s(t)]-V\hat{b}_i^m(t)\hat{\gamma}_i^m(t)\notag\\&~~~~-V\beta_i(t)n_i^m(t) \label{eqn:profit-oneshot4}\\
s.t.&~~~~\text{Constraint (\ref{eqn:capacity1})-(\ref{eqn:capacity3})}.\notag
\end{align}}\vspace{-4mm}

We can replace $\hat{\gamma}_i^m(t)$ with $\sum_{j\in [1, F]}\alpha_{ij}^m(t)$ with Eqn.~(\ref{eqn:buy-cons}). Next, we replace each $\mu_{ij}^s(t)$ with the optimal solutions in Eqn.~(\ref{eqn:schedule1}) and (\ref{eqn:schedule2}) in problem (\ref{eqn:profit-oneshot4}). So, we have that

\vspace{-4mm}{\small
\begin{align}
\max&~~~~\sum_{j\neq i, j\in [1, F]}\alpha_{ij}^m(t) [\frac{Q_i^{s_m^*}(t)+ Z_i^{s_m^*}(t)}{g_{s_m^*}}-V\hat{b}_i^m(t)]\notag\\&~~~~+\mu_{ii}^{s_m^*}(t)-V\beta_i(t)n_i^m(t) \label{eqn:profit-oneshot5}.
\end{align}}\vspace{-4mm}

According to Eqn.~(\ref{eqn:schedule1}) and (\ref{eqn:server-provision}), if $\frac{Q_i^{s_m^*}(t)+Z_i^{s_m^*}(t)}{V g_{s_m^*}}>\frac{\beta_i(t)}{C_i^m}$, we have that

\vspace{-4mm}{\small
\begin{align*}
\mu_{ii}^{s_m^*}(t)-V\beta_i(t)n_i^m(t)=N_i^m [C_i^m \frac{Q_i^{s_m^*}(t)+Z_i^{s_m^*}(t)}{g_{s_m^*}} - V \beta_i(t)],
\end{align*}}\vspace{-4mm}

\noindent otherwise, we have that

\vspace{-4mm}{\small
\begin{align*}
\mu_{ii}^{s_m^*}(t)-V\beta_i(t)n_i^m(t)=0.
\end{align*}}\vspace{-4mm}

Both of the above values are constants. As a result, the optimization problem (\ref{eqn:profit-oneshot5}) is finally equivalent to find optimal solution for the following problem,

\vspace{-4mm}{\small
\begin{align}
\max&~~~~\sum_{j\neq i, j\in [1, F]}\alpha_{ij}^m(t) [\frac{Q_i^{s_m^*}(t)+ Z_i^{s_m^*}(t)}{g_{s_m^*}}-V\hat{b}_i^m(t)]\label{eqn:profit-oneshot6}.
\end{align}}\vspace{-4mm}

According to the definition of the true value of buy-bid price, we know that the true valuation of $\tilde{b}_i^m(t)$ should be $\frac{Q_i^{s_m^*}(t)+ Z_i^{s_m^*}(t)}{V g_{s_m^*}}$ as defined in Eqn.~(\ref{eqn:true-buy}) since: i) if $\hat{b}_i^m(t)>\frac{Q_i^{s_m^*}(t)+ Z_i^{s_m^*}(t)}{V g_{s_m^*}}$, we will have utility loss for problem (\ref{eqn:profit-oneshot6}) and thus problem (\ref{eqn:profit-oneslot}); ii) if $\hat{b}_i^m(t)<\frac{Q_i^{s_m^*}(t)+ Z_i^{s_m^*}(t)}{V g_{s_m^*}}$, we will have utility improvement for problem (\ref{eqn:profit-oneshot6}) and thus problem (\ref{eqn:profit-oneslot}); and iii) if $\hat{b}_i^m(t)=\frac{Q_i^{s_m^*}(t)+ Z_i^{s_m^*}(t)}{V g_{s_m^*}}$, we will have zero utility improvement or loss for problem (\ref{eqn:profit-oneshot6}) and thus problem (\ref{eqn:profit-oneslot}).

\vspace{1mm}
\noindent -- Case 2: Cloud $i$ only wins the sell-bid. In this case, we have that i) all sold type-$m$ VMs are used by other clouds for job scheduling according to constraint Eqn.~(\ref{eqn:capacity3}); and ii) $\hat{b}_i^m(t)=0$ and $\hat{\gamma}_i^m(t)=0$.

Similar with above analysis, the optimal solutions to variables related with VM type $m$ can be found by solving the following optimization problem,

\vspace{-4mm}{\small
\begin{align}
\max&~~~~\sum_{s: m_s=m, s\in [1, S]}\sum_{j\neq i}\mu_{ij}^s(t) [Q_i^s(t)+ Z_i^s(t)]+V\hat{s}_i^m(t)\hat{\eta}_i^m(t)\notag\\&~~~~-V\beta_i(t)n_i^m(t) \label{eqn:profit-oneshot7}\\
s.t.&~~~~\text{Constraint (\ref{eqn:capacity1})-(\ref{eqn:capacity3})}.\notag
\end{align}}\vspace{-4mm}

We can replace $\hat{\eta}_i^m(t)$ with $\sum_{j\neq i}\alpha_{ji}^m(t)$ using Eqn.~(\ref{eqn:sell-cons}) and the fact in this case that $\alpha_{ii}^m(t)=0$. Next, we replace each $\mu_{ii}^s(t)$ and $n_i^m(t)$ with the optimal solutions in Eqn.~(\ref{eqn:schedule1}), (\ref{eqn:schedule2}) and (\ref{eqn:server-provision}) in problem (\ref{eqn:profit-oneshot7}). So, we have that,

\noindent -- i) if $\frac{Q_i^{s_m^*}(t)+Z_i^{s_m^*}(t)}{V g_{s_m^*}}>\frac{\beta_i(t)}{C_i^m}$, problem (\ref{eqn:profit-oneshot7}) is equivalent to

\vspace{-4mm}{\small
\begin{align*}
\max&~~~~N_i^m [C_i^m \frac{Q_i^{s_m^*}(t)+Z_i^{s_m^*}(t)}{g_{s_m^*}} - V \beta_i(t)]\\
    &~~~~+\sum_{j\neq i}\alpha_{ji}^m(t)[V \hat{s}_i^m(t) - \frac{Q_i^{s_m^*}(t)+Z_i^{s_m^*}(t)}{g_{s_m^*}}],
\end{align*}}\vspace{-4mm}

\noindent where the true value of $\tilde{s}_i^m(t)$ should be $\frac{Q_i^{s_m^*}(t)+Z_i^{s_m^*}(t)}{Vg_{s_m^*}}$ according the definition of true value of sell-bid;

\noindent -- ii) otherwise, problem (\ref{eqn:profit-oneshot7}) is equivalent to

\vspace{-4mm}{\small
\begin{align*}
\max&~~~~V\sum_{j\neq i}\alpha_{ji}^m(t) [\hat{s}_i^m(t)-\beta_i(t)/C_i^m],
\end{align*}}\vspace{-4mm}

\noindent where the true value of $\tilde{s}_i^m(t)$ should be $\beta_i(t)/C_i^m$ according the definition of true value of sell-bid. So, we have the solution to true value of sell-bid price as in Eqn.~(\ref{eqn:true-sell}).

\vspace{1mm}
\noindent -- Case 3: Cloud $i$ wins both buy-bid and sell-bid. In this case, we have that $\alpha_{ii}^m(t)>0$. However, based on property 5, we see that $\hat{s}_i^m(t)=\hat{b}_i^m(t)$ and the overall utility gain on the self-trading of VM, \emph{i.e.}, $\alpha_{ii}^m(t)[\hat{s}_i^m(t)-\hat{b}_i^m(t)]$, is zero. Besides, according to Eqn.~(\ref{eqn:capacity1}), the value of $\alpha_{ii}^m(t)$ does not influence the job scheduling and server provisioning. Hence, all the derivations are the same with above two cases and the true valuations of buy-bid and sell-bid are defined as in Eqn.(\ref{eqn:true-buy}) and (\ref{eqn:true-sell}).

\vspace{1mm}
\noindent \textbf{Problem (\ref{eqn:profit-oneshot2})}: It is not hard to find that, the only controllable variables are the job-dropping decisions, \emph{i.e.}, $D_i^s(t)$ ($\forall s\in [1,S]$), with $Q_i^s(t)+Z_i^s(t)-V\cdot \xi_i^s$ as the weight for $D_i^s(t)$ in the objective $\varphi_3^i(t)$. If $Q_i^s(t)+Z_i^s(t)-V\cdot \xi_i^s> 0$, the weight for $D_i^s(t)$ is positive and type-$s$ jobs should be dropped at maximum rate, \emph{i.e.}, $D_i^s(t)=D_i^{s(max)}$; otherwise, the weight is non-positive and the optimal solution is $D_i^s(t)=0$. So, we get the algorithm for job dropping as in Alg.~\ref{alg:profit}.

}

\section{Finding true values $\hat{b}_i^m(t)$, $\hat{\gamma}_i^m(t)$, $\hat{s}_i^m(t)$ and $\hat{\eta}_i^m(t)$}
\label{appendix:truevalue}



Based on individual rationality and truthfulness of the double auction mechanism, each buyer/seller pays/charges a price that is no higher/lower than the corresponding bid (true) value, while the number of VMs actually traded is no larger than the bid (true) value if the bid is successful, \emph{i.e.}, $\hat{b}_i^m(t)\leq \tilde{b}_i^m(t)$, $\hat{s}_i^m(t)\geq \tilde{s}_i^m(t)$, $\hat{\gamma}_i^m(t)\leq \tilde{\gamma}_i^m(t)$ and $\hat{\eta}_i^m(t)\leq \tilde{\eta}_i^m(t)$.
That is, the utility obtained by each cloud by participating in the auction is non-negative.
Hence, the utility obtained by trading each VM at each winning buyer, \emph{i.e.}, $\tilde{b}_i^m(t)-\hat{b}_i^m(t)$, or seller, \emph{i.e.}, $\tilde{s}_i^m(t)-\hat{s}_i^m(t)$, is non-negative. Therefore, bidding for the maximum number of potential VMs provisioned, maximizes the utility of a seller or buyer, 
 {\em i.e.}, the maximum number of type-m VMs a cloud is willing to sell or buy is the maximum number of potential type-$m$ VMs provisioned in the federation, and hence the true values of the VM volumes to bid at each cloud are derived as in Eqn.~(\ref{eqn:true-vol-b}) and (\ref{eqn:true-vol-s}), respectively.

We next identify the true values of the bidding prices for each type of VMs, $m\in[1,M]$, at cloud $i$ case by case:

\vspace{1mm}
\noindent $\triangleright$ {\em Case 1: Cloud $i$'s buy-bid for type-$m$ VMs wins, but not the sell-bid.}

In this case, we know that: i) all bought type-$m$ VMs are from other clouds and should be used for job scheduling according to constraint (\ref{eqn:capacity3}); ii) $\hat{s}_i^m(t)=0$ and $\hat{\eta}_i^m(t)=0$.

A nice property of problem (\ref{eqn:profit-oneshot1}) is that, all decision variables related to type-$m$ VMs, \emph{i.e.}, $b_i^m(t)$, $\gamma_i^m$, $\hat{b}_i^m(t)$, $\hat{\gamma}_i^m(t)$, $\alpha_{ij}^m(t)$, $n_i^m(t)$, and $\mu_{ij}^s(t)$ with $m_s=m$, are independent from those related to the other types of VMs. Hence, the optimal solutions to decision variables related to type-$m$ VMs can be derived by solving the following sub problem from (\ref{eqn:profit-oneshot1}):

\vspace{-4mm}{\small
\begin{align}
\max&~~~~\sum_{s: m_s=m, s\in [1, S]}\sum_{j\in [1, F]}\mu_{ij}^s(t) [Q_i^s(t)+ Z_i^s(t)]-V\hat{b}_i^m(t)\hat{\gamma}_i^m(t)\notag\\&~~~~-V\beta_i(t)n_i^m(t) \label{eqn:profit-oneshot4}\\
s.t.&~~~~\text{Constraint (\ref{eqn:capacity1})-(\ref{eqn:capacity3})}.\notag
\end{align}}\vspace{-4mm}

In (\ref{eqn:profit-oneshot4}), we replace $\hat{\gamma}_i^m(t)$ by $\sum_{j\in [1, F]}\alpha_{ij}^m(t)$ based on Eqn.~(\ref{eqn:buy-cons}), and replace $\mu_{ij}^s(t)$'s by the optimal solutions in Eqn.~(\ref{eqn:schedule1}) and (\ref{eqn:schedule2}) (to be derived in Sec.~\ref{sec:scheduling}). We obtain

\vspace{-4mm}{\small
\begin{align}
\max&~~~~\sum_{j\neq i, j\in [1, F]}\alpha_{ij}^m(t) [\frac{Q_i^{s_m^*}(t)+ Z_i^{s_m^*}(t)}{g_{s_m^*}}-V\hat{b}_i^m(t)]\notag\\
    &~~~~+\mu_{ii}^{s_m^*}(t) (Q_i^{s_m^*}(t)+ Z_i^{s_m^*}(t))-V\beta_i(t)n_i^m(t) \label{eqn:profit-oneshot5}.
\end{align}}\vspace{-4mm}

According to Eqn.~(\ref{eqn:schedule1}) and Eqn.~(\ref{eqn:server-provision}), if $\frac{Q_i^{s_m^*}(t)+Z_i^{s_m^*}(t)}{V g_{s_m^*}}>\frac{\beta_i(t)}{C_i^m}$, we have

\vspace{-4mm}{\small
\begin{align*}
&\mu_{ii}^{s_m^*}(t) (Q_i^{s_m^*}(t)+ Z_i^{s_m^*}(t))-V\beta_i(t)n_i^m(t)\\
=&N_i^m [C_i^m \frac{Q_i^{s_m^*}(t)+Z_i^{s_m^*}(t)}{g_{s_m^*}} - V \beta_i(t)];
\end{align*}}\vspace{-4mm}

\noindent otherwise, we have

\vspace{-4mm}{\small
\begin{align*}
\mu_{ii}^{s_m^*}(t) (Q_i^{s_m^*}(t)+ Z_i^{s_m^*}(t))-V\beta_i(t)n_i^m(t)=0.
\end{align*}}\vspace{-4mm}

Both RHS values of the above equations are constants. As a result, the optimization problem (\ref{eqn:profit-oneshot5}) is finally equivalent to the following one:

\vspace{-4mm}{\small
\begin{align}
\max&~~~~\sum_{j\neq i, j\in [1, F]}\alpha_{ij}^m(t) [\frac{Q_i^{s_m^*}(t)+ Z_i^{s_m^*}(t)}{g_{s_m^*}}-V\hat{b}_i^m(t)]\label{eqn:profit-oneshot6}.
\end{align}}\vspace{-4mm}

According to the definition of true values, we know that the true value of $\tilde{b}_i^m(t)$ should be $\frac{Q_i^{s_m^*}(t)+ Z_i^{s_m^*}(t)}{V g_{s_m^*}}$ as defined in Eqn.~(\ref{eqn:true-buy}), since: i) if $\hat{b}_i^m(t)>\frac{Q_i^{s_m^*}(t)+ Z_i^{s_m^*}(t)}{V g_{s_m^*}}$, the utility in (\ref{eqn:profit-oneshot6}) is negative, and hence a profit loss in terms of problem (\ref{eqn:profit-oneslot}) for cloud $i$; ii) if $\hat{b}_i^m(t)<\frac{Q_i^{s_m^*}(t)+ Z_i^{s_m^*}(t)}{V g_{s_m^*}}$, the utility in (\ref{eqn:profit-oneshot6}) is positive, and hence a profit gain in terms of problem (\ref{eqn:profit-oneslot}); and iii) if $\hat{b}_i^m(t)=\frac{Q_i^{s_m^*}(t)+ Z_i^{s_m^*}(t)}{V g_{s_m^*}}$, the utility in (\ref{eqn:profit-oneshot6}) is zero, and the profit of cloud $i$ in (\ref{eqn:profit-oneslot}) remains the same as not acquiring the VMs.

\vspace{1mm}
\noindent $\triangleright$ {\em Case 2: Cloud $i$'s sell-bid for type-$m$ VMs wins, but not the buy-bid.}

In this case, we know that: i) all type-$m$ VMs sold from cloud $i$ are used by other clouds for job scheduling according to constraint (\ref{eqn:capacity3}); and ii) $\hat{b}_i^m(t)=0$ and $\hat{\gamma}_i^m(t)=0$.

Similar to the analysis in Case 1, the optimal solutions to variables related to type-$m$ VMs can be obtained by solving the following optimization problem:

\vspace{-4mm}{\small
\begin{align}
\max&~~~~\sum_{s: m_s=m, s\in [1, S]}\sum_{j\neq i}\mu_{ij}^s(t) [Q_i^s(t)+ Z_i^s(t)]+V\hat{s}_i^m(t)\hat{\eta}_i^m(t)\notag\\&~~~~-V\beta_i(t)n_i^m(t) \label{eqn:profit-oneshot7}\\
s.t.&~~~~\text{Constraint (\ref{eqn:capacity1})-(\ref{eqn:capacity3})}.\notag
\end{align}}\vspace{-4mm}

In (\ref{eqn:profit-oneshot7}), we replace $\hat{\eta}_i^m(t)$ by $\sum_{j\neq i}\alpha_{ji}^m(t)$ based on Eqn.~(\ref{eqn:sell-cons}) and the fact in this case that $\alpha_{ii}^m(t)=0$, and replace $\mu_{ij}^s(t)$'s and $n_i^m(t)$'s with the optimal solutions in Eqn.~(\ref{eqn:schedule1}), (\ref{eqn:schedule2}) and Eqn.~(\ref{eqn:server-provision}). Then we have the following two cases:

\noindent  i) if $\frac{Q_i^{s_m^*}(t)+Z_i^{s_m^*}(t)}{V g_{s_m^*}}>\frac{\beta_i(t)}{C_i^m}$, problem (\ref{eqn:profit-oneshot7}) is equivalent to

\vspace{-4mm}{\small
\begin{align*}
\max&~~~~N_i^m [C_i^m \frac{Q_i^{s_m^*}(t)+Z_i^{s_m^*}(t)}{g_{s_m^*}} - V \beta_i(t)]\\
    &~~~~+\sum_{j\neq i}\alpha_{ji}^m(t)[V \hat{s}_i^m(t) - \frac{Q_i^{s_m^*}(t)+Z_i^{s_m^*}(t)}{g_{s_m^*}}],
\end{align*}}\vspace{-4mm}

\noindent where the true value of $\tilde{s}_i^m(t)$ should be $\frac{Q_i^{s_m^*}(t)+Z_i^{s_m^*}(t)}{Vg_{s_m^*}}$ according to the definition of true value of the price to sell a type-$m$ VM;

\noindent  ii) otherwise, problem (\ref{eqn:profit-oneshot7}) is equivalent to

\vspace{-4mm}{\small
\begin{align*}
\max&~~~~V\sum_{j\neq i}\alpha_{ji}^m(t) [\hat{s}_i^m(t)-\beta_i(t)/C_i^m],
\end{align*}}\vspace{-4mm}

\noindent where the true value of $\tilde{s}_i^m(t)$ should be $\beta_i(t)/C_i^m$ according to the definition of the true value of the price to sell a type-$m$ VM. Hence, we have derived the true values of  $\tilde{s}_i^m(t)$ given in Eqn.~(\ref{eqn:true-sell}).

\vspace{1mm}
\noindent $\triangleright$ {\em Case 3: Both Cloud $i$'s buy-bid and sell-bid for type-$m$ VMs win.}

In this case, the following properties hold:

\vspace{1mm}
\noindent \emph{Property 1. } If both cloud $i$'s buy-bid and sell-bid for type-$m$ VMs win, the cloud cannot buy a type-$m$ VM with a price strictly higher than its price to sell a type-$m$ VM, \emph{i.e.}, $\hat{s}_i^m(t)\geq  \hat{b}_i^m(t)$. Otherwise, there will be a positive profit loss at the cloud by self-trading its own type-$m$ VMs, which violates its individual rationality.

\vspace{1mm}
\noindent \emph{Property 2. } If both cloud $i$'s buy-bid and sell-bid for type-$m$ VMs win, the cloud cannot sell a type-$m$ VM with a price strictly higher than its price to buy a type-$m$ VM, \emph{i.e.}, $\hat{s}_i^m(t)\leq \hat{b}_i^m(t)$. Otherwise, the auctioneer has to pay a positive sum to compensate for the price difference for those inter-cloud traded type-$m$ VMs, which contradicts the ex-post budget balance property at the auctioneer.

\vspace{1mm}
\noindent \emph{Property 3. } Based on Properties 1 and 2, if both cloud $i$'s buy-bid and sell-bid for type-$m$ VMs win, the actual buy and sell prices at the cloud for type-$m$ VMs are the same, \emph{i.e.}, $\hat{s}_i^m(t)=\hat{b}_i^m(t)$.

According to Property 3, we derive that the overall profit gain at cloud $i$ for self-trading of type-$m$ VMs, $\alpha_{ii}^m(t)[\hat{s}_i^m(t)-\hat{b}_i^m(t)]$, is zero.

The optimal solutions to variables related to type-$m$ VMs can be obtained by solving the following optimization problem:

\vspace{-4mm}{\small
\begin{align}
\max&~~~~\sum_{s: m_s=m, s\in [1, S]}\sum_{j\in [1, F]}\mu_{ij}^s(t) [Q_i^s(t)+ Z_i^s(t)]-V\hat{b}_i^m(t)\hat{\gamma}_i^m(t)\notag\\
    &~~~~+V\hat{s}_i^m(t)\hat{\eta}_i^m(t)-V\beta_i(t)n_i^m(t) \label{eqn:profit-oneshot8}\\
s.t.&~~~~\text{Constraint (\ref{eqn:capacity1})-(\ref{eqn:capacity3})}.\notag
\end{align}}\vspace{-4mm}

In (\ref{eqn:profit-oneshot8}), we replace $\hat{\gamma}_i^m(t)$ by $\sum_{j\in [1, F]}\alpha_{ij}^m(t)$ based on Eqn.~(\ref{eqn:buy-cons}), and $\hat{\eta}_i^m(t)$ by $\sum_{j\in [1,F]}\alpha_{ji}^m(t)$ based on Eqn.~(\ref{eqn:sell-cons}).
We also replace $\mu_{ij}^s(t)$'s and $n_i^m(t)$'s with the optimal solutions in Eqn.~(\ref{eqn:schedule1}), (\ref{eqn:schedule2}) and Eqn.~(\ref{eqn:server-provision}). Then we have the following two cases:

\noindent  i) if $\frac{Q_i^{s_m^*}(t)+Z_i^{s_m^*}(t)}{V g_{s_m^*}}>\frac{\beta_i(t)}{C_i^m}$, problem (\ref{eqn:profit-oneshot8}) is equivalent to

\vspace{-4mm}{\small
\begin{align*}
\max&~~~~N_i^m [C_i^m \frac{Q_i^{s_m^*}(t)+Z_i^{s_m^*}(t)}{g_{s_m^*}} - V \beta_i(t)]\\
    &~~~~\sum_{j\neq i}\alpha_{ij}^m(t) [\frac{Q_i^{s_m^*}(t)+ Z_i^{s_m^*}(t)}{g_{s_m^*}}-V\hat{b}_i^m(t)]\\
    &~~~~+\sum_{j\neq i}\alpha_{ji}^m(t)[V \hat{s}_i^m(t) - \frac{Q_i^{s_m^*}(t)+Z_i^{s_m^*}(t)}{g_{s_m^*}}],
\end{align*}}\vspace{-4mm}

\noindent where the true values of $\tilde{b}_i^m(t)$ and $\tilde{s}_i^m(t)$ should both be $\frac{Q_i^{s_m^*}(t)+ Z_i^{s_m^*}(t)}{V g_{s_m^*}}$ according to the definition of true value of the price to buy/sell a type-$m$ VM;

\noindent  ii) otherwise, problem (\ref{eqn:profit-oneshot8}) is equivalent to

\vspace{-4mm}{\small
\begin{align*}
\max&~~~~\sum_{j\neq i}\alpha_{ij}^m(t) [\frac{Q_i^{s_m^*}(t)+ Z_i^{s_m^*}(t)}{g_{s_m^*}}-V\hat{b}_i^m(t)]\\
    &~~~~V\sum_{j\neq i}\alpha_{ji}^m(t) [\hat{s}_i^m(t)-\beta_i(t)/C_i^m],
\end{align*}}\vspace{-4mm}

\noindent where the true value of $\tilde{b}_i^m(t)$ should be $\frac{Q_i^{s_m^*}(t)+ Z_i^{s_m^*}(t)}{V g_{s_m^*}}$, and the true value of $\tilde{s}_i^m(t)$ should be $\beta_i(t)/C_i^m$ according to the definition of the true value of the price to buy/sell a type-$m$ VM. Hence, we have derived the true values of $\tilde{b}_i^m(t)$ and $\tilde{s}_i^m(t)$ given in Eqn.~(\ref{eqn:true-buy}) and Eqn.~(\ref{eqn:true-sell}).



\section{Derivation of the one-shot optimization problem for social welfare maximization 
}\label{appendix:drift2}

Squaring the queuing laws (\ref{eqn:queue1}) and (\ref{eqn:queue2}), we can derive the following inequality

\vspace{-4mm}{\small
\begin{align*}
\Delta(\Theta(t))\leq& \frac{1}{2}\sum_{i\in [1, F]}\sum_{s\in [1, S]}[[\sum_{j=1}^F \mu_{ij}^s(t)+D_i^s(t)]^2 + [r_i^s(t)]^2\\
& + 2Q_i^s(t)[r_i^s(t)-\sum_{j=1}^F \mu_{ij}^s(t)-D_i^s(t)] + [\mathbf{1}_{\{Q_i^s(t)>0\}}\epsilon_s]^2\\
& + [D_i^s(t)+ \mathbf{1}_{\{Q_i^s(t)=0\}}\sum_{j=1}^F C_j^{m_s}N_j^{m_s}/g_s\\ &+ \mathbf{1}_{\{Q_i^s(t)>0\}}\sum_{j=1}^F \mu_{ij}^s(t)]^2\\
& + 2Z_i^s(t)[\mathbf{1}_{\{Q_i^s(t)>0\}}[\epsilon_s-\sum_{j=1}^F \mu_{ij}^s(t)]-D_i^s(t)\\
&-\mathbf{1}_{\{Q_i^s(t)=0\}}\sum_{j=1}^F C_j^{m_s}N_j^{m_s}/g_s]]\\
\leq & \frac{1}{2}\sum_{i\in [1, F]}\sum_{s\in [1, S]}[[\sum_{j=1}^F C_j^{m_s}N_j^{m_s}/g_s+D_i^{s(max)}]^2 + [R_i^s]^2\\
&+ 2Q_i^s(t)[r_i^s(t)-\sum_{j=1}^F \mu_{ij}^s(t)-D_i^s(t)]\\
&+ [\epsilon_s]^2 + [D_i^{s(max)}+\sum_{j=1}^F C_j^{m_s}N_j^{m_s}/g_s]^2 \\
& + 2Z_i^s(t)[\epsilon_s-\sum_{j=1}^F \mu_{ij}^s(t)-D_i^s(t)]]\\
=& B + \sum_{i\in [1, F]}\sum_{s\in [1, S]}[Q_i^s(t)[r_i^s(t)-\sum_{j=1}^F \mu_{ij}^s(t)-D_i^s(t)]\\ &+ Z_i^s(t)[\epsilon_s-\sum_{j=1}^F \mu_{ij}^s(t)-D_i^s(t)]],
\end{align*}}\vspace{-4mm}

\noindent where {\small $B=\frac{1}{2}\sum_{i\in [1, F]}\sum_{s\in [1, S]}[[\sum_{j=1}^F C_j^{m_s}N_j^{m_s}/g_s+D_i^{s(max)}]^2 + [R_i^s]^2+[\epsilon_s]^2 + [D_i^{s(max)}+\sum_{j=1}^F C_j^{m_s}N_j^{m_s}/g_s]^2]$}.

By applying the drift-plus-penalty framework (or equivalently, drift-minus-profit here), we subtract the weighted one-shot social welfare in time $t$, \emph{i.e.}, $V\cdot \sum_{i\in [1, F]}[\sum_{m\in [1, M]} [\hat{s}_i^m(t) \hat{\eta}_i^m(t) - \hat{b}_i^m(t) \hat{\gamma}_i^m(t)-\beta_i(t) n_i^m(t)]+ \sum_{s\in [1, S]}[p_i^s(t)\cdot r_i^s(t)- D_i^s(t) \xi_i^s]]$, on both sides of the above inequality. Hence, we have the following inequality:


\vspace{-5mm}{\small\begin{align*}
&\Delta(\Theta(t))-V\cdot \sum_{i\in [1, F]}[\sum_{m\in [1, M]} [\hat{s}_i^m(t) \hat{\eta}_i^m(t) - \hat{b}_i^m(t) \hat{\gamma}_i^m(t)\notag\\
&-\beta_i(t) n_i^m(t)] + \sum_{s\in [1, S]}[p_i^s(t)\cdot r_i^s(t)- D_i^s(t) \xi_i^s]] \notag\\
\leq&  B + \sum_{i\in [1, F]}\sum_{s\in [1, S]}[Q_i^s(t) r_i^s(t) + Z_i^s(t) \epsilon_s - V p_i^s(t)\cdot r_i^s(t)] \notag\\
&- \varphi_1(t)-\varphi_2(t),
\end{align*}}\vspace{-5mm}

\noindent where $V>0$ is a user-defined positive constant that can be understood as the
weight of profit in the expression.

\section{Proof to Theorem \ref{theorem:truthfulness}}\label{appendix:truthfulness}

This theorem can be proved based on the following two lemmas.

\begin{lemma}[Monotonic winner determination]\label{lemma:monotonic}
Given prices of buy-bids $\{b_1^m(t), \ldots, b_i^m(t),\ldots, b_F^m(t)\}$ and sell-bids $\{s_1^m(t), \ldots, s_i^m(t), \ldots, s_F^m(t)\}$, we have that
\begin{enumerate}
\item If cloud $i$ wins the buy-bid by bidding with $b_i^m(t)$, then cloud $i$ also wins the buy-bid by bidding with $b'>b_i^m(t)$;
\item If cloud $i$ wins the sell-bid by bidding with $s_i^m(t)$, then cloud $i$ also wins the sell-bid by bidding with $s'<s_i^m(t)$.
\end{enumerate}
\end{lemma}

\begin{proof}
We prove the cases in the lemma respectively as follows,
\begin{enumerate}
\item Since cloud $i$ wins the buy-bid with $b_i^m(t)$, we know that $b_i^m(t)$ is the largest among all buy-bids, \emph{i.e.}, $b_i^m(t)\geq b_j^m(t)$, $\forall j\neq i, j\in [1, F]$. With $b'>b_i^m(t)$, we have that $b'> b_j^m(t)$, $\forall j\neq i, j\in [1, F]$. Hence, if cloud $i$ proposes a buy-bid with $b'$, it can still win the buy-bid according to our winner determination decision, since its buy-bid price is still the largest among all buy-bids.
\item Since cloud $i$ wins the sell-bid with $s_i^m(t)$ and the sell-bids are sorted in ascending order, we know that $s_i^m(t)\leq \vartheta_{j'}^m(t)$, where $j'$ is the critical index as defined in Eqn.~(\ref{eqn:j'}), and $s_i^m(t)$ is among the $(j'-1)^{th}$ lowest sell-bids in ascending order. With $s'<s_i^m(t)$, we also have that $s'< \vartheta_{j'}^m(t)$ and $s'$ is among the $(j'-1)^{th}$ lowest sell-bids in ascending order. Hence, if cloud $i$ propose a sell-bid with $s'$, it can still win the sell-bid according to our winner determination decision, since its sell-bid price is still among the $(j'-1)^{th}$ lowest sell-bids.
\end{enumerate}
\end{proof}

\begin{lemma}[Bid-independent pricing]\label{lemma:bid-independent}
Given prices of buy-bids $\{b_1^m(t), \ldots, b_i^m(t),\ldots, b_F^m(t)\}$ and sell-bids $\{s_1^m(t), \ldots, s_i^m(t), \ldots, s_F^m(t)\}$, we have that
\begin{enumerate}
\item If cloud $i$ wins the buy-bid by bidding with $b_i^m(t)$ and $b'$, the charged price $\hat{b}_i^m(t)$ to cloud $i$ is the same for both;
\item If cloud $i$ wins the sell-bid by bidding with $s_i^m(t)$ and $s'$, the charged price $\hat{s}_i^m(t)$ to cloud $i$ is the same for both.
\end{enumerate}
\end{lemma}

\begin{proof}
We prove the the cases in the lemma respectively as follows,
\begin{enumerate}
\item Since cloud $i$ wins the buy-bid, the charged price $\hat{b}_i^m(t)$ should be $\theta_2^m(t)$, which is the second largest buy-bid price independent of cloud $i$'s buy-bid, according to our pricing scheme in the auction. And we have that $b_i^m(t)\geq \theta_2^m(t)$ and $b'\geq \theta_2^m(t)$. We also have that, as long as cloud $i$ wins the buy-bid, the value of the second largest buy-bid price $\theta_2^m(t)$ does not change. Hence, cloud $i$ should be charged with $\theta_2^m(t)$ by bidding with no matter $b_i^m(t)$ or $b'$.
\item Since cloud $i$ wins the sell-bid, the charged price $\hat{s}_i^m(t)$ should be $\vartheta_{j'}^m(t)$, which is the $j'^{th}$ lowest sell-bid price, according to our pricing scheme in the auction. And we have that $s_i^m(t)\leq \vartheta_{j'}^m(t)$ and $s'\leq \vartheta_{j'}^m(t)$. As long as cloud $i$ wins the sell-bid, we have that $s_i^m(t)$ and $s'$ should be among the $(j'-1)^{th}$ lowest sell-bid prices and the value of $\vartheta_{j'}^m(t)$ does not change. Hence, cloud $i$ should be charged with the same price with $\vartheta_{j'}^m(t)$ by bidding with no matter $b_i^m(t)$ or $b'$ if it wins the sell-bid.
\end{enumerate}
\end{proof}

We can then prove the Theorem \ref{theorem:truthfulness} to show that any cloud $i$ cannot obtain higher utility gain by bidding untruthfully, \emph{i.e.}, $b_i^m(t)\neq \tilde{b}_i^m(t)$ and/or $s_i^m(t)\neq \tilde{s}_i^m(t)$, $\forall m\in [1,M]$, by analyzing all possible auction results.


\noindent \textbf{Case 1} -- \emph{Cloud $i$ wins both buy-bid and sell-bid with truthful bidding}: In this case, the charged/paid prices for buy-bid and sell-bid are $\theta_2^m(t)$ and $\vartheta_{j'}^m(t)$, respectively. We discuss the all possibly cases of untruthful bidding as follows,

\begin{itemize}
\item \emph{Bid untruthfully with $b_i^m(t)>\tilde{b}_i^m(t)$ and $s_i^m(t)>\tilde{s}_i^m(t)$}: According to Lemma \ref{lemma:monotonic}, cloud $i$ still wins the buy-bid but may either win or lose the sell-bid. If it also wins the sell-bid, we have that the charged/paid prices for buy-bid and sell-bid are $\theta_2^m(t)$ and $\vartheta_{j'}^m(t)$, respectively, according to Lemma \ref{lemma:bid-independent}; the utility gain is zero by bidding untruthfully since the charged/paid prices are the same with that by bidding truthfully. If it loses the sell-bid, we have that the charged prices for buy-bid and sell-bid are $\theta_2^m(t)$ and $0$, respectively, according to Lemma \ref{lemma:bid-independent} and our pricing scheme, and no VM is sold by cloud $i$; the utility gain is non-positive by bidding untruthfully, since the charged buy-bid price remains the same while there is non-negative utility loss, \emph{i.e.}, $[\vartheta_{j'}^m(t)-\tilde{s}_i^m(t)]\hat{\eta}_i^m(t)$ with $\vartheta_{j'}^m(t)\geq \tilde{s}_i^m(t)$, by losing the sell-bid.

\item \emph{Bid untruthfully with $b_i^m(t)>\tilde{b}_i^m(t)$ and $s_i^m(t)=\tilde{s}_i^m(t)$}: According to Lemma \ref{lemma:monotonic}, cloud $i$ still wins the buy-bid. It is also easy to see that cloud $i$ still wins the sell-bid with the same bidding price for sell-bid. Hence, the charged/paid prices for buy-bid and sell-bid are $\theta_2^m(t)$ and $\vartheta_{j'}^m(t)$, respectively, according to Lemma \ref{lemma:bid-independent}. The utility gain is zero by bidding untruthfully.

\item \emph{Bid untruthfully with $b_i^m(t)>\tilde{b}_i^m(t)$ and $s_i^m(t)<\tilde{s}_i^m(t)$}: According to Lemma \ref{lemma:monotonic}, cloud $i$ still wins both the buy-bid and the sell-bid. The charged/paid prices for buy-bid and sell-bid are $\theta_2^m(t)$ and $\vartheta_{j'}^m(t)$, respectively, according to Lemma \ref{lemma:bid-independent}. The utility gain is zero by bidding untruthfully.

\item \emph{Bid untruthfully with $b_i^m(t)=\tilde{b}_i^m(t)$ and $s_i^m(t)>\tilde{s}_i^m(t)$}: Cloud $i$ stills wins the buy-bid with the same bidding price for buy-bid. However, it may either win or lose the sell-bid. If it also wins the sell-bid, we have that the charged/paid prices for buy-bid and sell-bid are $\theta_2^m(t)$ and $\vartheta_{j'}^m(t)$, respectively, according to Lemma \ref{lemma:bid-independent} and our pricing scheme; the utility gain is zero by bidding untruthfully since the charged/paid prices are the same with that by bidding truthfully. If it loses the sell-bid, we have that the charged prices for buy-bid and sell-bid are $\theta_2^m(t)$ and $0$, respectively, according to our pricing scheme, and no VM is sold by cloud $i$; the utility gain is non-positive by bidding untruthfully, since the charged buy-bid price remains the same while there is non-negative utility loss, by losing the sell-bid.

\item \emph{Bid untruthfully with $b_i^m(t)=\tilde{b}_i^m(t)$ and $s_i^m(t)<\tilde{s}_i^m(t)$}: According to Lemma \ref{lemma:monotonic}, cloud $i$ still wins the sell-bid. It is also easy to see that cloud $i$ still wins the buy-bid with the same bidding price for buy-bid. Hence, the charged/paid prices for buy-bid and sell-bid are $\theta_2^m(t)$ and $\vartheta_{j'}^m(t)$, respectively, according to Lemma \ref{lemma:bid-independent}. The utility gain is zero by bidding untruthfully.

\item \emph{Bid untruthfully with $b_i^m(t)<\tilde{b}_i^m(t)$ and $s_i^m(t)>\tilde{s}_i^m(t)$}: The cloud can either win both buy and sell bids, or win buy-bid only, or win sell-bid only, or lose both bids. If the cloud still wins both bids, we have that the charged/paid prices for buy-bid and sell-bid are $\theta_2^m(t)$ and $\vartheta_{j'}^m(t)$, respectively, according to our pricing scheme; the utility gain is zero by bidding untruthfully since the charged/paid prices are the same with that by bidding truthfully. If it only wins the buy-bid, we have that the charged prices for buy-bid and sell-bid are $\theta_2^m(t)$ and $0$, respectively, according to Lemma our pricing scheme, and no VM is sold by cloud $i$; the utility gain is non-positive by bidding untruthfully, since the charged buy-bid price remains the same while there is non-negative utility loss, by losing the sell-bid. If it only wins the sell-bid, we have that the charged prices for buy-bid and sell-bid are $0$ and $\vartheta_{j'}^m(t)$, respectively, according to our pricing scheme, and no VM is bought by cloud $i$; the utility gain is non-positive by bidding untruthfully, since the charged sell-bid price remains the same while there is non-negative utility loss, \emph{i.e.}, $[\tilde{b}_i^m(t)-\theta_2^m(t)]\cdot \hat{\gamma}_i^m(t)$ with $\tilde{b}_i^m(t)\geq \theta_2^m(t)$, by losing the buy-bid. If it loses both bids, we have that the charged prices for buy-bid and sell-bid are both $0$, according to our pricing scheme, and no VM is bought by or sold by cloud $i$; the utility gain is non-positive by bidding untruthfully, since there is non-negative utility loss, \emph{i.e.}, $[\tilde{b}_i^m(t)-\theta_2^m(t)]\cdot \hat{\gamma}_i^m(t)+[\vartheta_{j'}^m(t)-\tilde{s}_i^m(t)]\cdot \hat{\eta}_i^m(t)$ with $\tilde{b}_i^m(t)\geq \theta_2^m(t)$ and $\vartheta_{j'}^m(t)\geq \tilde{s}_i^m(t)$, by losing the both bids.

\item \emph{Bid untruthfully with $b_i^m(t)<\tilde{b}_i^m(t)$ and $s_i^m(t)=\tilde{s}_i^m(t)$}: Cloud $i$ stills wins the sell-bid with the same bidding price for sell-bid. However, it may either win or lose the buy-bid. If it also wins the buy-bid, we have that the charged/paid prices for buy-bid and sell-bid are $\theta_2^m(t)$ and $\vartheta_{j'}^m(t)$, respectively, according to our pricing scheme; the utility gain is zero by bidding untruthfully since the charged/paid prices are the same with that by bidding truthfully. If it loses the buy-bid, we have that the charged prices for buy-bid and sell-bid are $0$ and $\vartheta_{j'}^m(t)$, respectively, according to our pricing scheme, and no VM is bought by cloud $i$; the utility gain is non-positive by bidding untruthfully, since the charged sell-bid price remains the same while there is non-negative utility loss, by losing the buy-bid.

\item \emph{Bid untruthfully with $b_i^m(t)<\tilde{b}_i^m(t)$ and $s_i^m(t)<\tilde{s}_i^m(t)$}: According to Lemma \ref{lemma:monotonic}, cloud $i$ still wins the sell-bid but may either win or lose the buy-bid. If it also wins the buy-bid, we have that the charged/paid prices for buy-bid and sell-bid are $\theta_2^m(t)$ and $\vartheta_{j'}^m(t)$, respectively, according to Lemma \ref{lemma:bid-independent} and our pricing scheme; the utility gain is zero by bidding untruthfully since the charged/paid prices are the same with that by bidding truthfully. If it loses the buy-bid, we have that the charged prices for buy-bid and sell-bid are $0$ and $\vartheta_{j'}^m(t)$, respectively, according to Lemma \ref{lemma:bid-independent} and our pricing scheme, and no VM is bought by cloud $i$; the utility gain is non-positive by bidding untruthfully, since the charged sell-bid price remains the same while there is non-negative utility loss, by losing the buy-bid.
\end{itemize}

\noindent \textbf{Case 2} -- \emph{Cloud $i$ wins buy-bid but loses sell-bid with truthful bidding}: In this case, the charged/paid prices for buy-bid and sell-bid are $\theta_2^m(t)$ and $0$, respectively. we discuss the all possible cases of untruthful bidding as follows,

\begin{itemize}
\item \emph{Bid untruthfully with $b_i^m(t)>\tilde{b}_i^m(t)$ and $s_i^m(t)>\tilde{s}_i^m(t)$}: According to Lemma \ref{lemma:monotonic}, cloud $i$ still wins the buy-bid. It is also easy to see that cloud $i$ still loses the sell-bid, since otherwise we will have a contradiction to Lemma \ref{lemma:monotonic}. Hence, the charged/paid prices for buy-bid and sell-bid are $\theta_2^m(t)$ and $0$, respectively, according to Lemma \ref{lemma:bid-independent} and our pricing scheme. The utility gain is zero by bidding untruthfully.

\item \emph{Bid untruthfully with $b_i^m(t)>\tilde{b}_i^m(t)$ and $s_i^m(t)=\tilde{s}_i^m(t)$}: According to Lemma \ref{lemma:monotonic}, cloud $i$ still wins the buy-bid. It is also easy to see that cloud $i$ still loses the sell-bid with the same bidding price for sell-bid. Hence, the charged/paid prices for buy-bid and sell-bid are $\theta_2^m(t)$ and $0$, respectively, according to Lemma \ref{lemma:bid-independent} and pricing scheme. The utility gain is zero by bidding untruthfully.

\item \emph{Bid untruthfully with $b_i^m(t)>\tilde{b}_i^m(t)$ and $s_i^m(t)<\tilde{s}_i^m(t)$}: According to Lemma \ref{lemma:monotonic}, cloud $i$ still wins the buy-bid but may either win or lose the sell-bid. If it loses the sell-bid, we have that the charged prices for buy-bid and sell-bid are $\theta_2^m(t)$ and $0$, respectively, according to Lemma \ref{lemma:bid-independent} and our pricing scheme, and no VM is sold by cloud $i$; the utility gain is zero by bidding untruthfully since the charged/paid prices are the same with that by bidding truthfully. If it also wins the sell-bid, we have that the charged/paid prices for buy-bid and sell-bid are $\theta_2^m(t)$ and $\vartheta_{j'}^m(t)\leq \tilde{s}_i^m(t)$, respectively, according to Lemma \ref{lemma:bid-independent} and our pricing scheme; the utility gain is non-positive by bidding untruthfully, since the charged buy-bid price remains the same while there is non-negative utility loss, \emph{i.e.}, $[\tilde{s}_i^m(t)-\vartheta_{j'}^m(t)]\hat{\eta}_i^m(t)$ with $\vartheta_{j'}^m(t)\leq \tilde{s}_i^m(t)$, by winning the sell-bid.

\item \emph{Bid untruthfully with $b_i^m(t)=\tilde{b}_i^m(t)$ and $s_i^m(t)>\tilde{s}_i^m(t)$}: Cloud $i$ still wins the buy-bid with the same bidding price. It is also easy to see that cloud $i$ still loses the sell-bid, since otherwise we will have a contradiction to Lemma \ref{lemma:monotonic}. Hence, the charged/paid prices for buy-bid and sell-bid are $\theta_2^m(t)$ and $0$, respectively, according to our pricing scheme. The utility gain is zero by bidding untruthfully.

\item \emph{Bid untruthfully with $b_i^m(t)=\tilde{b}_i^m(t)$ and $s_i^m(t)<\tilde{s}_i^m(t)$}: Cloud $i$ still wins the buy-bid with the same bidding price. However, it may either win or lose the sell-bid. If it loses the sell-bid, we have that the charged prices for buy-bid and sell-bid are $\theta_2^m(t)$ and $0$, respectively, according to our pricing scheme, and no VM is sold by cloud $i$; the utility gain is zero by bidding untruthfully since the charged/paid prices are the same with that by bidding truthfully. If it also wins the sell-bid, we have that the charged/paid prices for buy-bid and sell-bid are $\theta_2^m(t)$ and $\vartheta_{j'}^m(t)\leq \tilde{s}_i^m(t)$, respectively, according to our pricing scheme; the utility gain is non-positive by bidding untruthfully, since the charged buy-bid price remains the same while there is non-negative utility loss, by winning the sell-bid.

\item \emph{Bid untruthfully with $b_i^m(t)<\tilde{b}_i^m(t)$ and $s_i^m(t)>\tilde{s}_i^m(t)$}: It is easy to see that cloud $i$ still loses the sell-bid, since otherwise we will have a contradiction to Lemma \ref{lemma:monotonic}. However, it may either win or lose the buy-bid. If it also wins the buy-bid, we have that the charged/paid prices for buy-bid and sell-bid are $\theta_2^m(t)$ and $0$, respectively, according to our pricing scheme; the utility gain is zero by bidding untruthfully since the charged/paid prices are the same with that by bidding truthfully. If it loses the buy-bid, we have that the charged prices for buy-bid and sell-bid are $0$ and $0$, respectively, according to Lemma our pricing scheme, and no VM is bought or sold by cloud $i$; the utility gain is non-positive by bidding untruthfully, since the charged sell-bid price remains the same while there is non-negative utility loss, \emph{i.e.}, $[\tilde{b}_i^m(t)-\theta_2^m(t)]\cdot \hat{\gamma}_i^m(t)$ with $\tilde{b}_i^m(t)\geq \theta_2^m(t)$, by losing the buy-bid.

\item \emph{Bid untruthfully with $b_i^m(t)<\tilde{b}_i^m(t)$ and $s_i^m(t)=\tilde{s}_i^m(t)$}: The cloud $i$ still loses the sell-bid with the same bidding price. However, it may either win or lose the buy-bid. If it also wins the buy-bid, we have that the charged/paid prices for buy-bid and sell-bid are $\theta_2^m(t)$ and $0$, respectively, according to our pricing scheme; the utility gain is zero by bidding untruthfully since the charged/paid prices are the same with that by bidding truthfully. If it loses the buy-bid, we have that the charged prices for buy-bid and sell-bid are $0$ and $0$, respectively, according to our pricing scheme, and no VM is bought or sold by cloud $i$; the utility gain is non-positive by bidding untruthfully, since the charged sell-bid price remains the same while there is non-negative utility loss, by losing the buy-bid.

\item \emph{Bid untruthfully with $b_i^m(t)<\tilde{b}_i^m(t)$ and $s_i^m(t)<\tilde{s}_i^m(t)$}: The cloud can either win both buy and sell bids, or win buy-bid only, or win sell-bid only, or lose both bids. If the cloud wins both bids, we have that the charged/paid prices for buy-bid and sell-bid are $\theta_2^m(t)$ and $\vartheta_{j'}^m(t)\leq \tilde{s}_i^m(t)$, respectively, according to our pricing scheme; the utility gain is zero by bidding untruthfully since the charged buy-bid price remains the same while there is non-negative utility loss, by winning the sell-bid. If it only wins the buy-bid, we have that the charged prices for buy-bid and sell-bid are $\theta_2^m(t)$ and $0$, respectively, according to our pricing scheme, and no VM is sold by cloud $i$; the utility gain is zero by bidding untruthfully since the charged/paid prices are the same with that by bidding truthfully. If it only wins the sell-bid, we have that the charged prices for buy-bid and sell-bid are $0$ and $\vartheta_{j'}^m(t)\leq \tilde{s}_i^m(t)$, respectively, according to our pricing scheme, and no VM is bought by cloud $i$; the utility gain is non-positive by bidding untruthfully, since there is non-negative utility loss, \emph{i.e.}, $[\tilde{b}_i^m(t)-\theta_2^m(t)]\cdot \hat{\gamma}_i^m(t)-[\vartheta_{j'}^m(t)-\tilde{s}_i^m(t)]\cdot \hat{\eta}_i^m(t)$ with $\tilde{b}_i^m(t)\geq \theta_2^m(t)$ and $\vartheta_{j'}^m(t)\leq \tilde{s}_i^m(t)$, by losing the buy-bid while winning the sell-bid. If it loses both bids, we have that the charged prices for buy-bid and sell-bid are both $0$, according to our pricing scheme, and no VM is bought by or sold by cloud $i$; the utility gain is non-positive by bidding untruthfully, since the charged sell-bid price remains the same while there is non-negative utility loss, \emph{i.e.}, $[\tilde{b}_i^m(t)-\theta_2^m(t)]\cdot \hat{\gamma}_i^m(t)$ with $\tilde{b}_i^m(t)\geq \theta_2^m(t)$, by losing the buy-bid.

\end{itemize}

\noindent \textbf{Case 3} -- \emph{Cloud $i$ wins sell-bid but loses buy-bid with truthful bidding}: In this case, the charged/paid prices for buy-bid and sell-bid are $0$ and $\vartheta_{j'}^m(t)$, respectively. we discuss the all possible cases of untruthful bidding as follows,

\begin{itemize}
\item \emph{Bid untruthfully with $b_i^m(t)>\tilde{b}_i^m(t)$ and $s_i^m(t)>\tilde{s}_i^m(t)$}: The cloud can either win both buy and sell bids, or win buy-bid only, or win sell-bid only, or lose both bids. If the cloud wins both bids, we have that the charged/paid prices for buy-bid and sell-bid are $\theta_2^m(t)\geq \tilde{b}_i^m(t)$ and $\vartheta_{j'}^m(t)$, respectively, according to our pricing scheme; the utility gain is non-positive by bidding untruthfully since the charged sell-bid price remains the same while there is non-negative utility loss, \emph{i.e.}, $[\theta_2^m(t)-\tilde{b}_i^m(t)]\cdot \hat{\gamma}_i^m(t)$ with $\tilde{b}_i^m(t)\leq \theta_2^m(t)$, by winning the buy-bid. If it only wins the buy-bid, we have that the charged prices for buy-bid and sell-bid are $\theta_2^m(t)$ and $0$, respectively, according to our pricing scheme, and no VM is sold by cloud $i$; the utility gain is non-positive by bidding untruthfully since there is non-negative utility loss, \emph{i.e.}, $[\theta_2^m(t)-\tilde{b}_i^m(t)]\cdot \hat{\gamma}_i^m(t)+[\vartheta_{j'}^m(t)-\tilde{s}_i^m(t)]\cdot \hat{\eta}_i^m(t)$ with $\tilde{b}_i^m(t)\leq \theta_2^m(t)$ and $\vartheta_{j'}^m(t)\geq \tilde{s}_i^m(t)$, by winning the buy-bid while losing the sell-bid. If it only wins the sell-bid, we have that the charged prices for buy-bid and sell-bid are $0$ and $\vartheta_{j'}^m(t)\leq \tilde{s}_i^m(t)$, respectively, according to our pricing scheme, and no VM is bought by cloud $i$; the utility gain is zero by bidding untruthfully since the charged/paid prices are the same with that by bidding truthfully. If it loses both bids, we have that the charged prices for buy-bid and sell-bid are both $0$, according to our pricing scheme, and no VM is bought by or sold by cloud $i$; the utility gain is non-positive by bidding untruthfully, since the charged buy-bid price remains the same while there is non-negative utility loss, \emph{i.e.}, $[\vartheta_{j'}^m(t)-\tilde{s}_i^m(t)]\cdot \hat{\eta}_i^m(t)$ with $\vartheta_{j'}^m(t)\geq \tilde{s}_i^m(t)$, by losing the sell-bid.

\item \emph{Bid untruthfully with $b_i^m(t)>\tilde{b}_i^m(t)$ and $s_i^m(t)=\tilde{s}_i^m(t)$}: Cloud $i$ still wins the sell-bid with the same bidding price. However, it may either win or lose the buy-bid. If it also wins the buy-bid, we have that the charged/paid prices for buy-bid and sell-bid are $\theta_2^m(t)$ and $\vartheta_{j'}^m(t)$, respectively, according to our pricing scheme; the utility gain is non-positive by bidding untruthfully since the charged sell-bid price remains the same while there is non-negative utility loss, \emph{i.e.}, $[\theta_2^m(t)-\tilde{b}_i^m(t)]\cdot \hat{\gamma}_i^m(t)$ with $\tilde{b}_i^m(t)\leq \theta_2^m(t)$, by winning the buy-bid. If it loses the buy-bid, we have that the charged prices for buy-bid and sell-bid are $0$ and $\vartheta_{j'}^m(t)$, respectively, according to our pricing scheme, and no VM is bought or sold by cloud $i$; the utility gain is non-positive by bidding untruthfully, since the charged/paid prices are the same with that by bidding truthfully.

\item \emph{Bid untruthfully with $b_i^m(t)>\tilde{b}_i^m(t)$ and $s_i^m(t)<\tilde{s}_i^m(t)$}: According to Lemma \ref{lemma:monotonic}, cloud $i$ still wins the sell-bid but may either win or lose the buy-bid. If it loses the buy-bid, we have that the charged prices for buy-bid and sell-bid are $\theta_2^m(t)$ and $\vartheta_{j'}^m(t)$, respectively, according to Lemma \ref{lemma:bid-independent} and our pricing scheme, and no VM is sold by cloud $i$; the utility gain is zero by bidding untruthfully since the charged/paid prices are the same with that by bidding truthfully. If it also wins the buy-bid, we have that the charged/paid prices for buy-bid and sell-bid are $\theta_2^m(t)\geq \tilde{b}_i^m(t)$ and $\vartheta_{j'}^m(t)$, respectively, according to Lemma \ref{lemma:bid-independent} and our pricing scheme; the utility gain is non-positive by bidding untruthfully, since the charged sell-bid price remains the same while there is non-negative utility loss, by winning the buy-bid.

\item \emph{Bid untruthfully with $b_i^m(t)=\tilde{b}_i^m(t)$ and $s_i^m(t)>\tilde{s}_i^m(t)$}: Cloud $i$ still loses the buy-bid with the same bidding price. However, it may either win or lose the sell-bid. If it loses the sell-bid, we have that the charged prices for buy-bid and sell-bid are $0$ and $0$, respectively, according to our pricing scheme, and no VM is sold by cloud $i$; the utility gain is non-positive by bidding untruthfully since the charged buy-bid price remains the same while there is non-negative utility loss, \emph{i.e.}, $[\vartheta_{j'}^m(t)-\tilde{s}_i^m(t)]\cdot \hat{\eta}_i^m(t)$ with $\vartheta_{j'}^m(t)\geq \tilde{s}_i^m(t)$, by losing the sell-bid. If it also wins the sell-bid, we have that the charged/paid prices for buy-bid and sell-bid are $0$ and $\vartheta_{j'}^m(t)\leq \tilde{s}_i^m(t)$, respectively, according to our pricing scheme; the utility gain is zero by bidding untruthfully since the charged/paid prices are the same with that by bidding truthfully.

\item \emph{Bid untruthfully with $b_i^m(t)=\tilde{b}_i^m(t)$ and $s_i^m(t)<\tilde{s}_i^m(t)$}: According to Lemma \ref{lemma:monotonic}, cloud $i$ still wins the sell-bid. It is also easy to see that cloud $i$ still loses the buy-bid with the same bidding price. Hence, the charged/paid prices for buy-bid and sell-bid are $0$ and $\vartheta_{j'}^m(t)$, respectively, according to Lemma \ref{lemma:bid-independent} and our pricing scheme. The utility gain is zero by bidding untruthfully.

\item \emph{Bid untruthfully with $b_i^m(t)<\tilde{b}_i^m(t)$ and $s_i^m(t)>\tilde{s}_i^m(t)$}: It is easy to see that cloud $i$ still loses the buy-bid, since otherwise we will have a contradiction to Lemma \ref{lemma:monotonic}. However, it may either win or lose the sell-bid. If it wins the sell-bid, we have that the charged/paid prices for buy-bid and sell-bid are $0$ and $\vartheta_{j'}^m(t)$, respectively, according to our pricing scheme; the utility gain is zero by bidding untruthfully since the charged/paid prices are the same with that by bidding truthfully. If it loses the sell-bid, we have that the charged prices for buy-bid and sell-bid are both $0$, according to our pricing scheme, and no VM is bought or sold by cloud $i$; the utility gain is non-positive by bidding untruthfully, since the charged buy-bid price remains the same while there is non-negative utility loss, by losing the sell-bid.

\item \emph{Bid untruthfully with $b_i^m(t)<\tilde{b}_i^m(t)$ and $s_i^m(t)=\tilde{s}_i^m(t)$}: Cloud $i$ still wins the sell-bid with the same bidding price. It is also easy to see that cloud $i$ still loses the buy-bid, since otherwise we will have a contradiction to Lemma \ref{lemma:monotonic}. Hence, the charged/paid prices for buy-bid and sell-bid are $0$ and $\vartheta_{j'}^m(t)$, respectively, according to our pricing scheme. The utility gain is zero by bidding untruthfully.

\item \emph{Bid untruthfully with $b_i^m(t)<\tilde{b}_i^m(t)$ and $s_i^m(t)<\tilde{s}_i^m(t)$}: According to Lemma \ref{lemma:monotonic}, cloud $i$ still wins the sell-bid. It is also easy to see that cloud $i$ still loses the buy-bid, since otherwise we will have a contradiction to Lemma \ref{lemma:monotonic}. Hence, the charged/paid prices for buy-bid and sell-bid are $0$ and $\vartheta_{j'}^m(t)$, respectively, according to Lemma \ref{lemma:bid-independent} and our pricing scheme. The utility gain is zero by bidding untruthfully.

\end{itemize}

\noindent \textbf{Case 4} -- \emph{Cloud $i$ loses both buy-bid and sell-bid with truthful bidding}: In this case, the charged/paid prices for buy-bid and sell-bid are both $0$. we discuss the all possible cases of untruthful bidding as follows,

\begin{itemize}
\item \emph{Bid untruthfully with $b_i^m(t)>\tilde{b}_i^m(t)$ and $s_i^m(t)>\tilde{s}_i^m(t)$}: It is easy to see that cloud $i$ still loses the sell-bid, since otherwise we will have a contradiction to Lemma \ref{lemma:monotonic}. However, it may either win or lose the buy-bid. If it wins the buy-bid, we have that the charged/paid prices for buy-bid and sell-bid are $\theta_2^m(t)\geq \tilde{b}_i^m(t)$ and $0$, respectively, according to our pricing scheme; the utility gain is non-positive by bidding untruthfully there is non-negative utility loss, \emph{i.e.}, $[\theta_2^m(t)-\tilde{b}_i^m(t)]\cdot \hat{\gamma}_i^m(t)$ with $\tilde{b}_i^m(t)\leq \theta_2^m(t)$, by winning the buy-bid. If it loses the buy-bid, we have that the charged prices for buy-bid and sell-bid are both $0$, according to our pricing scheme, and no VM is bought or sold by cloud $i$; the utility gain is zero by bidding untruthfully.

\item \emph{Bid untruthfully with $b_i^m(t)>\tilde{b}_i^m(t)$ and $s_i^m(t)=\tilde{s}_i^m(t)$}: Cloud $i$ still loses the sell-bid with the same bidding price. However, it may either win or lose the buy-bid. If it loses the buy-bid, we have that the charged prices for buy-bid and sell-bid are both $0$, according to our pricing scheme, and no VM is sold by cloud $i$; the utility gain is zero by bidding untruthfully. If it also wins the buy-bid, we have that the charged/paid prices for buy-bid and sell-bid are $\theta_2^m(t)\geq \tilde{b}_i^m(t)$ and $0$, respectively, according to our pricing scheme; the utility gain is non-positive by bidding untruthfully there is non-negative utility loss, by winning the buy-bid.

\item \emph{Bid untruthfully with $b_i^m(t)>\tilde{b}_i^m(t)$ and $s_i^m(t)<\tilde{s}_i^m(t)$}: The cloud can either win both buy and sell bids, or win buy-bid only, or win sell-bid only, or lose both bids. If the cloud wins both bids, we have that the charged/paid prices for buy-bid and sell-bid are $\theta_2^m(t)\geq \tilde{b}_i^m(t)$ and $\vartheta_{j'}^m(t)\leq \tilde{s}_i^m(t)$, respectively, according to our pricing scheme; the utility gain is non-positive by bidding untruthfully since there is non-negative utility loss, \emph{i.e.}, $[\theta_2^m(t)-\tilde{b}_i^m(t)]\cdot \hat{\gamma}_i^m(t) + [\tilde{s}_i^m(t)-\vartheta_{j'}^m(t)]\cdot \hat{\eta}_i^m(t)$ with $\tilde{b}_i^m(t)\leq \theta_2^m(t)$ and $\vartheta_{j'}^m(t)\leq \tilde{s}_i^m(t)$, by winning the buy-bid and sell-bid. If it only wins the buy-bid, we have that the charged prices for buy-bid and sell-bid are $\theta_2^m(t)$ and $0$, respectively, according to our pricing scheme, and no VM is sold by cloud $i$; the utility gain is non-positive by bidding untruthfully since the paid price for sell-bid is the same with that by bidding truthfully and there is non-negative utility loss, by winning the buy-bid. If it only wins the sell-bid, we have that the charged prices for buy-bid and sell-bid are $0$ and $\vartheta_{j'}^m(t)\leq \tilde{s}_i^m(t)$, respectively, according to our pricing scheme, and no VM is bought by cloud $i$; the utility gain is non-positive by bidding untruthfully since the charged price for the buy-bid is the same with that by bidding truthfully while there is non-negative utility loss, \emph{i.e.}, $[\tilde{s}_i^m(t)-\vartheta_{j'}^m(t)]\cdot \hat{\eta}_i^m(t)$ with $\vartheta_{j'}^m(t)\leq \tilde{s}_i^m(t)$, by winning the sell-bid. If it loses both bids, we have that the charged prices for buy-bid and sell-bid are both $0$, according to our pricing scheme, and no VM is bought by or sold by cloud $i$; the utility gain is zero by bidding untruthfully.

\item \emph{Bid untruthfully with $b_i^m(t)=\tilde{b}_i^m(t)$ and $s_i^m(t)>\tilde{s}_i^m(t)$}: It is easy to see that cloud $i$ still loses the sell-bid, since otherwise we will have a contradiction to Lemma \ref{lemma:monotonic}. It is also not hard to find that cloud $i$ still loses the buy-bid with the same bidding price. Hence, the charged/paid prices for buy-bid and sell-bid are both $0$, according to our pricing scheme. The utility gain is zero by bidding untruthfully.

\item \emph{Bid untruthfully with $b_i^m(t)=\tilde{b}_i^m(t)$ and $s_i^m(t)<\tilde{s}_i^m(t)$}: It is not hard to find that cloud $i$ still loses the buy-bid with the same bidding price. However, it may either win or lose the sell-bid. If it wins the sell-bid, we have that the charged/paid prices for buy-bid and sell-bid are $0$ and $\vartheta_{j'}^m(t)\leq \tilde{s}_i^m(t)$, respectively, according to our pricing scheme; the utility gain is non-positive by bidding untruthfully since the charged price for buy-bid is the same with that by bidding truthfully while there is non-negative utility loss, by winning the sell-bid. If it loses the sell-bid, we have that the charged prices for buy-bid and sell-bid are both $0$, according to our pricing scheme, and no VM is bought or sold by cloud $i$; the utility gain is zero by bidding untruthfully.

\item \emph{Bid untruthfully with $b_i^m(t)<\tilde{b}_i^m(t)$ and $s_i^m(t)>\tilde{s}_i^m(t)$}: It is easy to see that cloud $i$ still loses both the buy-bid and the sell-bid, since otherwise we will have a contradiction to Lemma \ref{lemma:monotonic}. Hence, the charged/paid prices for buy-bid and sell-bid are both $0$, according to our pricing scheme. The utility gain is zero by bidding untruthfully.

\item \emph{Bid untruthfully with $b_i^m(t)<\tilde{b}_i^m(t)$ and $s_i^m(t)=\tilde{s}_i^m(t)$}: It is easy to see that cloud $i$ still loses the buy-bid, since otherwise we will have a contradiction to Lemma \ref{lemma:monotonic}. It is also not hard to find that cloud $i$ still loses the sell-bid with the same bidding price. Hence, the charged/paid prices for buy-bid and sell-bid are both $0$, according to our pricing scheme. The utility gain is zero by bidding untruthfully.

\item \emph{Bid untruthfully with $b_i^m(t)<\tilde{b}_i^m(t)$ and $s_i^m(t)<\tilde{s}_i^m(t)$}: It is easy to see that cloud $i$ still loses the buy-bid, since otherwise we will have a contradiction to Lemma \ref{lemma:monotonic}. However, it may either win or lose the sell-bid. If it wins the sell-bid, we have that the charged/paid prices for buy-bid and sell-bid are $0$ and $\vartheta_{j'}^m(t)\leq \tilde{s}_i^m(t)$, respectively, according to our pricing scheme; the utility gain is non-positive by bidding untruthfully since the charged price for buy-bid is the same with that by bidding truthfully while there is non-negative utility loss, by winning the sell-bid. If it loses the sell-bid, we have that the charged prices for buy-bid and sell-bid are both $0$, according to our pricing scheme, and no VM is bought or sold by cloud $i$; the utility gain is zero by bidding untruthfully.

\end{itemize}

To conclude, we have shown that bidding truthfully is the dominant strategy of each cloud.

\section{Proof to Theorem \ref{theorem:rationality}}\label{appendix:rationality}

We prove the individual rationality for buy-bids and sell-bids, respectively.

\noindent \textbf{Winner of buy-bid}: If cloud $i$ wins the buy-bid for VM type $m$, we have that $b_i^m(t)$ is the largest buy-bid price among all buy-bids, and $b_i^m(t)\geq \theta_2^m(t)$, according to our winner determination scheme. Since $\hat{b}_i^m(t) = \theta_2^m(t)$ according to our pricing scheme, we have that $\hat{b}_i^m(t)\leq b_i^m(t)$.

\noindent \textbf{Winner of sell-bid}: If cloud $i$ wins the sell-bid for VM type $m$, we have that $s_i^m(t)$ is among the $(j'-1)^{th}$ lowest sell-bid prices of all sell-bids, and $s_i^m(t)\leq \vartheta_{j'}^m(t)$ according to our winner determination scheme. Since $\hat{s}_i^m(t) = \vartheta_{j'}^m(t)$ according to our pricing scheme, we have that $\hat{s}_i^m(t)\geq s_i^m(t)$.

\section{Proof to Theorem \ref{theorem:budget}}\label{appendix:budget}
We first calculate the total payment from the buyers and the total price paid to the sellers, respectively. We then show that the ex-post budget balance is guaranteed.

\noindent \textbf{Total payment from the buyers}: According to our winner determination scheme, only the buyer with largest bidding price wins the buy-bid. With the pricing scheme in Eqn.~(\ref{eqn:buy-price}), the charged price is $\theta_2^m(t)$ for each bought VM. With the allocation scheme in Eqn.~(\ref{eqn:buy-vol}), the overall number of bought VMs is $\sum_{j=1}^{j'-1}L_j^m(t)$. Hence, the total payment from the buyers is that

\vspace{-4mm}{\small
\begin{align*}
\sum_{i\in [1, F]}[\hat{b}_i^m(t)\cdot \hat{\gamma}_i^m(t)] = \theta_2^m(t)\cdot\sum_{j=1}^{j'-1}L_j^m(t).
\end{align*}}\vspace{-4mm}

\noindent \textbf{Overall price paid to the sellers}: According to our winner determination scheme, only the seller with $(j'-1)^{th}$ lowest bidding price wins the sell-bid. With the pricing scheme in Eqn.~(\ref{eqn:sell-price}), the paid price is $\vartheta_{j'}^m(t)$ for each sold VM. With the allocation scheme in Eqn.~(\ref{eqn:sell-vol}), the overall number of sold VMs is $\sum_{j=1}^{j'-1}L_j^m(t)$. Hence, the total payment to the sellers is that

\vspace{-4mm}{\small
\begin{align*}
\sum_{i\in [1, F]}[\hat{s}_i^m(t)\cdot \hat{\eta}_i^m(t)] = \vartheta_{j'}^m(t)\cdot\sum_{j=1}^{j'-1}L_j^m(t).
\end{align*}}\vspace{-4mm}

According to Eqn.~(\ref{eqn:j'}) for the winner determination scheme, we know that $\vartheta_{j'}^m(t) \leq \theta_2^m(t)$. Hence, the ex-post budget balance at the auctioneer for each VM type $m\in [1, M]$ can be guaranteed as follows,

\vspace{-4mm}{\small
\begin{align*}
&\sum_{i\in [1, F]}[\hat{b}_i^m(t)\cdot \hat{\gamma}_i^m(t)-\hat{s}_i^m(t)\cdot \hat{\eta}_i^m(t)]\\
=&\sum_{i\in [1, F]}[\hat{b}_i^m(t)\cdot \hat{\gamma}_i^m(t)] - \sum_{i\in [1, F]}[\hat{s}_i^m(t)\cdot \hat{\eta}_i^m(t)]\\
=&\theta_2^m(t)\cdot\sum_{j=1}^{j'-1}L_j^m(t) - \vartheta_{j'}^m(t)\cdot\sum_{j=1}^{j'-1}L_j^m(t)\\
\geq & 0.
\end{align*}}\vspace{-3mm}

\section{Proof to Lemma \ref{lemma:bounded-queue}}\label{appendix:bounded-queue}

We prove the lemma by induction.

\vspace{1mm}\noindent\textbf{Induction Basis}: At time slot $0$, the beginning of the federation,
all queues are empty. Therefore,

\vspace*{-5mm}{\small
\begin{align*}
&Q_i^s(0)=0 \leq Q_i^{s(max)},\ \forall i\in [1, F], s\in [1, S],\\
&Z_i^s(0)=0 \leq Z_i^{s(max)},\ \forall i\in [1, F], s\in [1, S].
\end{align*}}\vspace*{-5mm}

\noindent\textbf{Induction Step}: Suppose that, at time slot $t\geq 0$, $Q_i^s(t) \leq Q_i^{s(max)}$ and $Z_i^s(t) \leq Z_i^{s(max)}$, $\forall
i\in [1, F], s\in [1, S]$. Then, for any
$Q_i^s(t)$ and $Z_i^s(t)$, we have the
following possible cases.
\begin{itemize}
\item $0\leq Q_i^s(t)\leq V \xi_i^s$ or $V \xi_i^s < Q_i^s(t) \leq V \xi_i^s + R_i^s$;
\item $0\leq Z_i^s(t)\leq V \xi_i^s$ or $V \xi_i^s < Z_i^s(t) \leq V \xi_i^s + \epsilon_s$.
\end{itemize}

\noindent - We first analyze the size of $Q_i^s(t+1)$:
\begin{itemize}
\item If $0\leq Q_i^s(t)\leq V \xi_i^s$, we have that

\vspace{-4mm}{\small
\begin{align*}
Q_i^s(t+1)&=\max\{Q_i^s(t)-\sum_{j=1}^F \mu_{ij}^{s}(t)- D_i^{s}(t),0\}+ r_i^{s}(t)\\
        &\leq \max\{Q_i^s(t),0\}+R_i^s\\
        &\leq V \xi_i^s + R_i^s=Q_i^{s(max)}
\end{align*}}\vspace{-4mm}

\noindent according to the queueing law (\ref{eqn:queue1}). The first inequality is based on the fact that $0\leq r_i^s(t)\leq R_i^s$.

\item If $V \xi_i^s < Q_i^s(t) \leq V \xi_i^s + R_i^s$, we have that

\vspace{-4mm}{\small
\begin{align*}
\begin{split}
D_i^s(t)=D_i^{s(max)},
\end{split}
\end{align*}}\vspace{-4mm}

\noindent according to the job drop decision with Eqn.~(\ref{eqn:drop}).

Hence, we have that

\vspace{-4mm}{\small
\begin{align*}
Q_i^s(t+1)&=\max\{Q_i^s(t)-\sum_{j=1}^F \mu_{ij}^{s}(t)- D_i^{s}(t),0\}+ r_i^{s}(t)\\
        &\leq \max\{V \xi_i^s + R_i^s-D_i^{s(max)},0\}+R_i^s\\
        &\leq V \xi_i^s + R_i^s\leq Q_i^{s(max)}.
\end{align*}}\vspace{-4mm}

\noindent The second inequality is based on the fact that $D_i^{s(max)}\geq \max\{R_i^s, \epsilon_s\}$.

\end{itemize}

So far, $Q_i^s(t)\leq Q_i^{s(max)},\ \forall i\in [1, F], s\in [1, S]$ for each
time slot $t$ is proved.

\noindent - We next analyze the size of $Z_i^s(t+1)$:
\begin{itemize}
\item If $0\leq Z_i^s(t)\leq V \xi_i^s$, we have that

\vspace{-4mm}{\small
\begin{align*}
Z_i^s(t+1)&=\max\{Z_i^s(t)+ \mathbf{1}_{\{Q_i^{s}(t)> 0\}}\cdot [\epsilon_{s} - \sum_{j=1}^F\mu_{ij}^{s}(t)]- D_i^{s}(t)\nonumber\\
        & - \mathbf{1}_{\{Q_i^{s}(t)=0\}}\cdot \sum_{j=1}^F \frac{C_{j}^{m_s}\cdot N_{j}^{m_s}}{g_s}, 0\}\\
        &\leq \max\{Z_i^s(t)+\epsilon_{s},0\}\\
        &\leq V \xi_i^s + \epsilon_{s}=Z_i^{s(max)},
\end{align*}}\vspace{-4mm}

\noindent according to the queueing law (\ref{eqn:queue2}). The first inequality is based on the fact that $\epsilon_{s}>0$.

\item If $V \xi_i^s < Z_i^s(t) \leq V \xi_i^s + \epsilon_{s}$, we have that

\vspace{-4mm}{\small
\begin{align*}
\begin{split}
D_i^s(t)=D_i^{s(max)},
\end{split}
\end{align*}}\vspace{-4mm}

\noindent according to the job drop decision with Eqn.~(\ref{eqn:drop}).

Hence, we have that

\vspace{-4mm}{\small
\begin{align*}
Z_i^s(t+1)&=\max\{Z_i^s(t)+ \mathbf{1}_{\{Q_i^{s}(t)> 0\}}\cdot [\epsilon_{s} - \sum_{j=1}^F\mu_{ij}^{s}(t)]- D_i^{s}(t)\nonumber\\
        & - \mathbf{1}_{\{Q_i^{s}(t)=0\}}\cdot \sum_{j=1}^F \frac{C_{j}^{m_s}\cdot N_{j}^{m_s}}{g_s}, 0\}\\
        &\leq \max\{Z_i^s(t)+\epsilon_{s}-D_i^{s(max)},0\}\\
        &\leq V \xi_i^s + \epsilon_{s}=Z_i^{s(max)}.
\end{align*}}\vspace{-4mm}

\noindent The second inequality is based on the fact that $D_i^{s(max)}\geq \max\{R_i^s, \epsilon_s\}$.

\end{itemize}

So far, $Z_i^s(t)\leq Z_i^{s(max)},\ \forall i\in [1, F], s\in [1, S]$ for each
time slot $t$ is proved.

In conclusion, Lemma \ref{lemma:bounded-queue} is proven.

\section{Proof to Theorem \ref{theorem:SLA}}\label{appendix:SLA}

We prove this theorem by contradiction.

For each cloud $i\in [1, F]$ and each service type $s\in [1, S]$, the job requests arrive at time slot $t\geq 0$ is $r_i^s(t)$ and the earliest time they can depart the queue is $t+1$. We show that all these jobs depart (by being either scheduled or dropped) on or before time $t+d_s$.

Suppose this is not true, we will come to a contradiction. We must have that $Q_i^s(\tau)>0$ for all $\tau\in [t+1,\ldots, t+d_s]$ (otherwise, all the jobs are scheduled by time $t+d_s$). With the queueing law in Eqn.~(\ref{eqn:queue2}), we have that

\vspace{-4mm}{\small
\begin{align*}
Z_i^s(\tau+1)&=\max\{Z_i^s(\tau)+ \epsilon_{s} - \sum_{j=1}^F\mu_{ij}^{s}(\tau)- D_i^{s}(\tau), 0\}\\
          &\geq Z_i^s(\tau)+ \epsilon_{s} - \sum_{j=1}^F\mu_{ij}^{s}(\tau)- D_i^{s}(\tau).
\end{align*}}\vspace{-4mm}

Summing the above over $\tau\in [t+1,\ldots, t+d_s]$, we have that

\vspace{-4mm}{\small
\begin{align*}
Z_i^s(t+d_s+1)-Z_i^s(t+1)&\geq \epsilon_s d_s - \sum_{\tau=t+1}^{t+d_s}[\sum_{j=1}^F\mu_{ij}^{s}(\tau)+ D_i^{s}(\tau)].
\end{align*}}\vspace{-4mm}

Rearranging the above inequality and using the fact that $Z_i^s(t+d_s+1)\leq Z_i^{s(max)}$ and $Z_i^s(t+1)\geq 0$, we have that

\vspace{-4mm}{\small
\begin{align}\label{eqn:proof1}
\epsilon_s d_s - Z_i^{s(max)}\leq \sum_{\tau=t+1}^{t+d_s}[\sum_{j=1}^F\mu_{ij}^{s}(\tau)+ D_i^{s}(\tau)].
\end{align}}\vspace{-4mm}

Since the jobs are scheduled in a FIFO fashion, the jobs $r_i^s(t)$ that arrive at slot $t$ are placed at the end of the queue at slot $t+1$, and should be fully cleared only when all the jobs backlogged in $Q_i^s(t+1)$ have been scheduled. That is, the last job of $r_i^s(t)$ departs on slot $t+T$ with $T>0$ as the smallest integer satisfying $\sum_{\tau=t+1}^{t+T}[\sum_{j=1}^F\mu_{ij}^{s}(\tau)+ D_i^{s}(\tau)]\geq Q_i^s(t+1)$. Based on our assumption that not all of the $r_i^s(t)$ jobs depart by time $t+d_s$, we must have that

\vspace{-4mm}{\small
\begin{align}\label{eqn:proof2}
\sum_{\tau=t+1}^{t+d_s}[\sum_{j=1}^F\mu_{ij}^{s}(\tau)+ D_i^{s}(\tau)]< Q_i^s(t+1) \leq Q_i^{s(max)}.
\end{align}}\vspace{-4mm}

Combining Eqn.~(\ref{eqn:proof1}) and (\ref{eqn:proof2}), we have that

\vspace{-4mm}{\small
\begin{align}\label{eqn:proof2}
&\epsilon_s d_s - Z_i^{s(max)}< Q_i^{s(max)}\\
\Rightarrow &\epsilon_s<\frac{Q_i^{s(max)}+Z_i^{s(max)}}{d_s}.
\end{align}}\vspace{-4mm}

This contradicts with the given fact that $\epsilon_s= \frac{Q_i^{s(max)}+Z_i^{s(max)}}{d_s}$. Hence, we have proved that each job of type $s\in [1, S]$ is either scheduled or dropped with Alg.~\ref{alg:profit} before its maximum response delay $d_s$, if we set $\epsilon_s = \frac{Q_i^{s(max)}+Z_i^{s(max)}}{d_s}$.

\section{Proof to Theorem \ref{theorem:profit}}\label{appendix:profit}
Since the system status, \emph{i.e.}, the job arrival $r_i^s(t)$ and service pricing $p_i^s(t)$ of each type $s\in [1, S]$ of services and the operational price $\beta_i(t)$ at each cloud $i\in [1, F]$, changes with ergodic processes, we have that there exists a stationary randomized algorithm \cite{book2010}, which dynamically decides the VM valuation \& pricing (with $b_i^{m*}(t)$, $\gamma_i^{m*}(t)$, $s_i^{m*}(t)$, $\eta_i^{m*}(t)$), job scheduling (with $\mu_{ij}^{s*}(t)$) \& dropping (with $D_i^{s*}(t)$) and server provisioning (with $n_i^{m*}(t)$) at each cloud $i$, such that the offline optimum of cloud $i$'s individual profit $\Omega_i^*$ can be achieved, together with $\bar{r}_i^s\leq \sum_{j=1}^F \bar{\mu}_{ij}^{s*}+ \bar{D}_i^{s*}$ and $\epsilon_s\leq \sum_{j=1}^F \bar{\mu}_{ij}^{s*}+ \bar{D}_i^{s*}$. Here, $\bar{a}$ denotes the time averaged expectation of variable $a(t)$.

Based on the derivations of the optimization problem (\ref{eqn:profit-oneslot}) in and its solution in Sec.~\ref{sec:algorithm}, we know that Algorithm \ref{alg:profit} minimizes the right-hand-side of the drift-plus-penalty (drift-minus-utility) inequality in Eqn.~(\ref{eqn:drift-plus-penalty1}) at each slot $t$, with individual profit maximization as the utility, over all possible algorithms. Then, we can have that

\vspace*{-5mm}{\small\begin{align*}
&\Delta(\Theta_i(t))-V\cdot [\sum_{m\in [1, M]} [\hat{s}_i^m(t) \hat{\eta}_i^m(t) - \hat{b}_i^m(t) \hat{\gamma}_i^m(t)-\beta_i(t) n_i^m(t)] \\ &- \sum_{s\in [1, S]}D_i^s(t) \xi_i^s + \sum_{s\in [1, S]}r_i^s(t)p_i^s(t)]
\end{align*}
\vspace*{-5mm}

\begin{align*}
\leq&  B_i + \sum_{s\in [1, S]}[Q_i^s(t) r_i^s(t) + Z_i^s(t) \epsilon_s]+ \sum_{s\in [1, S]}[r_i^s(t)p_i^s(t)]\\
    &- V \sum_{m\in [1, M]} [\hat{s}_i^{m*}(t) \hat{\eta}_i^{m*}(t) - \hat{b}_i^{m*}(t) \hat{\gamma}_i^{m*}(t)-\beta_i(t) n_i^{m*}(t)]\\
    &-\sum_{s=\in [1, S]}\sum_{j\in [1, F]}\mu_{ij}^{s*}(t) [Q_i^s(t) + Z_i^s(t)]\\
    &-\sum_{s\in [1, S]}D_i^{s*}(t)[Q_i^s(t)+Z_i^s(t)-V\cdot \xi_i^s]\\
=& B_i - V\cdot [\sum_{m\in [1, M]} [\hat{s}_i^{m*}(t) \hat{\eta}_i^{m*}(t) - \hat{b}_i^{m*}(t) \hat{\gamma}_i^{m*}(t)-\beta_i(t) n_i^{m*}(t)] \\ &- \sum_{s\in [1, S]}D_i^{s*}(t) \xi_i^s+ \sum_{s\in [1, S]}r_i^s(t)p_i^s(t)] \\
    & +\sum_{s\in [1, S]}[Q_i^s(t)[r_i^s(t)-\sum_{j\in [1, F]} \mu_{ij}^{s*}(t)-D_i^{s*}(t)]\\ &+ Z_i^s(t)[\epsilon_s-\sum_{j\in [1, F]} \mu_{ij}^{s*}(t)-D_i^{s*}(t)]].
\end{align*}}\vspace*{-4mm}

Taking conditional expectations over queue status $\Theta_i(t)$ on both sides the inequality, we have that

\vspace*{-5mm}{\small\begin{align*}
&\mathbb{E}\{\Delta(\Theta_i(t))|\Theta_i(t)\}-V\cdot \mathbb{E}\{[\sum_{m\in [1, M]} [\hat{s}_i^m(t) \hat{\eta}_i^m(t) - \hat{b}_i^m(t) \hat{\gamma}_i^m(t)\\ & -\beta_i(t) n_i^m(t)] - \sum_{s\in [1, S]}D_i^s(t) \xi_i^s+ \sum_{s\in [1, S]}[r_i^s(t)p_i^s(t)]]|\Theta_i(t) \} \\
\leq& B_i - V\cdot \mathbb{E}\{[\sum_{m\in [1, M]} [\hat{s}_i^{m*}(t) \hat{\eta}_i^{m*}(t) - \hat{b}_i^{m*}(t) \hat{\gamma}_i^{m*}(t)-\beta_i(t) n_i^{m*}(t)] \\ &- \sum_{s\in [1, S]}D_i^{s*}(t) \xi_i^s+ \sum_{s\in [1, S]}[r_i^s(t)p_i^s(t)]]|\Theta_i(t)\} \\
    & +\sum_{s\in [1, S]}[Q_i^s(t)\mathbb{E}\{[r_i^s(t)-\sum_{j=1}^F \mu_{ij}^{s*}(t)-D_i^{s*}(t)]|\Theta_i(t)\}\\ &+ Z_i^s(t)\mathbb{E}\{[\epsilon_s-\sum_{j=1}^F \mu_{ij}^{s*}(t)-D_i^{s*}(t)]|\Theta_i(t)\}].
\end{align*}}\vspace*{-5mm}

Next, we take expectations on both sides the inequality, we have that

\vspace*{-5mm}{\small\begin{align*}
&\mathbb{E}\{\Delta(\Theta_i(t))\}-V\cdot \mathbb{E}\{[\sum_{m\in [1, M]} [\hat{s}_i^m(t) \hat{\eta}_i^m(t) - \hat{b}_i^m(t) \hat{\gamma}_i^m(t) \\ & -\beta_i(t) n_i^m(t)] - \sum_{s\in [1, S]}D_i^s(t) \xi_i^s+ \sum_{s\in [1, S]}[r_i^s(t)p_i^s(t)]] \} \\
\leq& B_i - V\cdot \mathbb{E}\{[\sum_{m\in [1, M]} [\hat{s}_i^{m*}(t) \hat{\eta}_i^{m*}(t) - \hat{b}_i^{m*}(t) \hat{\gamma}_i^{m*}(t)-\beta_i(t) n_i^{m*}(t)] \\ &- \sum_{s\in [1, S]}D_i^{s*}(t) \xi_i^s+ \sum_{s\in [1, S]}[r_i^s(t)p_i^s(t)]]\} \\
    & +\sum_{s\in [1, S]}[Q_i^s(t)\mathbb{E}\{[r_i^s(t)-\sum_{j=1}^F \mu_{ij}^{s*}(t)-D_i^{s*}(t)]\}\\ &+ Z_i^s(t)\mathbb{E}\{[\epsilon_s-\sum_{j=1}^F \mu_{ij}^{s*}(t)-D_i^{s*}(t)]\}].
\end{align*}}\vspace*{-5mm}

By summing over the $T$ slots on both sides of the inequality, we have that

\vspace*{-5mm}{\small\begin{align*}
&\mathbb{E}\{L(\Theta_i(T))\}-\mathbb{E}\{L(\Theta_i(0))\}-V\cdot \sum_{t=0}^{T-1}\mathbb{E}\{[\sum_{m\in [1, M]} [\hat{s}_i^m(t) \hat{\eta}_i^m(t) \\ & - \hat{b}_i^m(t) \hat{\gamma}_i^m(t)-\beta_i(t) n_i^m(t)] - \sum_{s\in [1, S]}D_i^s(t) \xi_i^s+ \sum_{s\in [1, S]}[r_i^s(t)p_i^s(t)]] \} \\
\leq& T\cdot B_i - V\cdot \sum_{t=0}^{T-1} \mathbb{E}\{[\sum_{m\in [1, M]} [\hat{s}_i^{m*}(t) \hat{\eta}_i^{m*}(t) - \hat{b}_i^{m*}(t) \hat{\gamma}_i^{m*}(t)\\ & -\beta_i(t) n_i^{m*}(t)] - \sum_{s\in [1, S]}D_i^{s*}(t) \xi_i^s+ \sum_{s\in [1, S]}[r_i^s(t)p_i^s(t)] ]\} \\
    & +\sum_{t=0}^{T-1} \sum_{s\in [1, S]}[Q_i^s(t)\mathbb{E}\{[r_i^s(t)-\sum_{j=1}^F \mu_{ij}^{s*}(t)-D_i^{s*}(t)]\}\\ &+ Z_i^s(t)\mathbb{E}\{[\epsilon_s-\sum_{j=1}^F \mu_{ij}^{s*}(t)-D_i^{s*}(t)]\}].
\end{align*}}\vspace*{-4mm}

Since $\mathbb{E}\{L(\Theta_i(T))\}\geq 0$ and $\mathbb{E}\{L(\Theta_i(0))\}=0$ according to the definition of the Lyapunov function, we have that

\vspace*{-5mm}{\small\begin{align*}
&-V\cdot \sum_{t=0}^{T-1}\mathbb{E}\{[\sum_{m\in [1, M]} [\hat{s}_i^m(t) \hat{\eta}_i^m(t) - \hat{b}_i^m(t) \hat{\gamma}_i^m(t)-\beta_i(t) n_i^m(t)] \\ & - \sum_{s\in [1, S]}D_i^s(t) \xi_i^s+ \sum_{s\in [1, S]}[r_i^s(t)p_i^s(t)]] \} \\
\leq& T\cdot B_i - V\cdot \sum_{t=0}^{T-1} \mathbb{E}\{[\sum_{m\in [1, M]} [\hat{s}_i^{m*}(t) \hat{\eta}_i^{m*}(t) - \hat{b}_i^{m*}(t) \hat{\gamma}_i^{m*}(t)\\ & -\beta_i(t) n_i^{m*}(t)] - \sum_{s\in [1, S]}D_i^{s*}(t) \xi_i^s+ \sum_{s\in [1, S]}[r_i^s(t)p_i^s(t)] ]\}\\
    & +\sum_{t=0}^{T-1} \sum_{s\in [1, S]}[Q_i^s(t)\mathbb{E}\{[r_i^s(t)-\sum_{j=1}^F \mu_{ij}^{s*}(t)-D_i^{s*}(t)]\}\\ &+ Z_i^s(t)\mathbb{E}\{[\epsilon_s-\sum_{j=1}^F \mu_{ij}^{s*}(t)-D_i^{s*}(t)]\}].
\end{align*}}\vspace*{-4mm}

Dividing $T\cdot V$ on both sides of the above inequality and taking limitation on $T$ to infinity, we have that

\vspace*{-5mm}{\small\begin{align*}
-\Omega_i \leq& B_i/V - \Omega_i^*  +\sum_{s\in [1, S]}[\bar{Q}_i^s[\bar{r}_i^s-\sum_{j=1}^F \bar{\mu}_{ij}^{s*}-\bar{D}_i^{s*}]\\ &+ \bar{Z}_i^s[\epsilon_s-\sum_{j=1}^F \bar{\mu}_{ij}^{s*}-\bar{D}_i^{s*}]]\\
\leq& B_i/V - \Omega_i^*.
\end{align*}}\vspace*{-5mm}

The second inequality comes from the fact that $\bar{r}_i^s-\sum_{j=1}^F \bar{\mu}_{ij}^{s*}-\bar{D}_i^{s*}\leq 0$ and  $\epsilon_s-\sum_{j=1}^F \bar{\mu}_{ij}^{s*}-\bar{D}_i^{s*}\leq 0$. Rearranging the two sides, we have that

\vspace{-4mm}{\small
\begin{align*}
\Omega_i \geq \Omega_i^* - B_i/V,
\end{align*}}\vspace{-6mm}

\section{Proof to Theorem \ref{theorem:welfare-alg2}}\label{appendix:social}


Similar with the proof to the optimality of individual profit with Alg.~\ref{alg:profit}, we have the following proof to the optimality in social welfare for our benchmark algorithm. Since the system status, \emph{i.e.}, the job arrival $r_i^s(t)$ and service pricing $p_i^s(t)$ of each type $s\in [1, S]$ of services and the operational price $\beta_i(t)$ at each cloud $i\in [1, F]$, changes with ergodic processes, we have that there exists a stationary randomized algorithm \cite{book2010}, which dynamically decides the job scheduling (with $\mu_{ij}^{s*}(t)$) \& dropping (with $D_i^{s*}(t)$) and server provisioning (with $n_i^{m*}(t)$) at each cloud $i$, such that the offline optimum of the federation's social welfare $\prod^*$ can be achieved, together with $\bar{r}_i^{s*}\leq \sum_{j=1}^F \bar{\mu}_{ij}^{s*}+ \bar{D}_i^{s*}$ and $\epsilon_s\leq \sum_{j=1}^F \bar{\mu}_{ij}^{s*}+ \bar{D}_i^{s*}$. Here, $\bar{a}$ denotes the time averaged expectation of variable $a(t)$.

Based on the above derivations of the optimization problem (\ref{eqn:profit-oneslot2}) and its solution in Alg.~\ref{alg:social}, we know that Algorithm \ref{alg:social} minimizes the right-hand-side of the drift-plus-penalty (drift-minus-welfare) inequality in Eqn.~(\ref{eqn:drift-plus-penalty2}) at each slot $t$, with social welfare maximization as the utility, over all possible algorithms. Then, we can have that

\vspace*{-5mm}{\small\begin{align*}
&\Delta(\Theta(t))+V\cdot \sum_{i\in [1, F]}[\sum_{m\in [1, M]} [\beta_i(t) n_i^m(t)] + \sum_{s\in [1, S]}D_i^s(t) \xi_i^s \\
& - \sum_{s\in [1, S]}r_i^s(t)p_i^s(t)]\\
\leq&  B + \sum_{i\in [1, F]} [\sum_{s\in [1, S]}[Q_i^s(t) r_i^s(t) + Z_i^s(t) \epsilon_s]- V\sum_{s\in [1, S]}r_i^s(t)p_i^s(t)]\\
    &+ V \sum_{i\in [1, F]} \sum_{m\in [1, M]} [\beta_i(t) n_i^{m*}(t)]\\
    &-\sum_{i\in [1, F]}\sum_{s=\in [1, S]}\sum_{j\in [1, F]}\mu_{ij}^{s*}(t) [Q_i^s(t) + Z_i^s(t)]\\
    &-\sum_{i\in [1, F]}\sum_{s\in [1, S]}D_i^{s*}(t)[Q_i^s(t)+Z_i^s(t)-V\cdot \xi_i^s]\\
=& B + V\cdot \sum_{i\in [1, F]}[\sum_{m\in [1, M]} [\beta_i(t) n_i^{m*}(t)] +\sum_{s\in [1, S]}D_i^{s*}(t) \xi_i^s\\
    &- \sum_{s\in [1, S]}r_i^s(t)p_i^s(t)]\\
    & +\sum_{i\in [1, F]} \sum_{s\in [1, S]}[Q_i^s(t)[r_i^s(t)-\sum_{j=1}^F \mu_{ij}^{s*}(t)-D_i^{s*}(t)]\\ &+ Z_i^s(t)[\epsilon_s-\sum_{j=1}^F \mu_{ij}^{s*}(t)-D_i^{s*}(t)]].
\end{align*}}\vspace*{-5mm}

Taking conditional expectations over queue status $\Theta(t)$ on both sides the inequality, we have that

\vspace*{-4mm}{\small\begin{align*}
&\mathbb{E}\{\Delta(\Theta(t))|\Theta(t)\}+V\cdot \sum_{i\in [1, F]}\mathbb{E}\{[\sum_{m\in [1, M]} \beta_i(t) n_i^m(t)\\
    &+ \sum_{s\in [1, S]}D_i^s(t) \xi_i^s - \sum_{s\in [1, S]}r_i^s(t)p_i^s(t)]|\Theta(t) \} \\
\leq& B + V\cdot \sum_{i\in [1, F]}\mathbb{E}\{[\sum_{m\in [1, M]} \beta_i(t) n_i^{m*}(t) + \sum_{s\in [1, S]}D_i^{s*}(t) \xi_i^s\\
    &- \sum_{s\in [1, S]}r_i^s(t)p_i^s(t)]|\Theta(t)\}\\
    & +\sum_{i\in [1, F]}\sum_{s\in [1, S]}[Q_i^s(t)\mathbb{E}\{[r_i^s(t)-\sum_{j=1}^F \mu_{ij}^{s*}(t)-D_i^{s*}(t)]|\Theta(t)\}\\ &+ Z_i^s(t)\mathbb{E}\{[\epsilon_s-\sum_{j=1}^F \mu_{ij}^{s*}(t)-D_i^{s*}(t)]|\Theta(t)\}].
\end{align*}}\vspace*{-4mm}

Next, we take expectations on both sides the inequality, we have that

\vspace*{-5mm}{\small\begin{align*}
&\mathbb{E}\{\Delta(\Theta(t))\}+V\cdot \sum_{i\in [1, F]}\mathbb{E}\{[\sum_{m\in [1, M]} \beta_i(t) n_i^m(t) + \sum_{s\in [1, S]}D_i^s(t) \xi_i^s\\
    &- \sum_{s\in [1, S]}r_i^s(t)p_i^s(t)] \} \\
\leq& B + V\cdot \sum_{i\in [1, F]}\mathbb{E}\{[\sum_{m\in [1, M]} \beta_i(t) n_i^{m*}(t) + \sum_{s\in [1, S]}D_i^{s*}(t) \xi_i^s \\
    &- \sum_{s\in [1, S]}r_i^s(t)p_i^s(t)]\} \\
    & +\sum_{i\in [1, F]}\sum_{s\in [1, S]}[Q_i^s(t)\mathbb{E}\{[r_i^s(t)-\sum_{j=1}^F \mu_{ij}^{s*}(t)-D_i^{s*}(t)]\}\\ &+ Z_i^s(t)\mathbb{E}\{[\epsilon_s-\sum_{j=1}^F \mu_{ij}^{s*}(t)-D_i^{s*}(t)]\}].
\end{align*}}\vspace*{-5mm}

By summing over the $T$ slots on both sides of the inequality, we have that

\vspace*{-5mm}{\small\begin{align*}
&\mathbb{E}\{L(\Theta(T))\}-\mathbb{E}\{L(\Theta(0))\}-V\cdot \sum_{t=0}^{T-1}\sum_{i\in [1, F]}\mathbb{E}\{[\sum_{m\in [1, M]} \beta_i(t) n_i^m(t)\\
    & + \sum_{s\in [1, S]}D_i^s(t) \xi_i^s - \sum_{s\in [1, S]}r_i^s(t)p_i^s(t)] \} \\
\leq& T\cdot B + V\cdot \sum_{t=0}^{T-1} \sum_{i\in [1, F]}\mathbb{E}\{[\sum_{m\in [1, M]} \beta_i(t) n_i^{m*}(t) + \sum_{s\in [1, S]}D_i^{s*}(t) \xi_i^s\\
    &- \sum_{s\in [1, S]}r_i^s(t)p_i^s(t)]\} \\
    & +\sum_{t=0}^{T-1} \sum_{i\in [1, F]}\sum_{s\in [1, S]}[Q_i^s(t)\mathbb{E}\{[r_i^s(t)-\sum_{j=1}^F \mu_{ij}^{s*}(t)-D_i^{s*}(t)]\}\\ &+ Z_i^s(t)\mathbb{E}\{[\epsilon_s-\sum_{j=1}^F \mu_{ij}^{s*}(t)-D_i^{s*}(t)]\}].
\end{align*}}\vspace*{-4mm}

Since $\mathbb{E}\{L(\Theta(T))\}\geq 0$ and $\mathbb{E}\{L(\Theta(0))\}=0$ according to the definition of the Lyapunov function, we have that

\vspace*{-5mm}{\small\begin{align*}
&-V\cdot \sum_{t=0}^{T-1}\sum_{i\in [1, F]}\mathbb{E}\{[\sum_{m\in [1, M]} \beta_i(t) n_i^m(t) + \sum_{s\in [1, S]}D_i^s(t) \xi_i^s\\
    & - \sum_{s\in [1, S]}r_i^s(t)p_i^s(t)] \} \\
\leq& T\cdot B + V\cdot \sum_{t=0}^{T-1} \sum_{i\in [1, F]}\mathbb{E}\{[\sum_{m\in [1, M]} \beta_i(t) n_i^{m*}(t) + \sum_{s\in [1, S]}D_i^{s*}(t) \xi_i^s\\
    &- \sum_{s\in [1, S]}r_i^s(t)p_i^s(t)]\} \\
    & +\sum_{t=0}^{T-1} \sum_{i\in [1, F]}\sum_{s\in [1, S]}[Q_i^s(t)\mathbb{E}\{[r_i^s(t)-\sum_{j=1}^F \mu_{ij}^{s*}(t)-D_i^{s*}(t)]\}\\ &+ Z_i^s(t)\mathbb{E}\{[\epsilon_s-\sum_{j=1}^F \mu_{ij}^{s*}(t)-D_i^{s*}(t)]\}].
\end{align*}}\vspace*{-4mm}

Dividing $T\cdot V$ on both sides and taking limitation on $T$ to infinity, we have that

\vspace*{-5mm}{\small\begin{align*}
-\Pi \leq& B/V - \Pi^* \\
    & +\sum_{i\in [1, F]}\sum_{s\in [1, S]}[\bar{Q}_i^s[\bar{r}_i^s-\sum_{j=1}^F \bar{\mu}_{ij}^{s*}-\bar{D}_i^{s*}]\\ &+ \bar{Z}_i^s[\epsilon_s-\sum_{j=1}^F \bar{\mu}_{ij}^{s*}-\bar{D}_i^{s*}]]\\
\leq& B/V - \Pi^*.
\end{align*}}\vspace*{-5mm}

The second inequality is based on the fact that $\bar{r}_i^s-\sum_{j=1}^F \bar{\mu}_{ij}^{s*}-\bar{D}_i^{s*}\leq 0$ and $\epsilon_s-\sum_{j=1}^F \bar{\mu}_{ij}^{s*}-\bar{D}_i^{s*}\leq 0$. Rearranging the two sides, we have that

\vspace{-4mm}{\small
\begin{align*}
\Pi \geq \Pi^* - B/V,
\end{align*}}\vspace{-6mm}

\section{Proof to Theorem \ref{theorem:welfare}}\label{appendix:welfare}


We have shown that, Algorithm \ref{alg:social} achieves a social welfare with a constant gap to the offline optimum, by minimizing the RHS of \emph{drift-plus-penalty} inequality in Eqn.~(\ref{eqn:drift-plus-penalty2}) according to the Lyapunov optimization theory \cite{book2010}. Hence, if we can prove that, Algorithm \ref{alg:profit} can also minimize the RHS of Eqn.~(\ref{eqn:drift-plus-penalty2}), \emph{i.e.}, maximizing problem (\ref{eqn:varphi1}) and (\ref{eqn:varphi2}), we can also prove its social welfare optimality. Our intuition of the proof is that, when the number of clouds in the federation grows to infinity, the gap to the minimum of RHS of Eqn.~(\ref{eqn:drift-plus-penalty2}), \emph{i.e.}, the gap to
the maximum of problem (\ref{eqn:varphi1}) and (\ref{eqn:varphi2}), by Algorithm \ref{alg:profit} is infinitely close to zero.

\noindent -- \emph{Gap to the minimum of RHS of Eqn.~(\ref{eqn:drift-plus-penalty2}) with Algorithm \ref{alg:profit}}:

As discussed above, problem (\ref{eqn:varphi2}) is only controlled by the job drop decisions, \emph{i.e.}, $D_i^s(t)$. Since Algorithm \ref{alg:profit} and Algorithm \ref{alg:social} have the same decision on job dropping as in Eqn.~(\ref{eqn:drop}) and Eqn.~(\ref{eqn:drop2}),
the maximum of problem (\ref{eqn:varphi2}) is also achieved by Algorithm \ref{alg:profit}. Hence,
the gap to the minimum of RHS of Eqn.~(\ref{eqn:drift-plus-penalty2}), by Algorithm \ref{alg:profit} only depends on its gap to
the maximum of problem (\ref{eqn:varphi1}), which is determined by the job scheduling and server provisioning decisions.

We first map the job scheduling and server provisioning decisions in Algorithm \ref{alg:social} to an equivalent VM allocation
based on an idealized double auction scenario. Let each cloud still proposes its buy-bid and sell-bid based on
the true valuations given in Eqn.~(\ref{eqn:true-buy}), (\ref{eqn:true-sell}), (\ref{eqn:true-vol-b}) and (\ref{eqn:true-vol-s}).

With Eqn.~(\ref{eqn:true-buy}), we have that the price of buy-bid for VM type $m$ at cloud $i$ is $\frac{1}{V}$ of the maximum
weight among all jobs at this cloud demanding type-$m$ VMs. With the winner determination of our double auction mechanism,
$\theta_1^m(t)$ is the maximum price of buy-bids from all clouds for type-$m$ VMs. Hence, the cloud with buy-bid $\theta_1^m(t)$
has the maximum weight among all jobs demanding type-$m$ VMs at all clouds.

According to the definition of $<\acute{i}_m,\acute{s}_m>$ in Eqn.~(\ref{eqn:weight2}), we know that jobs of service type $\acute{s}_m$
at cloud $\acute{i}_m$ has the maximum weight for VM type $m$ over all service types at each cloud demanding for the same VMs.
Hence, cloud $\acute{i}_m$ proposes the maximum buy-bid price $\theta_1^m(t)$, which is $\frac{1}{V}$ of the weight of type-$\acute{s}_m$
jobs at cloud $\acute{i}_m$.

With Eqn.~(\ref{eqn:true-sell}), we have that the price of sell-bid for VM type $m$ at cloud $i$ is the larger one between
i) $\frac{1}{V}$ of the maximum weight among all jobs at this cloud demanding type-$m$ VMs; and ii) the per-VM operational
price $\beta_{i}(t)/C_{i}^m$ at cloud $i$. With the winner determination of our double auction mechanism,
$\vartheta_j^m(t)$ is the $j^{th}$ lowest price of sell-bids from all clouds for type-$m$ VMs.

If the maximum buy-bid price $\theta_1^m(t)$ is larger than the $j^{th}$ lowest sell-bid price $\vartheta_j^m(t)$,
we have that i) with an idealized double auction, cloud $\acute{i}_m$ has a higher buy-bid price than the sell-bid such that it can
buy all VMs of type $m$ from the cloud proposing $\vartheta_j^m(t)$; ii) with the job scheduling decision in Eqn.~(\ref{eqn:schedule3}),
all VMs of type $m$ are allocated for job scheduling at cloud $\acute{i}_m$ for service type $\acute{s}_m$. Hence, the job scheduling
decision in Algorithm \ref{alg:social} is equivalent to the idealized double auction that, the bidder with highest buy-bid price
can buy all VMs from those sellers with a lower sell-bid price. Fig.~\ref{fig:demand-price} gives an illustration for the case. In
Fig.~\ref{fig:demand-price}, the buy-bids are sorted in descending order while the sell-bids are sorted in ascending order. Let
there are $k$ sell-bids with lower price than that of the highest buy-bid $\theta_1^m(t)$. Here, $k$ is the maximum number
of sell-bids, whose prices $\vartheta_j^m(t)$ ($j\in [1,k]$) are lower than that of the highest buy-bid $\theta_1^m(t)$, \emph{i.e.},

\vspace{-4mm}{\small
\begin{align*}
k=\text{arg}\max_{i\in [1, F]}\{\vartheta_j^m(t)< \theta_1^m(t)|\forall j\in [1,i]\}.
\end{align*}}\vspace{-4mm}

\noindent We know that all the $k$ sellers
sell all VMs of type $m$ to the cloud with $\theta_1^m(t)$, \emph{i.e.}, cloud $\acute{i}_m$.

However, the idealized cloud cannot give truthfulness guarantee. With our double auction mechanism, in Fig.~\ref{fig:demand-price},
only the $(j'-1)$ sellers, whose sell-bid prices are no larger than the second highest buy-bid price $\theta_2^m(t)$, will
sell their VMs to cloud $\acute{i}_m$. Here, $j'$ is the maximum number of sellers with sell-bid prices no larger than
$\theta_2^m(t)$, as defined in Eqn.~(\ref{eqn:j'}). Hence, the gap to the maximum of problem (\ref{eqn:varphi1}) by Algorithm
\ref{alg:profit} is determined by the VMs that are not sold to cloud $\acute{i}_m$ by sellers between $j'$ and $k$.

If these VMs of type $m$ are allocated to cloud $\acute{i}_m$ with Algorithm \ref{alg:social}, their utility gain in problem (\ref{eqn:varphi1})
is that

\vspace{-4mm}{\small
\begin{align*}
&\sum_{j=[j',k]}[\frac{Q_{\acute{i}_m}^{\acute{s}_m}(t)+Z_{\acute{i}_m}^{\acute{s}_m}(t)}{g_{\acute{s}_m}}-V\beta_j /C_j^{m}]\cdot C_j^{m} N_j^{m}\\
=&\sum_{j=[j',k]}[\theta_1^m(t)-\beta_j /C_j^{m}]\cdot V \cdot L_j^m(t).
\end{align*}}\vspace{-4mm}

\noindent Here, the server provisioning decisions, \emph{i.e.}, $n_i^m(t)$, are replaced by job scheduling decisions, \emph{i.e.},
$\mu_{ij}^{s}(t)$, according to Eqn.~(\ref{eqn:server-provision2}) in Algorithm \ref{alg:social}.

If these VMs of type $m$ of sellers $j\in [j',k]$ are not traded to cloud $\acute{i}_m$ based on our double auction mechanism, these VMs are either
scheduled to serve the jobs with maximum weight at its own cloud, if the maximum weight is higher than the per-VM operational
price, or inactivated otherwise, according to the job scheduling decision in Eqn.~(\ref{eqn:schedule1}) of Algorithm \ref{alg:profit}.
Hence, the utility gain of these VMs for problem (\ref{eqn:varphi1}) with Algorithm \ref{alg:profit} is that

\vspace{-4mm}{\small
\begin{align*}
\sum_{j=[j',k]}[\vartheta_j^m(t)-\beta_j /C_j^{m}]\cdot V \cdot L_j^m(t).
\end{align*}}\vspace{-4mm}

So, the gap to the maximum of problem (\ref{eqn:varphi1}) by Algorithm \ref{alg:profit} is the difference between the
above utility gains, as follows,

\vspace{-4mm}{\small
\begin{align}
V \cdot \sum_{j=[j',k]}[\theta_1^m(t)-\vartheta_j^m(t)]\cdot L_j^m(t),\label{eqn:gap}
\end{align}}\vspace{-4mm}

\noindent which is equivalent to the size of the shadow area in Fig.~\ref{fig:demand-price}, multiplied by $V$.

Since $\vartheta_j^m(t) > \theta_2^m(t)$ for each $j\in [j'+1, k]$, we can give an upperbound to the gap as follows,

\vspace{-4mm}{\small
\begin{align}
&V \cdot \sum_{j=[j',k]}[\theta_1^m(t)-\vartheta_j^m(t)]\cdot L_j^m(t),\notag\\
=&V \cdot [[\theta_1^m(t)-\vartheta_{j'}^m(t)]\cdot L_{j'}^m(t) + \sum_{j=[j'+1,k]}[\theta_1^m(t)-\vartheta_j^m(t)]\cdot L_j^m(t)],\notag\\
\leq& V \cdot [[\theta_1^m(t)-\theta_2^m(t)+\vartheta_{j'+1}^m(t)-\vartheta_{j'}^m(t)]\cdot L_{j'}^m(t)\notag \\
    &+ \sum_{j=[j'+1,k]}[\theta_1^m(t)-\theta_2^m(t)]\cdot L_j^m(t)]\notag\\
=&V \cdot [[\vartheta_{j'+1}^m(t)-\vartheta_{j'}^m(t)]\cdot L_{j'}^m(t) + \sum_{j=[j',k]}[\theta_1^m(t)-\theta_2^m(t)]\cdot L_j^m(t)].\label{eqn:gap2}
\end{align}}\vspace{-4mm}

Since the system is homogenous with the same number of servers for each VM type $m$ at each cloud, the value of $L_{j}^m(t)$ at
each cloud $j\in [1, F]$ is also the same. We use $L^m$ to denote that value. Hence, the gap in Eqn.~(\ref{eqn:gap2}) can be
rewritten as

\vspace{-4mm}{\small
\begin{align}
V \cdot L^m [[\vartheta_{j'+1}^m(t)-\vartheta_{j'}^m(t)]+ [k-j'][\theta_1^m(t)-\theta_2^m(t)]].\label{eqn:gap3}
\end{align}}\vspace{-4mm}

\vspace{-4mm}
\begin{figure}[H]
  \centering
  \includegraphics[width=0.8\columnwidth]{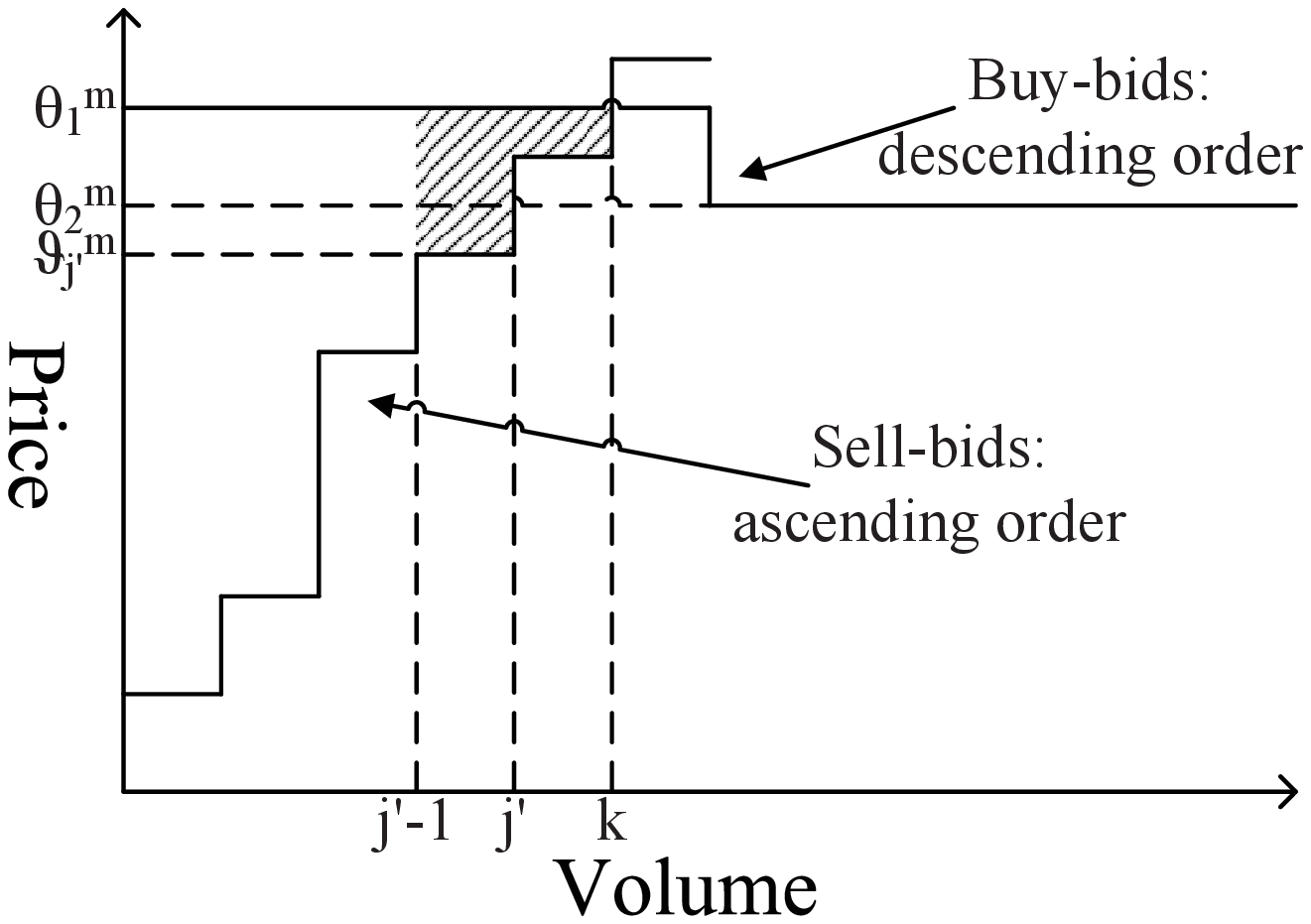}\\\vspace{-4mm}
  \caption{Illustration of the auction.}\label{fig:demand-price}\vspace{-4mm}
\end{figure}

We next show that, when the number of clouds scales to infinity, the gap in Eqn.~(\ref{eqn:gap3}) is infinitely close to zero,
by analyzing the distribution of bidding prices based on the cloud number $F$.

\noindent -- \emph{Distribution of bidding prices}:

The distribution of the bidding prices should be analyzed in order to find an analytical result between the size of the
gap in Eqn.~(\ref{eqn:gap3}) and the number of clouds $F$.

Let $\Theta_i(t)$ be the status of cloud $i$ at time slot $t$. We see that the status $\Theta_i(t)$ is a Markov chain with Algorithm \ref{alg:profit}.
Define $\mathcal{S}=\{\rho:\text{Pr}(\Theta_i(t)=\rho|\Theta_i(0)=0) \text{ for some }$t$\}$, then
$\Theta_i(t)$ is an irreducible Markov chain on state space $\mathcal{S}$ with $\Theta_i(0)=0$. This claim is true because
i) any state in $\mathcal{S}$ is reachable from $0$ and ii) since $\text{Pr}(r_i^s(t)=0)>0,\ \forall s\in [1, S]$, the Markov
chain can move from $\Theta_i(t)$ to $0$ in finite time with a positive probability. ($Q_i^s(t)$ can be cleared by job scheduling or dropping while
virtual queue $Z_i^s(t)$ can also be cleared after $Q_i^s(t)$ is zero for a constant time, according to the queueing laws
Eqn.~(\ref{eqn:queue1}) and (\ref{eqn:queue2})) Based on the same reason as above, the state $0$ is an aperiodic state.
With Lemma \ref{lemma:bounded-queue} we know that $Q_i^s(t)$ and $Z_i^s(t)$ have finite upperbounds, we can then have that
the state space $\mathcal{S}$ is also finite. In conclusion, the Markov chain is irreducible with an aperiodic state and
finite state space. Hence, the Markov chain is ergodic.

We know that the prices of buy-bid and sell-bid for each VM type $m$ at cloud $i$ are calculated with its current status $\Theta_i(t)$,
as well as the current operational price $\beta_i(t)$ (only for sell-bid), according Alg.~\ref{alg:profit}. Since the cloud's
status follows an ergodic process and the operational price is also ergodic according to our problem model, we know that
the prices of buy-bid and sell-bid are also ergodic. Recall that we are proving the asymptotic optimality in social welfare
by Algorithm \ref{alg:profit} under homogenous system settings. Hence, we can have that the prices of buy-bids and sell-bids
at different clouds follow the same ergodic process.

Let the price of buy-bid for VM type $m$ follow a distribution $\mathcal{E}$
with continuous density $e$ on the compact interval $[0, \bar{b}]$. Here, $\bar{b}=\frac{Q_i^{\bar{s}(max)}+Z_i^{\bar{s}(max)}}{V\cdot g_{\bar{s}}}$,
where $\bar{s}$ is the service type with maximum value of $\frac{Q_i^{s(max)}+Z_i^{s(max)}}{V\cdot g_s}$ among all types in $[1, S]$ with $s_m=m$.
Let the price of sell-bid for VM type $m$ follow a distribution $\mathcal{H}$
with continuous density $h$ on the compact interval $[\b{b}, \bar{b}]$. Here, $\b{b}=\beta_i^{(min)}/C_i^{m}$. Denote the maximum and minimum of
$e$ and $h$ as follows,

\vspace{-4mm}{\small
\begin{align*}
&e^{min}=\min_{x\in [0, \bar{b}]} e(x)>0,\ e^{max}=\max_{x\in [0, \bar{b}]} e(x)>0,\\
&h^{min}=\min_{x\in [\b{b}, \bar{b}]} h(x)>0,\ h^{max}=\max_{x\in [\b{b}, \bar{b}]} h(x)>0.
\end{align*}}\vspace{-4mm}

Then, we can have the following lemma according to Lemma 1 in \cite{huang-CI02},


\begin{lemma}\label{lemma:gap-bound}

\vspace{-0mm}{\small
\begin{align*}
\frac{1}{e^{max}(F+1)}\leq & \mathbb{E}\{\theta_j-\theta_{j+1}\} \leq \frac{1}{e^{min}(F+1)},\ \forall j\in [1, F-1],\\
\frac{1}{h^{max}(F+1)}\leq & \mathbb{E}\{\vartheta_{j+1}-\vartheta_{j}\} \leq \frac{1}{h^{min}(F+1)},\ \forall j\in [1, F-1].
\end{align*}}\vspace{-4mm}

\noindent Here, $\mathbb{E}\{\cdot\}$ denotes the expectation.

\end{lemma}

\noindent -- \emph{Asymptotic Optimality in Minimizing the RHS of Eqn.~(\ref{eqn:drift-plus-penalty2}) with Algorithm \ref{alg:profit}}:

With Lemma \ref{lemma:gap-bound}, we can further bound the gap in Eqn.~(\ref{eqn:gap3}) to the maximum of problem (\ref{eqn:varphi1}) by
Algorithm \ref{alg:profit} as follows,

\vspace{-4mm}{\small
\begin{align*}
&V \cdot L^m [[\vartheta_{j'+1}^m(t)-\vartheta_{j'}^m(t)]+ [k-j'][\theta_1^m(t)-\theta_2^m(t)]]\\
\leq &V \cdot L^m [\frac{1}{h^{min}(F+1)}+ \frac{k-j'}{e^{min}(F+1)}]
\end{align*}}\vspace{-4mm}

Since each seller between $j'$ and $k$ has a sell-bid price between $\theta_1^m(t)$ and $\theta_2^m(t)$, the sell-bid prices
of all these sellers resides in an interval which has an expectation no larger than $\frac{1}{e^{min}(F+1)}$. On the other hand, since the expected interval between each sell-bid is no smaller than $\frac{1}{h^{max}(F+1)}$, we know that the expected interval between
the $j'$ seller and the $k$ seller should be no smaller than $\frac{k-j'}{h^{max}(F+1)}$, which should be still no larger than $\frac{1}{e^{min}(F+1)}$.
Hence, we have that

\vspace{-4mm}{\small
\begin{align*}
&\frac{k-j'}{h^{max}(F+1)}\leq \frac{1}{e^{min}(F+1)}\\
\Rightarrow & k-j' \leq \frac{h^{max}}{e^{min}}.
\end{align*}}\vspace{-4mm}

Finally, we can bound the gap to the maximum of problem (\ref{eqn:varphi1}) by
Algorithm \ref{alg:profit} as follows,

\vspace{-4mm}{\small
\begin{align*}
V \cdot L^m [\frac{1}{h^{min}(F+1)}+ \frac{h^{max}/e^{min}}{e^{min}(F+1)}].
\end{align*}}\vspace{-4mm}

It is clear that this gap is infinitely close to zero when the number of cloud in the federation $F$ grows to infinity.

\noindent -- \emph{Asymptotic Optimality in Social Welfare of Algorithm \ref{alg:profit}}:

When the number of clouds in the federation scales to infinite large, \emph{i.e.}, $F\rightarrow \infty$, the RHS of Eqn.~(\ref{eqn:drift-plus-penalty2}) is minimized in each time slot $t$ with our Algorithm \ref{alg:profit} and our double auction mechanism. Thus, following the same steps, as above, to prove the social welfare optimality of Algorithm \ref{alg:social}, which also minimizes the RHS of Eqn.~(\ref{eqn:drift-plus-penalty2}) in each time slot $t$, we can prove that the time-averaged social welfare achieved with our Algorithm \ref{alg:profit} is within a constant gap $B/V$ from the offline optimum $\Pi^*$, when $F\rightarrow \infty$, \emph{i.e.},

\vspace{-4mm}{\small
\begin{align*}
\Pi \geq \Pi^* - B/V,
\end{align*}}\vspace{-6mm}

\end{appendices} 
}

\end{document}